\documentclass[11pt]{article}
\usepackage{fullpage}
\usepackage{citesort}
\usepackage{times}
\usepackage{amssymb, amsmath}
\usepackage{color}
\usepackage{graphicx}
\usepackage{verbatim} 
\usepackage{url} 
\usepackage[ruled,vlined]{algorithm2e}
\usepackage{balance}
\allowdisplaybreaks

\newtheorem{definition}{Definition}
\newtheorem{theorem}{Theorem}
\newtheorem{lemma}{Lemma}
\newtheorem{remark}{Remark}
\newtheorem{observation}{Observation}
\newtheorem{corollary}{Corollary}

\newtheorem{property}{Property}
\newtheorem{proposition}{Proposition}

\newcommand{\eps}{\varepsilon}

\newcommand{\E}{\mathsf{E}}
\newcommand{\ED}{\mathsf{ED}}
\newcommand{\TV}{\mathsf{TV}}

\newcommand{\C}{\mathcal{C}}
\newcommand{\var}{\mathsf{Var}}
\renewcommand{\Pr}{\mathsf{Pr}}

\newcommand{\CC}{\mathsf{CC}}

\newcommand{\abs}[1]{\left| #1 \right|}

\newcommand{\gm}{{$k$-GAP-MAJ}}
\newcommand{\bit}{{$k$-APPROX-SUM}}
\newcommand{\bitG}{{$k$-APPROX-SUM$_{f,\tau}$}}
\newcommand{\bitD}{{$k$-APPROX-SUM$_{\text{$2$-DISJ},\tau}$}}
\newcommand{\gap}{{$2$-GAP-ORT}}
\newcommand{\DISJ}{{$k$-DISJ}}

\newcommand{\TWODISJ}{{$2$-DISJ}}
\newcommand{\XOR}{{$k$-XOR}}

\newcommand{\thresh}{{$k$-BTX}}
\newcommand{\quan}{{QUAN}}
\newcommand{\entropy}{{ENTROPY}}
\newcommand{\GHD}{{GHD}}
\newcommand{\guess}{{$k$-GUESS}}

\renewcommand{\paragraph}[1]{\medskip\noindent{\bf #1} }

\newcommand{\qinomit}[1]{}

\newcommand{\ifconf}[1]{}
\newcommand{\iffull}[1]{#1}

\newcommand{\remove}[1]{}

\newcommand{\poly}{{\mathrm{poly}}}

\newtheorem{fact}{Fact}

\newcommand{\qedsymb}{\hfill{\rule{2mm}{2mm}}}
\def\FullBox{\hbox{\vrule width 8pt height 8pt depth 0pt}}
\def\qed{\ifmmode\qquad\FullBox\else{\unskip\nobreak\hfil
\penalty50\hskip1em\null\nobreak\hfil\FullBox
\parfillskip=0pt\finalhyphendemerits=0\endgraf}\fi}

\newenvironment{proof}{\begin{trivlist}
\item[\hspace{\labelsep}{\bf\noindent Proof: }]
}{\qedsymb\end{trivlist}}

\begin{document}
\title{Tight Bounds for Distributed Functional Monitoring}

\author{
David P. Woodruff\\IBM Almaden\\dpwoodru@us.ibm.com 
\and Qin Zhang\thanks{Most of this work was done while Qin Zhang was a postdoc in MADALGO (Center for Massive Data Algorithmics - a Center of the Danish National Research Foundation), Aarhus University.}\\IBM  Almaden\\qinzhang@cse.ust.hk}

\date{}

\maketitle
\begin{abstract}
We resolve several fundamental questions in the area of distributed functional monitoring,
initiated by Cormode, Muthukrishnan, and Yi (SODA, 2008), and receiving 
recent attention. In this model there are $k$ sites each
tracking their input streams and communicating with a central coordinator. The coordinator's
task is to continuously maintain an approximate output to a function 
computed over the union of the $k$ streams. 
The goal is to minimize the number of bits communicated.

Let the $p$-th frequency moment be defined as $F_p = \sum_i f_i^p$, where $f_i$ is the frequency of element $i$. 
We show the randomized communication complexity of estimating 
the number of distinct elements (that is, $F_0$) up to a $1+\eps$ factor is $\tilde{\Omega}(k/\eps^2)$, 
improving upon the previous $\Omega(k + 1/\eps^2)$ bound and matching known upper bounds up to a logarithmic factor. 
For $F_p$, $p > 1$, we improve
the previous $\Omega(k + 1/\eps^2)$ bits communication bound to ${\Omega}(k^{p-1}/\eps^2)$. 
We obtain similar improvements for heavy hitters, empirical entropy, and other problems. 
Our lower bounds
are the first of any kind in distributed functional monitoring to depend on the product
of $k$ and $1/\eps^2$. Moreover, the lower bounds are for the static version of the distributed functional 
monitoring model  where the coordinator only needs to compute the function at the time when all $k$ input 
streams end; surprisingly they almost match what is achievable in the (dynamic version of) distributed 
functional monitoring model where the coordinator needs to keep track of the function continuously at any time step.
We also show that we can estimate $F_p$, for any $p > 1$, using $\tilde{O}(k^{p-1}\poly(\eps^{-1}))$ bits of
communication. This drastically improves upon the previous $\tilde{O}(k^{2p+1}N^{1-2/p} \poly(\eps^{-1}))$ bits
bound of Cormode, Muthukrishnan, and Yi for general $p$, and their
$\tilde{O}(k^2/\eps + k^{1.5}/\eps^3)$ bits bound for $p = 2$. For $p = 2$, our bound resolves
their main open question. 

Our lower bounds are based on new direct sum theorems for approximate majority, and 
yield improvements to classical problems in the standard data stream
model. First, we improve the known lower bound for estimating $F_p, p > 2,$ in $t$ passes from
$\tilde{\Omega}(n^{1-2/p}/(\eps^{2/p} t))$ to ${\Omega}(n^{1-2/p}/(\eps^{4/p} t))$, giving the first bound
that matches what we expect when $p = 2$ for any constant number of passes. 
Second, we give the first lower bound for estimating
$F_0$ in $t$ passes with $\Omega(1/(\eps^2 t))$ bits of space that does not use the hardness
of the gap-hamming problem. 
\end{abstract}


\section{Introduction}
Recent applications in sensor networks and distributed systems have motivated the {\it distributed functional monitoring} model, initiated by Cormode, Muthukrishnan, and Yi \cite{CMY11}. In this model there are $k$ {\it sites} and a single central {\it coordinator}. 
Each site $S_i \ (i \in [k])$ receives a stream of data $A_i(t)$ for timesteps $t = 1, 2, \ldots$, and the coordinator wants to keep track of a function $f$ that is defined over the multiset union of the $k$ data streams at each time $t$. For example, the function $f$ could be the number of distinct elements in the union of the $k$ streams. We assume that there is a two-way communication channel between each site and the coordinator so that the sites can communicate with the coordinator. The goal is to minimize the total amount of communication between the sites and the coordinator so that the coordinator can approximately maintain $f(A_1(t), \ldots, A_k(t))$ at any time $t$. Minimizing the total communication is motivated by power constraints in sensor networks, since communication typically uses a power-hungry radio
\cite{eghk99}; and also by network bandwidth constraints in distributed systems.
There is a large body of work on monitoring problems in this model, including maintaining a random sample \cite{CMYZ10,TW11}, estimating frequency moments \cite{CG05,CMY11}, finding the heavy hitters \cite{BO03,KCR06,MSDO05,YZ09},
approximating the quantiles \cite{CGMR05,YZ09,HYZ11}, and estimating the entropy \cite{ABC09}. 

We can think of the distributed functional monitoring model as follows. Each of the $k$ sites holds an $N$-dimentional vector where $N$ is the size of the universe. An update to 
a coordinate $j$ on site $S_i$ causes $v^i_j$ to increase by $1$. The goal is to estimate a statistic of
$v = \sum_{i=1}^k v^i$, such as the $p$-th frequency moment $F_p = \|v\|_p^p$, the number of distinct elements
$F_0 = |\textrm{support}(v)|$, and the empirical entropy $H = \sum_i \frac{v_i}{\|v\|_1} \log \frac{\|v\|_1}{v_i}$. 
This is the standard {\it insertion-only} model. For many of these problems, with the 
exception of the empirical entropy, 
there are strong lower bounds (e.g., $\Omega(N)$) if
allowing updates to coordinates that cause $v^i_j$ to decrease \cite{ABC09}. The latter
is called the {\it update} model. 
Thus, except for entropy, we follow previous work and consider the insertion-only model.

To prove lower bounds, we consider the static version of the distributed functional monitoring model, where the coordinator only needs to compute the function at the time when all $k$ input streams end. It is clear that a lower bound for the static case is also a lower bound for the dynamic case in which the coordinator has to keep track of the function at any point in time. The static version of the distributed functional monitoring model is closely related to the multiparty {\it number-in-hand} communication model, where we again have $k$ sites each holding an $N$-dimensional vector $v^i $, and they want to jointly compute a function defined on the $k$ input vectors. It is easy to see that these two models are essentially the same since in the former, if site $S_i$ would like to send a message to $S_j$, it can always send the message first to the coordinator and then the coordinator can forward the message to $S_j$. Doing this will only increase the total amount of communication by a factor of two. Therefore, we do not distinguish between these two models in this paper.

There are two variants of the multiparty number-in-hand communication model we will consider: the {\em blackboard model}, 
in which each message a site sends is received by all other sites, i.e., it is broadcast, 
and the {\em message-passing} model, in which each message is between the coordinator and a 
specific site. 

%

Despite the large body of work in the distributed functional monitoring model, the complexity of basic problems 
is not well understood. For example, for estimating $F_0$ up to 
a $(1+\eps)$-factor, 
the best upper bound is $\tilde{O}(k/\eps^2)$~\footnote{
We use $\tilde{O}(f)$ to denote a function of the form $f \cdot \log^{O(1)}(Nk/\eps)$.} \cite{CMY11} (all communication and information bounds in this paper, if not otherwise stated, are in terms of bits), while the only
known lower bound is $\Omega(k + 1/\eps^2)$. The dependence on $\eps$ in the lower bound is
not very insightful, as the $\Omega(1/\eps^2)$ bound follows just by considering two sites 
\cite{ABC09,CR:11}. 
The real question is whether the $k$ and $1/\eps^2$ factors should multiply. Even more 
embarrassingly, for the frequency moments $F_p$, $p > 2$, the 
known algorithms use communication $\tilde{O}(k^{2p+1} N^{1-2/p} \poly(1/\eps))$, 
while the only known lower bound is $\Omega(k + 1/\eps^2)$ \cite{ABC09,CR:11}.
Even for $p = 2$, the best known upper bound is $\tilde{O}(k^2/\eps + k^{1.5}/\eps^3)$ \cite{CMY11},
and the authors' main open question in their paper is ``It remains to close the gap
in the $F_2$ case: can a better lower bound than $\Omega(k)$ be shown, or do there exist
$\tilde{O}(k \cdot \poly(1/\eps))$ solutions?''

{\bf Our Results:}
We significantly improve the previous communication bounds for approximating the frequency moments,
entropy, heavy hitters, and quantiles in the distributed functional monitoring model. 
In many cases our bounds are optimal. 
Our results are summarized in Table \ref{tab:results}, where they are compared with previous bounds. 
\begin{table*}[t!]
\centering
\scalebox{0.93}{
\begin{tabular}{|c|cc|cc|}
\hline
& Previous work& This paper & Previous work& This paper\\
\hline
Problem & LB  &  LB  (all static) & UB & UB \\
\hline
$F_0$ & $\tilde{\Omega}(k)$ \cite{CMY11} & $\mathbf{\tilde{\Omega}(k/\eps^2)}$ & $\tilde{O}(k/\eps^2)$ \cite{CMY11} & -- \\ 
$F_2$ & $\Omega(k)$ \cite{CMY11} & $\mathbf{{\Omega}(k/\eps^2)}$ (BB) & $\tilde{O}(k^2/\eps + k^{1.5}/\eps^3)$ \cite{CMY11} & $\mathbf{\tilde{O}(\frac{k}{\mathrm{poly}(\eps)})}$ \\
$F_p\ (p > 1)$ & $\Omega(k+1/\eps^2)$ \cite{ABC09,CR:11} & $\mathbf{{\Omega}(k^{p-1}/\eps^2)}$ (BB) & $\tilde{O}(\frac{p}{\eps^{1+2/p}}k^{2p+1}N^{1-2/p})$ \cite{CMY11} & $\mathbf{\tilde{O}(\frac{k^{p-1}} {\mathrm{poly}(\eps)})}$  \\
All-quantile & $\tilde{\Omega}(\min\{\frac{\sqrt{k}}{\eps}, \frac{1}{\eps^2}\})$ \cite{HYZ11} & $\mathbf{{\Omega}(\min\{\frac{\sqrt{k}}{\eps}, \frac{1}{\eps^2}\})}$ (BB) & $\tilde{O}(\min\{\frac{\sqrt{k}}{\eps}, \frac{1}{\eps^2}\})$ \cite{HYZ11} & -- \\
Heavy Hitters & $\tilde{\Omega}(\min\{\frac{\sqrt{k}}{\eps}, \frac{1}{\eps^2}\})$ \cite{HYZ11} & $\mathbf{{\Omega}(\min\{\frac{\sqrt{k}}{\eps}, \frac{1}{\eps^2}\})}$ (BB) & $\tilde{O}(\min\{\frac{\sqrt{k}}{\eps}, \frac{1}{\eps^2}\})$ \cite{HYZ11} & -- \\
Entropy & $\tilde{\Omega}(1/\sqrt{\eps})$ \cite{ABC09} & $\mathbf{{\Omega}(k/\eps^2)}$ (BB) & $\tilde{O}(\frac{k}{\eps^3})$ \cite{ABC09}, $\tilde{O}(\frac{k}{\eps^2})$  (static) \cite{HNO08} & -- \\
$\ell_p\ (p \in (0,2])$ & -- & $\mathbf{{\Omega}(k/\eps^2)}$ (BB) & $\tilde{O}(k/\eps^2)$ (static) \cite{knpw11} & -- \\
\hline
\end{tabular}}
\caption{UB denotes upper bound; LB denotes lower bound; BB denotes blackboard model. $N$ denotes the universe size. All bounds are for randomized algorithms. We assume all bounds hold in the dynamic setting by default, and will state explicitly if they hold in the static setting. For lower bounds we assume the message-passing model by default, and state explicitly if they also hold in the blackboard model.} 
\label{tab:results}
\end{table*}
We have three main results, each introducing a new technique:
\begin{enumerate}
\item We show that for estimating $F_0$ in the message-passing model, 
$\tilde{\Omega}(k/\eps^2)$ communication is required, matching
an upper bound of \cite{CMY11} up to a polylogarithmic factor. 
Our lower bound holds in the static model in which the $k$ sites just need to approximate $F_0$ once on their inputs. 
\item We show that we can estimate $F_p$, for any $p > 1$, using $\tilde{O}(k^{p-1}\poly(\eps^{-1}))$
communication in the message-passing model\footnote{We assume the total number of updates is $\poly(N)$.}. This drastically improves upon the previous bound $\tilde{O}(k^{2p+1}N^{1-2/p} \poly(\eps^{-1}))$
 of \cite{CMY11}. In particular, setting $p = 2$, we resolve the main open question of \cite{CMY11}.
\item We show ${\Omega}(k^{p-1}/\eps^2)$ communication is necessary for approximating $F_p\ (p > 1)$
in the blackboard model, significantly
improving the prior $\Omega(k + 1/\eps^2)$ bound. As with our lower bound for $F_0$, these are 
the first lower bounds
which depend on the product of $k$ and $1/\eps$. As with $F_0$, our lower bound holds in the
static model in which the sites just approximate $F_p$ once. 
\end{enumerate}
Our other results in Table \ref{tab:results} are explained in the body of the paper, and use similar
techniques. 

We would like to mention that after the conference version of our paper, our results found applications in proving a space lower bound at each site for tracking heavy hitters in the functional monitoring model \cite{HYZ12}, and a communication complexity lower bound of computing $\eps$-approximations of range spaces in $\mathbb{R}^2$ in the message-passing model~\cite{huang-communication}.

{\bf Our Techniques:}
{\it Lower Bound for $F_0$:} 
For illustration, suppose $k = 1/\eps^2$. 
There are $1/\eps^2$ sites each holding a random independent bit.
Their task is to approximate the sum of the $k$ bits up to an additive error $1/\eps$. Call this problem \bit.\footnote{In the conference version of this paper we introduced a problem called \gm, in which sites need to
decide if at least $1/(2\eps^2) + 1/\eps$ of the bits are $1$, or at most $1/(2\eps^2) - 1/\eps$ of the bits are $1$. We instead use \bit\ here since we feel it is easier to work with: This problem is stronger than \gm\ thus is easier to lower bound, and it suffices for our purpose. \gm\ will be introduced and used in Section~\ref{sec:quantile} for heavy-hitters and quantiles.}
We show any correct protocol
must reveal $\Omega(1/\eps^2)$ bits of information about the sites' inputs. We ``compose'' this
with $2$-party disjointness {(\TWODISJ)} \cite{Raz90}, in which each party has a bitstring of length $1/\eps^2$ 
and either the strings have disjoint support (the solution is $0$) or there is a single coordinate 
which is $1$ in both strings (the solution is $1$). 
Let $\tau$ be the hard distribution for {\TWODISJ}, shown to require $\Omega(1/\eps^2)$ bits of communication to
solve \cite{Raz90}. Suppose the coordinator and each site share an instance of {\TWODISJ}
in which the solution
to {\TWODISJ} is a random bit, which is the site's effective input to \bit. 
The coordinator has the same input for each of the $1/\eps^2$
instances, while the sites have an independent input drawn 
from $\tau$ conditioned on the coordinator's input and 
output bit determined by \bit. 
The inputs are chosen so that if the output of {\TWODISJ} is $1$, 
then $F_0$ increases by $1$, 
otherwise it remains the same. This is not entirely accurate, but it illustrates the main
idea. Now, the key is that by the rectangle property of $k$-party communication protocols, the $1/\eps^2$ different output bits are independent conditioned on the transcript. Thus if a protocol does not reveal $\Omega(1/\eps^2)$ bits of information about these output bits, by an anti-concentration theorem we can show that the protocol cannot succeed with large probability. Finally, since a
$(1+\eps)$-approximation to $F_0$ can decide \bit, and since any correct protocol for \bit\ 
must reveal $\Omega(1/\eps^2)$ bits of information, the protocol must solve $\Omega(1/\eps^2)$ instances of {\TWODISJ}, 
each requiring $\Omega(1/\eps^2)$ bits of
communication (otherwise the coordinator could simulate $k-1$ of the sites and obtain an $o(1/\eps^2)$-
communication protocol for \TWODISJ \ with the remaining site, contradicting the communication
lower bound for {\TWODISJ} on this distribution). We obtain an $\tilde{\Omega}(k/\eps^2)$ bound for $k \geq 1/\eps^2$ 
by using similar arguments. 
One cannot show this in the blackboard model since there
is an $\tilde{O}(k + 1/\eps^2)$ bound for $F_0$ \footnote{The idea is to first obtain a $2$-approximation.
Then, sub-sample so that there are $\Theta(1/\eps^2)$ distinct elements. Then the first party
broadcasts his distinct elements, the second party broadcasts the distinct elements 
he has that the first party does not, etc.}.

{\it Lower Bound for $F_p$:} 
Our ${\Omega}(k^{p-1}/\eps^2)$ bound for $F_p$ cannot use the above reduction since we do not know how to
turn a protocol for approximating $F_p$ into a protocol for solving the composition of
\bit\ and {\TWODISJ}. Instead, our starting point is a recent $\Omega(1/\eps^2)$ lower
bound for the $2$-party gap-hamming distance problem {\GHD} \cite{CR:11}. The parties have a
length-$1/\eps^2$ bitstring, $x$ and $y$, respectively, and they must decide if the Hamming distance 
$\Delta(x,y) > 1/(2\eps^2) + 1/\eps$ or $\Delta(x,y) < 1/(2\eps^2)-1/\eps$. A simplification by 
Sherstov \cite{Sherstov11} 
shows a related problem called \gap \ also has communication complexity of $\Omega(1/\eps^2)$ bits. Here there
are two parties, each with $1/\eps^2$-length bitstrings $x$ and $y$, and they must decide if 
$|\Delta(x,y) - 1/(2\eps^2)| > 2/\eps$ or $|\Delta(x,y) - 1/(2\eps^2)| < 1/\eps$. 
Chakrabarti et al. \cite{CKW12} showed that 
any correct protocol for \gap \ must reveal 
${\Omega}(1/\eps^2)$ bits of
information about $(x,y)$. By independence
and the chain rule, this means for ${\Omega}(1/\eps^2)$ indices $i$,
${\Omega}(1)$ bits of information is revealed about $(x_i, y_i)$ conditioned on values $(x_j, y_j)$ for $j < i$.
We now ``embed'' an independent copy of a variant of $k$-party-disjointness, the \XOR\ problem,  
on each of the $1/\eps^2$ coordinates of \gap. In this variant, there are $k$ parties each 
holding a bitstring
of length $k^p$. On all but one ``special'' randomly chosen coordinate, there is a single site
assigned to the coordinate and that site uses private randomness to choose whether the
value on the coordinate is $0$ or $1$ (with equal probability), and the remaining $k-1$ sites
have $0$ on this coordinate. On the special coordinate, with probability $1/4$ all sites
have a $0$ on this coordinate (a ``00'' instance), 
with probability $1/4$ the first $k/2$ parties have a $1$ on this 
coordinate and the remaining $k/2$ parties have a $0$ (a ``10'' instance), 
with probability $1/4$
the second $k/2$ parties have a $1$ on this coordinate and the remaining $k/2$ parties have a $0$
(a ``01'' instance), 
and with the remaining probability $1/4$ all $k$ parties have a $1$ on this coordinate
(a ``11'' instance). We show, 
via a direct sum for distributional communication complexity, that any deterministic protocol that decides which
case the special coordinate is in with probability $1/4 + {\Omega}(1)$ 
has conditional information cost ${\Omega}(k^{p-1})$. 
This implies that any protocol that can decide whether the output is in
the set $\{10, 01\}$ (the ``XOR'' of the output bits) 
with probability $1/2 + {\Omega}(1)$ has conditional information cost
${\Omega}(k^{p-1})$. We do the direct sum argument by conditioning the mutual information
on low-entropy random variables which allow us to fill in inputs on remaining coordinates
without any communication between the parties 
and without asymptotically affecting our ${\Omega}(k^{p-1})$ lower bound. 
We design a reduction so that on the $i$-th coordinate of \gap , 
the input of the first $k/2$-players of \XOR\ is determined
by the public coin (which we condition on) and the first party's input bit to \gap, and the input
of the second $k/2$-players of \XOR\ is determined by the public coin and the second party's
input bit to \gap \ . 
We show that any protocol that solves the composition of \gap \ with $1/\eps^2$ copies of
\XOR\ , a problem that we call \thresh\ , 
must reveal ${\Omega}(1)$ bits of information about the two output bits of
an ${\Omega}(1)$ fraction of the $1/\eps^2$ copies, and from our ${\Omega}(k^{p-1})$
information cost lower bound for a single copy, we can obtain an overall
${\Omega}(k^{p-1}/\eps^2)$ bound. Finally, one can show that a
$(1+\eps)$-approximation algorithm for $F_p$ can be used to solve \thresh\ .
%
%

{\it Upper Bound for $F_p$:} We illustrate the algorithm for $p = 2$ and constant $\eps$. 
Unlike \cite{CMY11}, we do not
use AMS sketches \cite{alon96:_space}. A nice property of our protocol is that it is the first
$1$-way protocol (the protocol of \cite{CMY11} is not), 
in the sense that only the sites send messages to the coordinator (the coordinator does
not send any messages). Moreover, all messages are simple: if a site receives an update
to the $j$-th coordinate, provided the frequency of coordinate $j$ in its stream exceeds a threshold, it
decides with a certain probability to send $j$ to the coordinator. Unfortunately, one can show that
this probability cannot be the same for all coordinates $j$, as otherwise the communication would be
too large. 

To determine the threshold
and probability to send an update to a coordinate $j$, 
the sites use the public coin to randomly group all coordinates $j$ into buckets
$S_{\ell}$, where $S_{\ell}$ contains a $1/2^{\ell}$ fraction of the input coordinates. For $j \in S_{\ell}$,
the threshold and probability are only a function of $\ell$. 
Inspired by work on sub-sampling \cite{indyk05:_optim_}, 
we try to estimate the number of coordinates $j$ of magnitude in the range $[2^h, 2^{h+1})$, for each $h$. Call
this class of coordinates $C_h$. 
If the contribution to $F_2$ from $C_h$ is significant, then $|C_h| \approx 2^{-2h} \cdot F_2$, and to 
estimate $|C_h|$ we only consider those $j \in C_h$ that are in $S_{\ell}$ for 
a value $\ell$ which satisfies 
$|C_h| \cdot 2^{-\ell} \approx 2^{-2h} \cdot F_2 \cdot 2^{-\ell} \approx 1$.
We do not know $F_2$ and so we also do not know $\ell$, but we can make a logarithmic number
of guesses. We note that the work \cite{indyk05:_optim_} was available to the authors of
\cite{CMY11} for several years, but adapting it to the distributed framework here is tricky in the sense
that the ``heavy hitters'' algorithm used in \cite{indyk05:_optim_} for finding elements in different 
$C_h$ needs to be implemented in a 
$k$-party communication-efficient way. 

When choosing the threshold and probability we have two competing
constraints; on the one hand these values must be chosen so that we can accurately estimate
the values $|C_h|$ from the samples. On the other hand, these values need to be chosen so that
the communication is not excessive. Balancing these
two constraints forces us to use a threshold instead of just the same probability for all coordinates in $S_{\ell}$. 
By choosing the thresholds and probabilities to be 
appropriate functions of $\ell$, we can satisfy both constraints. Other minor issues in the analysis
arise from 
the fact that different classes contribute at different times, and that the coordinator must be
correct at all times. These issues can be resolved by conditioning on a quantity related
to the protocol's correctness being accurate at a small number of selected times in the stream, and then
arguing that the quantity is non-decreasing and that this implies that it is correct at all times. 
%
\\\\
{\bf Implications for the Data Stream Model:}
In 2003, Indyk and Woodruff introduced the {\GHD} problem 
\cite{iw03}, where a $1$-round lower bound shortly followed 
\cite{woodruff04:_optim_}. Ever since, it seemed the space complexity of estimating
$F_0$ in a data stream with $t > 1$ passes hinged on whether {\GHD} 
required $\Omega(1/\eps^2)$ communication for $t$ rounds, see, e.g., Question 10 in \cite{iitk}. 
A flurry \cite{bc09,bcrvw10,CR:11,v11,Sherstov11} of recent work finally resolved the 
complexity of {\GHD}. What our lower
bound shows for $F_0$ is that this is not the only way to prove the $\Omega(1/\eps^2)$ space
bound for multiple passes for $F_0$. Indeed, we just needed to look at $\Theta(1/\eps^2)$ parties instead
of $2$ parties. Since we have an $\Omega(1/\eps^4)$ communication lower bound for $F_0$
with $\Theta(1/\eps^2)$ parties, this implies an
$\Omega((1/\eps^4)/(t/\eps^2)) = \Omega(1/(t \eps^2))$ bound for $t$-pass algorithms for approximating $F_0$.
Arguably our proof is simpler than the recent {\GHD} lower bounds. 

Our ${\Omega}(k^{p-1}/\eps^2)$ bound for $F_p$ also
improves a long line of work on the space complexity of estimating $F_p$ for $p > 2$ in a data stream.
The current best upper bound is $\tilde{O}(N^{1-2/p}\eps^{-2})$ bits of space \cite{Ganguly11}.
See Figure 1 of \cite{Ganguly11} for a list of papers which make progress on the $\eps$ and logarithmic
factors. The previous best lower bound is $\tilde{\Omega}(N^{1-2/p}\eps^{-2/p}/t)$ for $t$ passes 
\cite{BYJKS02}. By setting $k^p = \eps^2 N$,
we obtain that the total communication is at least ${\Omega}(\eps^{2-2/p}N^{1-1/p}/\eps^2)$, 
and so the implied space lower bound for $t$-pass algorithms for $F_p$
in a data stream is ${\Omega}(\eps^{-2/p}N^{1-1/p}/(tk)) = {\Omega}(N^{1-2/p}/(\eps^{4/p} t))$. 
This gives the first bound
that agrees with the tight $\tilde{\Theta}(1/\eps^2)$ bound when $p = 2$ for any constant $t$. 
After our work, 
Ganguly~\cite{Ganguly12} improved this for the special case $t = 1$. That is, for $1$-pass algorithms for estimating
$F_p$, $p > 2$, he shows a 
space lower bound of $\Omega(N^{1-2/p}/(\eps^2 \log n))$. 

{\bf Other Related Work:}
There are quite a few papers on multiparty number-in-hand communication complexity, though they are not directly relevant for the problems studied in this paper. Alon et al.~\cite{alon96:_space} and Bar-Yossef et al.~\cite{BYJKS02} studied lower bounds for multiparty set-disjointness, which has applications to $p$-th frequency moment estimation for $p > 2$ in the streaming model. Their results were further improved in \cite{Chakrabarti03,Gronemeier09,Jayram09}. Chakrabarti et al.~\cite{CCM08} studied random-partition communication lower bounds for multiparty set-disjointness and pointer jumping, which have a number of applications in the random-order data stream model. Other work includes Chakrabarti et al.~\cite{CJP08} for median selection, Magniez et al.~\cite{MMN10} and Chakrabarti et al.~\cite{CCKM10} for streaming language recognition. Very few studies have been conducted in the message-passing model. Duris and Rolim~\cite{DR98} proved several lower bounds in the message-passing model, but only for some simple boolean functions. Three related but more restrictive private-message models were studied by Gal and Gopalan~\cite{GG07}, Erg{\"u}n and Jowhari~\cite{EJ08}, and Guha and Huang~\cite{GH09}. The first two only investigated deterministic protocols and the third was tailored for the random-order data stream model. 

Recently Phillips et al. \cite{PVZ12} introduced a technique called symmetrization 
for the number-in-hand communication model. The idea is to try to find a symmetric hard distribution for the $k$ players. Then one reduces the $k$-player problem to a $2$-player problem by assigning Alice the input of a random player and Bob the inputs of the remaining $k-1$ players. The answer to the $k$-player problem gives the answer to the $2$-player problem. By symmetrization one can argue that if the communication lower bound for the resulting $2$-player problem is $L$, then the lower bound for the $k$-player problem is $\Omega(k L)$.
While symmetrization developed in \cite{PVZ12} can be used to solve some problems for which other techniques are not known, 
such as bitwise AND/OR and graph connectivity, it has several limitations. 
First, symmetrization requires a symmetric hard distribution, and for many problems (e.g., $F_p\ (p > 1)$ in this paper) this is 
not known or unlikely to exist.
Second, for many problems (e.g., $F_0$ in this paper), we need a direct-sum type of argument with certain combining functions (e.g., the majority (MAJ)), while in \cite{PVZ12}, only outputting all copies or with the combining function OR is considered. 
Third, the symmetrization technique in \cite{PVZ12} does not give information cost bounds, and so 
it is difficult to use when composing problems as is done in this paper. In this paper, we have further developed symmetrization to make it work with the combining function MAJ and the information cost. 
\\\\
{\bf Paper Outline:}
In Section~\ref{sec:F0} and Section~\ref{sec:F2-LB} we prove our lower bounds for $F_0$ and $F_p$, $p > 1$. 
The lower bounds apply to functional monitoring, but hold even in the static model.
In Section~\ref{sec:F2-UB} we show improved upper bounds for $F_p, p > 1,$ for functional monitoring. Finally, \iffull{in Section~\ref{sec:app}} we prove lower bounds for all-quantile, heavy hitters, entropy and $\ell_p$ for any $p \geq 1$ in the blackboard model. 
\ifconf{Due to space constraints, we defer this section to the full version of this paper.}

\section{Preliminaries}\label{sec:prelim}
%
\ifconf{In this section we review some basic concepts and definitions. Due to space constraints, we defer an introduction of communication complexity and information theory to the full version of this paper. We also refer readers to the text \cite{Cover:Thomas:91} for a comprehensive introduction to information theory, 
the text \cite{eyal97:_commun_} for communication complexity, and Bar-Yossef's Thesis~\cite{Bar-Yossef:02}, Chapter 6 for information complexity.}

\iffull{
In this section we review some basics on communication complexity and information theory. 

\paragraph{Information Theory}
We refer the reader to \cite{Cover:Thomas:91} for a comprehensive introduction to information theory. Here we review a few concepts and notations. 

Let $H(X)$ denote the Shannon entropy of the random variable $X$, and let $H_b(p)$ denote the binary entropy function when $p \in [0,1]$. Let $H(X\ |\ Y)$ denote conditional entropy of $X$ given $Y$. Let $I(X; Y)$ denote the mutual information between two random variables $X, Y$. Let $I(X; Y\ |\ Z)$ denote the mutual information between two random variables $X, Y$ conditioned on $Z$. The following is a summarization of the basic properties of entropy and mutual information that we need.
\begin{proposition}\label{prop:mut}
Let $X, Y, Z, W$ be random variables.
\begin{enumerate}
\item If $X$ takes value in $\{1,2, \ldots, m\}$, then $H(X) \in [0, \log m]$.

\item $H(X) \ge H(X\ |\ Y)$ and $I(X; Y)  = H(X) - H(X\ |\ Y) \ge 0$.

\item  If $X$ and $Z$ are independent, then we have $I(X; Y\ |\ Z) \ge I(X; Y)$. Similarly, if $X, Z$ are independent given $W$, then $I(X; Y\ |\ Z, W) \ge I(X; Y\ |\ W)$.

\item (Chain rule of mutual information)

$I(X, Y; Z) = I(X; Z) + I(Y; Z\ |\ X).$ 

And in general, for any random variables $X_1, X_2, \ldots, X_n, Y$,

$\textstyle I(X_1, \ldots, X_n; Y) = \sum_{i = 1}^n I(X_i; Y\ |\ X_1, \ldots, X_{i-1})$.

Thus, $I(X, Y; Z\ |\ W) \ge I(X; Z\ |\ W)$.

\item (Data processing inequality) 
If $X$ and $Z$ are conditionally independent given $Y$, then $I(X; Y\ |\ Z, W) \le I(X; Y\ |\ W)$.

\item (Fano's inequality) Let $X$ be a random variable chosen from domain $\mathcal{X}$ according to distribution $\mu_X$, and $Y$ be a random variable chosen from domain $\mathcal{Y}$ according to distribution $\mu_Y$. For any reconstruction function $g : \mathcal{Y} \to \mathcal{X}$ with error $\delta_g$,
$$H_b(\delta_g) + \delta_g \log(\abs{\mathcal{X}} - 1) \ge H(X\ |\ Y).$$

\item (The Maximum Likelihood Estimation principle) With the notations as in Fano's inequality, if the (deterministic) reconstruction function is $g(y) = x$ for the $x$ that maximizes the conditional probability $\mu_X(x\ |\ Y = y)$, then $$\delta_g \le 1 - \frac{1}{2^{H(X\ |\ Y)}}.$$ Call this $g$ the maximum likelihood function.
\end{enumerate}
\end{proposition}

\paragraph{Communication complexity}
In the two-party randomized communication complexity model (see e.g., \cite{eyal97:_commun_}), we have two players Alice and Bob. Alice is given $x \in \mathcal{X}$ and Bob is given $y \in \mathcal{Y}$, and they want to jointly compute a function $f(x,y)$ by exchanging messages according to a protocol $\Pi$. Let $\Pi(x,y)$ denote the message transcript when Alice and Bob run protocol $\Pi$ on input pair $(x,y)$. We sometimes abuse notation by identifying the protocol and the corresponding random transcript, as long as there is no confusion.

The {\em communication complexity} of a protocol is defined as the maximum number of bits exchanged among all pairs of inputs. 
We say a protocol $\Pi$ computes $f$ with error probability $\delta\ (0 \le \delta \le 1)$ if there exists a function $g$ such that for all input pairs $(x,y)$, $\Pr[g(\Pi(x,y)) \neq f(x,y)] \le \delta$. The $\delta$-error randomized communication complexity, denoted by $R^\delta(f)$, is the cost of the minimum-communication randomized protocol that computes $f$ with error probability $\delta$. The $(\mu, \delta)$-distributional communication complexity of $f$, denoted by $D_\mu^\delta(f)$, is the cost of the minimum-communication deterministic protocol that gives the correct answer for $f$ on at least a $1 - \delta$ fraction of all input pairs, weighted by distribution $\mu$. Yao~\cite{Yao77} showed that 
\begin{lemma}[Yao's Lemma]
\label{lem:yao}
$R^\delta(f) \ge \max_\mu D_\mu^\delta(f)$.
\end{lemma} 
Thus, one way to prove a lower bound for randomized protocols is to find a hard distribution $\mu$ and lower bound
$D_\mu^\delta(f)$. This is called Yao's Minimax Principle.

We will use the notion {\em expected distributional communication complexity} $\ED_\mu^\delta(f)$, which was introduced in \cite{PVZ12} (where it was written as $\E[D_\mu^\delta(f)]$, with a bit abuse of notation)  and is defined to be the expected cost (rather than the worst case cost) of the deterministic protocol that gives the correct answer for $f$ on at least $1 - \delta$ fraction of all inputs, where the expectation is taken over distribution $\mu$.

The definitions for two-party protocols can be easily extended to the multiparty setting, where we have $k$ players and the $i$-th player is given an input $x_i \in \mathcal{X}_i$. Again the $k$ players want to jointly compute a function $f(x_1, x_2, \ldots, x_k)$ by exchanging messages according to a protocol $\Pi$.

\paragraph{Information complexity}
Information complexity was introduced in a series of papers including \cite{Chakrabarti01,BYJKS02}. We refer the reader to Bar-Yossef's Thesis~\cite{Bar-Yossef:02}; see Chapter 6 for a detailed introduction. Here we briefly review the concepts of information cost and conditional information cost for $k$-player communication problems. All of them are defined in the blackboard number-in-hand model.

Let $\mu$ be an input distribution on $\mathcal{X}_1 \times \mathcal{X}_2 \times \ldots \times \mathcal{X}_k$ and let $X$ be a random input chosen from $\mu$. Let $\Pi$ be a randomized protocol running on inputs in $\mathcal{X}_1 \times \mathcal{X}_2 \times \ldots \times \mathcal{X}_k$. The {\em information cost} of $\Pi$ with respect to $\mu$ is $I(X; \Pi)$.
The {\em information complexity} of a problem $f$ with respect to a distribution $\mu$ and error parameter $\delta\ (0 \le \delta \le 1)$, denoted $\mathrm{IC}_\mu^\delta(f)$, is the minimum information cost of a $\delta$-error protocol for $f$ with respect to $\mu$.  We will work in the public coin model, in which all parties also share a common source of randomness. 

We say a distribution $\lambda$ partitions $\mu$ if conditioned on $\lambda$, $\mu$ is a product distribution. Let $X$ be a random input chosen from $\mu$ and $D$ be a random variable chosen from $\lambda$. For a randomized protocol $\Pi$ on $\mathcal{X}_1 \times \mathcal{X}_2 \times \ldots \times \mathcal{X}_k$, the {\em conditional information cost} of $\Pi$ with respect to the distribution $\mu$ on $\mathcal{X}_1 \times \mathcal{X}_2 \times \ldots \times \mathcal{X}_k$ and a distribution $\lambda$ partitioning $\mu$ is defined as $I(X; \Pi\ |\ D)$. The {\em conditional information complexity} of a problem $f$ with respect to a distribution $\mu$, a distribution $\lambda$ partitioning $\mu$, and error parameter $\delta\ (0 \le \delta \le 1)$, denoted $\mathrm{IC}_\mu^\delta(f |\ \lambda)$, is the minimum information cost of a $\delta$-error protocol for $f$ with respect to $\mu$ and $\lambda$. 
The following proposition can be found in \cite{BYJKS02}.
\begin{proposition}
\label{prop:icost}
For any distribution $\mu$, distribution $\lambda$ partitioning $\mu$, and error parameter $\delta\ (0 \le \delta \le 1)$,
$$R^\delta(f) \ge \mathrm{IC}_\mu^\delta(f) \ge \mathrm{IC}_\mu^\delta(f\ |\ \lambda).$$
\end{proposition}
}

\paragraph{Statistical distance measures}
Given two probability distributions $\mu$ and $\nu$ over the same space $\mathcal{X}$, the following statistical distance measures will be used in this paper:
\begin{enumerate}
\item {Total variation distance:} $\TV(\mu, \nu) \stackrel{def}{=} \max_{A \subseteq \mathcal{X}} \abs{\mu(A) - \nu(A)}$.

\item {Hellinger distance:} $h(\mu, \nu) \stackrel{def}{=} \sqrt{\frac{1}{2}\sum_{x\in \mathcal{X}}\left(\sqrt{\mu(x)} - \sqrt{\nu(x)}\right)^2}$
\end{enumerate}
We have the following relation between total variation distance and Hellinger distance (cf. \cite{Bar-Yossef:02}, Chapter $2$).
\begin{proposition}\label{sec:varHell}
\label{prop:relation}
$h^2(\mu, \nu) \le \TV(\mu, \nu) \le h(\mu, \nu) \sqrt{2 - h^2(\mu, \nu)}.$
\end{proposition}

The total variation distance of transcripts on a pair of inputs is closely related to the error of a randomized protocol. The following proposition can be found in \cite{Bar-Yossef:02}, Proposition 6.22 (the original proposition is for the 2-party case, and generalizing it to the multiparty case is straightforward).
\begin{proposition} 
\label{prop:diameter}
Let $0 < \delta < 1/2$, and $\Pi$ be a $\delta$-error randomized protocol for a function $f : \mathcal{X}_1 \times \ldots \times \mathcal{X}_k \to \mathcal{Z}$. Then, for every two inputs $(x_1, \ldots, x_k), (x'_1, \ldots, x'_k) \in \mathcal{X}_1 \times \ldots \times \mathcal{X}_k$ for which $f(x_1, \ldots, x_k) \neq f(x'_1, \ldots, x'_k)$, it holds that
$$\TV(\Pi_{x_1, \ldots, x_k}, \Pi_{x'_1, \ldots, x'_k}) > 1 - 2\delta.$$
\end{proposition}


\paragraph{Conventions.}
In the rest of the paper we call a player a {\em site}, as to be consistent with the distributed functional monitoring model. We denote $[n] = \{1, \ldots, n\}$. Let $\oplus$ be the XOR function. All logarithms are base-$2$ unless noted otherwise. We say $\tilde{W}$ is a $(1+\eps)$-approximation of $W$, $0 < \eps < 1$, if $W \le \tilde{W} \le (1+\eps)W$. 

\section{A Lower Bound for $F_0$}
\label{sec:F0}
We introduce a problem called \bit, and then compose it with \TWODISJ\ (studied, e.g., in \cite{Raz90}) to prove a lower bound for $F_0$. In this section we work in the message-passing model.
\subsection{The \bit\ Problem}
\label{sec:bit}
In the \bitG\ problem, we have $k$ sites $S_1, S_2, \ldots, S_k$ and the coordinator. Let $f : \mathcal{X} \times \mathcal{Y} \to \{0,1\}$ be an arbitrary function, and let $\tau$ be an arbitrary distribution on $\mathcal{X} \times \mathcal{Y}$ such that for $(X, Y) \sim \tau$, $f(X, Y) = 1$ with probability $\beta$, and $0$ with probability $1 - \beta$, where $\beta\ (c_\beta/k \le \beta \le 1/c_\beta$
for a sufficiently large constant $c_\beta$) is a parameter. We define the input distribution $\mu$ for \bitG\ on $\{X_1, \ldots, X_k, Y\} \in \mathcal{X}^k \times \mathcal{Y}$ as follows: We first sample $(X_1, Y) \sim \tau$, and then independently sample $X_2, \ldots, X_k \sim \tau\ |\ Y$. Note that each pair $(X_i, Y)$ is distributed according to $\tau$. Let $Z_i = f(X_i, Y)$. Thus $Z_i$'s are i.i.d.\ Bernoulli$(\beta)$. Let $Z = \{Z_1, Z_2, \ldots, Z_k\}$. We assign $X_i$ to site $S_i$ for each $i \in [k]$, and assign $Y$ to the coordinator. 

In the \bitG\ problem, the $k$ sites want to approximate $\sum_{i \in [k]} Z_i$ up to an additive factor of $\sqrt{\beta k}$. In the rest of this section, for convenience, we omit subscripts $f, \tau$ in \bitG, since our results will hold for all $f, \tau$ having the properties mentioned above. 

For a fixed transcript $\Pi = \pi$, let $q_i^\pi = \Pr[Z_i = 1\ |\ \Pi = \pi]$. Thus $\sum_{i \in [k]} q_i^\pi = \E[\sum_{i \in [k]} Z_i\ |\ \Pi = \pi]$. Let $c_0$ be a sufficiently large constant. 

\begin{definition}
Given an input $(x_1, \ldots, x_k, y)$ and a transcript $\Pi = \pi$, let $z_i = f(x_i, y)$ and $z = \{z_1, \ldots, z_k\}$. For convenience, we define $\Pi(z) \triangleq \Pi(x_1, \ldots, x_k, y)$. We say
\begin{enumerate}
\item $\pi$ is {\em bad$_1$} for $z$ (denoted by $z \perp_1 \pi$) if $\Pi(z) = \pi$, and for at least $0.1$ fraction of $\{i \in [k] \ |\ z_i = 1\}$, it holds that $q_i^\pi \le \beta/c_0$, and 

\item $\pi$ is {\em bad$_0$} for $z$ (denoted by $z \perp_0 \pi$) if $\Pi(z) = \pi$, and for at least $0.1$ fraction of $\{i \in [k] \ |\ z_i = 0\}$, it holds that $q_i^\pi \ge \beta/c_0$. 
\end{enumerate}
And $\pi$ is {\em good} for $z$ otherwise.
\end{definition}

In this section, we will prove the following theorem. Except stated explicitly, all probabilities, expectations and variances are taken with respect to the input distribution $\mu$.

\begin{theorem}
\label{thm:bit}
Let $\Pi$ be the transcript of any deterministic protocol for \bit\ on input distribution $\mu$ with error probability $\delta$ for some sufficiently small constant $\delta$, then $\Pr[\Pi \textrm{ is good}] \ge 0.96$.
\end{theorem}

The following observation, which easily follows from the rectangle property of communication protocols, is crucial to our proof. We have included a proof in Appendix~\ref{sec:independent}.
\begin{observation}\label{obs:first}
Conditioned on $\Pi$, $Z_1, Z_2, \ldots, Z_k$ are independent.
\end{observation}

\begin{definition}
\label{def:normal}
We say a transcript $\pi$ is {\em rare$^+$} if $\sum_{i \in [k]} q_i^\pi \ge 4\beta k$ and {\em rare$^-$} if $\sum_{i \in [k]} q_i^\pi  \le \beta k/4$. In both cases we say $\pi$ is {\em rare}. Otherwise we say it is {\em normal}.
\end{definition}

\begin{definition}
\label{def:joker}
We say $Z = \{Z_1, Z_2, \ldots, Z_k\}$ is a {\em joker$^+$} if $\sum_{i \in [k]}Z_i\ge 2\beta k$, and a {\em joker$^-$} if $\sum_{i \in [k]}Z_i\le \beta k / 2$. In both cases we say $Z$ is a {\em joker}. 
\end{definition}

\begin{lemma}
\label{lem:pi-normal}
Under the assumption of Theorem~\ref{thm:bit}, $\Pr[\Pi \textrm{ is normal}] \ge 0.99$.
\end{lemma}

\begin{proof}
First, we can apply a Chernoff bound on random variables $Z_1, \ldots, Z_k$, and get
$$ \Pr[Z \textrm{ is a joker}^+] = \Pr\left[\sum_{i \in [k]}Z_i\ge 2 \beta k \right] \le e^{- \beta k /3}.$$ 

Second, by Observation \ref{obs:first}, we can 
apply a Chernoff bound on random variables $Z_1, \ldots, Z_k$ conditioned on $\Pi$ being rare$^+$,
\begin{eqnarray*} 
&& \Pr[Z \textrm{ is a joker}^+\ |\ \Pi \textrm{ is rare}^+]  \\
&\ge& \sum_{\pi}\Pr\left[\Pi = \pi\ |\ \Pi \textrm{ is rare}^+\right] \Pr\left[Z \textrm{ is a joker}^+\ |\ \Pi = \pi, \Pi \textrm{ is rare}^+\right]\\
&=&   \sum_{\pi}\Pr\left[\Pi = \pi\ |\ \Pi \textrm{ is rare}^+\right]  \Pr\left[ \left. \sum_{i \in [k]}Z_i\ge 2\beta k \ \right| \sum_{i \in [k]} q_i^\pi \ge 4\beta k, \Pi = \pi \right] \\
&\ge&  \sum_{\pi}\Pr\left[\Pi = \pi\ |\ \Pi \textrm{ is rare}^+\right] \left(1 - e^{-\beta k/2}\right) \\
&=&  \left(1 - e^{-\beta k/2}\right).
\end{eqnarray*}

Finally by Bayes' theorem, we have that 
\begin{eqnarray*} 
 \Pr[\Pi \textrm{ is rare}^+]
& = & \frac{\Pr[Z \textrm{ is a joker}^+] \cdot \Pr[\Pi \textrm{ is rare}^+\ |\ Z \textrm{ is a joker}^+]}{\Pr[Z \textrm{ is a joker}^+\ |\ \Pi \textrm{ is rare}^+]} \\
&\le& \frac{e^{-\beta k/3}}{1 - e^{-\beta k/2}} \le 2 e^{-\beta k/3}.
\end{eqnarray*} 
Similarly, we can also show that $\Pr[\Pi \textrm{ is rare}^-] \le 2e^{-\beta k / 8}$. Therefore $\Pr[\Pi \textrm{ is rare}] \le 4e^{-\beta k / 8} \le 0.01$ (recall that by our assumption $\beta k \ge c_\beta$ for a sufficiently large constant $c_\beta$).
\end{proof}

\begin{definition}
\label{def:pi-strong}
Let $c_\ell = 40 c_0$. We say a transcript $\pi$ is {\em weak} if $\sum_{i \in [k]}q_i^\pi(1 - q_i^\pi) \ge  \beta k / c_\ell$, and {\em strong} otherwise.
\end{definition}

\begin{lemma}
\label{lem:pi-strong}
Under the assumption of Theorem~\ref{thm:bit}, $\Pr[\Pi \textrm{ is normal and strong}] \ge 0.98$.
\end{lemma}

\begin{proof}
We first show that for a normal and weak transcript $\pi$, there exists a constant $\delta_\ell = \delta_\ell(c_\ell)$ such that
\begin{eqnarray}
 \Pr\left[ \left. \sum_{i \in [k]}Z_i\le \sum_{i \in [k]}q_i^\pi + 2\sqrt{\beta k} \ \right|\ \Pi = \pi \right]  
&\ge& \delta_\ell, \label{eq:anti-1}\\ 
  \text{and \quad}   \Pr\left[\left. \sum_{i \in [k]}Z_i\ge \sum_{i \in [k]}q_i^\pi  + 4\sqrt{\beta k} \ \right|\ \Pi = \pi \right] 
&\ge& \delta_\ell. \label{eq:anti-2}
\end{eqnarray}
The first inequality is a simple application of Chernoff-Hoeffding bound. Recall that for a normal $\pi$, $\sum_{i \in [k]} q_i^\pi \le 4\beta k$. We have
\begin{eqnarray*}
&&\Pr\left[\left. \sum_{i \in [k]} Z_i \le \sum_{i \in [k]} q_i^\pi +  2\sqrt{\beta k}\ \right|\ \Pi = \pi, \Pi \text{ is normal}\right]\\
& \ge &  1 - \Pr\left[\left. \sum_{i \in [k]} Z_i \ge \sum_{i \in [k]} q_i^\pi  +  2\sqrt{\beta k}\ \right|\ \Pi = \pi, \Pi \text{ is normal} \right] \\
&\ge& 1 - e^{- \frac{8\sqrt{\beta k}^2}{\sum_{i \in [k]} q_i^\pi}} \ge 1 - e^{-2} \ge \delta_\ell. \quad (\text{for a sufficiently small constant $\delta_\ell$})
\end{eqnarray*}

Now we prove for the second inequality. We will need the following anti-concentration result which is an easy consequence of Feller~\cite{feller:43} (cf. \cite{Matousek:08}).

\begin{fact}(\cite{Matousek:08})
\label{lem:feller}
Let $Y$ be a sum of independent random variables, each attaining values in $[0,1]$, and let $\sigma = \sqrt{\var[Y]} \ge 200$. Then for all $t \in [0, \sigma^2/100]$, we have
$$\Pr[Y \ge \E[Y] + t] \ge c\cdot e^{-t^2/(3\sigma^2)}$$
for a universal constant $c > 0$. 
\end{fact}

For a normal and weak $\Pi = \pi$, it holds that
\begin{eqnarray*}
 \var\left[\sum_{i \in [k]} Z_i\ |\ \Pi = \pi \right] &=&  \sum_{i \in [k]} \var\left[Z_i\ |\ \Pi = \pi\right] \quad (\text{by observation	\ref{obs:first}}) \\
&=&  \sum_{i \in [k]}q_i^\pi(1 - q_i^\pi) \\
&\ge& \beta k / c_\ell.  \quad (\text{by definition of a weak $\pi$})
\end{eqnarray*}
Recall that by our assumption, $\beta k \ge c_\beta$ for a sufficiently large constant $c_\beta$,  thus $\sqrt{\beta k} \le \beta k / (100 c_\ell)$ and $\beta k/c_\ell \ge 200^2$. Using Lemma~\ref{lem:feller}, we have for a universal constant $c$,
\begin{eqnarray*}
\label{eq:bit-2}
&&  \Pr\left[\left. \sum_{i \in [k]} Z_i \ge \sum_{i \in [k]} q_i^\pi +  4\sqrt{\beta k}\ \right|\ \Pi = \pi, \Pi \text{ is weak} \right] \\
&\ge& c \cdot e^{-\frac{(4\sqrt{\beta k})^2}{3 \beta k / c_\ell}}\ge c \cdot e^{- 16c_\ell/3} \ge \delta_\ell. \quad (\text{for a sufficiently small constant $\delta_\ell$})
\end{eqnarray*}

By (\ref{eq:anti-1}) and (\ref{eq:anti-2}), it is easy to see that given that $\Pi$ is normal, it cannot be weak with probability more than $0.01$, since otherwise by Lemma~\ref{lem:pi-normal} and the analysis above, the error probability of the protocol will be at least 
$0.99 \cdot 0.01 \cdot \delta_\ell > \delta$,
for an arbitrarily small constant error $\delta$, violating the success guarantee of the lemma. Therefore, 
$$\Pr[\Pi \text{ is normal and strong}] \ge \Pr[\Pi \text{ is normal}]\ \Pr[\Pi \text{ is strong}\ |\ \Pi \text{ is normal}] \ge 0.99 \cdot 0.99 \ge 0.98.$$
\end{proof}

Now we analyze the probability of $\Pi$ being good. 
For a $Z= z$, let $H_0(z) = \{i\ |\ z_i = 0\}$ and $H_1(z) = \{i\ |\ z_i = 1\}$. We have the following two lemmas.

\begin{lemma}
\label{lem:pi-bad-0}  
Under the assumption of Theorem~\ref{thm:bit}, $\Pr[\Pi \text{ is bad$_0$}\ |\ \Pi \text{ is normal and strong}] \le 0.01$.
\end{lemma}

\begin{proof}
Consider any $Z = z$. First, by the definition of a normal $\pi$, we have
$\sum_{i : z_i = 0} q_i^\pi \le \sum_{i \in [k]} q_i^\pi \le 4 \beta k$. Therefore the number of $i$'s such that $z_i = 0$ and $q_i^\pi > (1 - \beta/c_0)$ is at most $4\beta k / (1 - \beta/c_0) \le 8 \beta k$.
Second, by the definition of a strong $\pi$, we have 
$\sum_{i : z_i = 0} q_i^\pi (1 - q_i^\pi) \le \sum_{i \in [k]} q_i^\pi (1 - q_i^\pi) \le \beta k / c_\ell$.
Therefore the number of $i$'s such that $z_i = 0$ and $\beta/c_0 \le q_i^\pi \le (1 - \beta/c_0)$ is at most $\frac{\beta k/c_\ell}{\beta/c_0 \cdot (1 - \beta/c_0)} \le 0.05 k\ (c_\ell = 40c_0)$. Also note that if $z$ is not joker, then $\abs{H_0(z)} \ge (k - 2\beta k)$. Thus conditioned on a normal and strong $\pi$, as well as $z$ is not a joker, the number of $i$'s such that $z_i = 0$ and $q_i^\pi <\beta/c_0$ is at least
\begin{eqnarray*}
(k - 2\beta k) - 8 \beta k - 0.05 k > 0.9 k \ge 0.9 \abs{H_0(z)},
\end{eqnarray*}
where we have used our assumption that $\beta \le 1/c_\beta$ for a sufficiently large constant $c_\beta$. We conclude that
$$\Pr[\Pi \text{ is bad$_0$}\ |\ \Pi \text{ is normal and strong}] 
\le \Pr[Z \text{ is a joker}] \le 2e^{-\beta k/8} \le 0.01.$$
\end{proof}

\begin{lemma}
\label{lem:pi-bad-1}
Under the assumption of Theorem~\ref{thm:bit}, $\Pr[\Pi \text{ is bad$_1$}\ |\ \Pi \text{ is normal}] \le 0.01$.
\end{lemma}

\begin{proof}
Call a $\pi$ is bad$_1$ for a set $T \subseteq [k]$ (denoted by $T \perp_1 \pi$), if for more than $0.1$ fraction of $i \in T$, we have $q_i^\pi \le \beta/c_0$. Let $\chi(\mathcal{E}) = 1$ if $\mathcal{E}$ holds and $\chi(\mathcal{E}) = 0$ otherwise. We have
\begin{eqnarray}
&&\Pr[\Pi \text{ is bad$_1$}\ |\ \Pi \text{ is normal}] \nonumber \\
&=& \sum_{\pi} \Pr[\Pi = \pi\ |\ \Pi \text{ is normal}]  \sum_z \Pr[Z = z\ |\ \Pi = \pi, \Pi \text{ is normal}] \ \chi(z \perp_1 \pi) \nonumber\\
&\le& \Pr[Z \text{ is a joker}] + \sum_{\pi} \Pr[\Pi = \pi\ |\ \Pi \text{ is normal}] \nonumber\\
&&\quad \quad  \sum_{\ell \in [{\beta k}/{2}, 2\beta k]} \sum_{T \subseteq [k]: \abs{T} = \ell} \sum_z \Pr[Z = z\ |\ \Pi = \pi, \Pi \text{ is normal}] \ \chi(H_1(z) = T) \ \chi(T \perp_1 \pi) \label{eq:pi-bad-0} \\
&\le& \Pr[Z \text{ is a joker}] + \sum_{\pi} \Pr[\Pi = \pi\ |\ \Pi \text{ is normal}]  \nonumber \\
&&  \quad \quad \sum_{\ell \in [\beta k /2, 2\beta k]} \left(\left. \sum_{\substack{T \subseteq [k]: \abs{T} = \ell \\ T \perp_1 \pi}} \prod_{i \in T} q_i^\pi\ \right|\ \Pi = \pi, \Pi \text{ is normal} \right) \label{eq:pi-bad-1}
\end{eqnarray}
The last inequality holds since in (\ref{eq:pi-bad-1}), in the last term, we count the probability of each possible set $T$ of size $\ell$ and is $\perp_1$ to $\pi$ that its elements are all $1$, which upper bounds the corresponding summation in (\ref{eq:pi-bad-0}). 
Now for a fixed $\ell$, conditioned on a normal $\pi$, we consider the term 
\begin{eqnarray}
\label{eq:sum-2}
\sum_{\substack{T \subseteq [k] : \abs{T} = \ell \\ T \perp_1 \pi}} \prod_{i \in T} q_i^\pi.
\end{eqnarray}
W.l.o.g., we can assume that $q_1^\pi \ge \ldots \ge q_s^\pi > \beta/c_0 \ge q_{s+1}^\pi \ge \ldots \ge q_k^\pi$ for an $s = \kappa_s k\ (0 < \kappa_s \le 1)$. We consider a pair $(q_u^\pi, q_v^\pi)\ (u, v \in [k])$. Terms in the summation (\ref{eq:sum-2}) that includes either $q_u^\pi$ or $q_v^\pi$ can be written as
\begin{eqnarray*}
q_u^\pi \sum_{\substack{T \subseteq [k] : \abs{T} = \ell \\ T \perp_1 \pi\\ u \in T, v \not\in T}} \prod_{i \in T \backslash u} q_i^\pi + q_v^\pi \sum_{\substack{T \subseteq [k] : \abs{T} = \ell \\ T \perp_1 \pi \\ v \in T, u \not\in T}} \prod_{i \in T \backslash v} q_i^\pi + q_u^\pi q_v^\pi \sum_{\substack{T \subseteq [k] : \abs{T} = \ell \\ T \perp_1 \pi \\ v \in T, u \in T}} \prod_{i \in T \backslash v, u} q_i^\pi.
\end{eqnarray*}
By the symmetry of $q_u^\pi, q_v^\pi$, the sets $\{T \backslash u\ |\ T \subseteq [k], \abs{T} = \ell, T \perp_1 \pi, u \in T, v \not\in T\}$ and $\{T \backslash v\ |\ T \subseteq [k], \abs{T} = \ell, T \perp_1 \pi, v \in T, u \not\in T\}$ are the same. Using this fact and the AM-GM inequality, it is easy to see that the sum will not decrease if we set $(q_u^\pi)' = (q_v^\pi)' = (q_u^\pi + q_v^\pi)/{2}$. Call such an operation an equalization. We repeat applying such equalizations to any pair $(q_u^\pi, q_v^\pi)$, with the constraint that if $u \in [1, s]$ and $v \in [s+1, k]$, then we only ``average" them to the extent that $(q_u^\pi)' = \beta/c_0, (q_v^\pi)' = q_u^\pi + q_v^\pi - \beta/c_0$ if $q_u^\pi + q_v^\pi \le 2\beta/c_0$, and $(q_v^\pi)' = \beta/c_0, (q_u^\pi)' = q_u^\pi + q_v^\pi - \beta/c_0$ otherwise. We introduce this constraint because we do not want to change $\abs{\{i\ |\ (q_i^\pi)' \le \beta/c_0\}}$, since otherwise a set $T$ which was originally $\perp_1 \Pi$ can be $\not\perp_1 \Pi$ after these equalizations.  We cannot further apply equalizations when one of the followings happen.
\begin{eqnarray}
\label{eq:average}
(q_1^\pi)' = \ldots = (q_s^\pi)' > \beta/c_0 = (q_{s+1}^\pi)' = \ldots = (q_k^\pi)'.
\end{eqnarray}
\begin{eqnarray}
\label{eq:average-2}
(q_1^\pi)' = \ldots = (q_s^\pi)' = \beta/c_0 \ge (q_{s+1}^\pi)' = \ldots = (q_k^\pi)'.
\end{eqnarray}
We note that actually (\ref{eq:average-2}) cannot happen since $\sum_{i \in [k]}(q_i^\pi)' = \sum_{i \in [k]}q_i^\pi$ is preserved during equalizations, and conditioned on a normal $\pi$, we have $\sum_{i \in [k]} q_i^\pi \ge \beta k / 4 > \beta k/c_0$. 

Let $q = (q_1^\pi)' = \ldots = (q_s^\pi)'$. For a normal $\pi$, it holds that $\sum_{i \in [k]}(q_i^\pi)' =  s \cdot q + (k - s) \cdot \beta/c_0 = r \in [\beta k/4, 4\beta k]$. Let $\alpha \in (0.1, 1]$. Recall that $\ell \in [\beta k/2, 2\beta k]$, and we have set $s = \kappa_s k$.
We try to upper bound (\ref{eq:sum-2}) using (\ref{eq:average}).
\begin{eqnarray}
\sum_{\substack{T \subseteq [k] : \abs{T} = \ell \\ T \perp_1 \pi}} \prod_{i \in T} q_i^\pi.
&\le& \left( {k - s \choose \alpha \ell} \cdot {s \choose (1 - \alpha) \ell} \right) \cdot \left( \left(\frac{\beta}{c_0}\right)^{\alpha \ell} \cdot \left(\frac{r}{s} - \frac{(k - s)\beta}{c_0 s} \right)^{(1 - \alpha) \ell} \right) \label{eq:alpha} \\
&\le& \left( \left(\frac{e (1 - \kappa_s)k}{\alpha \ell}\right)^{\alpha \ell} \cdot \left(\frac{e \kappa_s k}{(1 - \alpha) \ell}\right)^{(1 - \alpha) \ell} \right) \cdot \left( \left(\frac{\beta}{c_0}\right)^{\alpha \ell} \cdot \left(\frac{r}{\kappa_s k}  \right)^{(1 - \alpha) \ell} \right) \nonumber \\
&\le& \left( \frac{e}{\alpha c_0} \cdot \frac{\beta k}{\ell} \right)^{\alpha \ell} \cdot \left(\frac{e r}{(1 - \alpha) \ell}\right)^{(1 - \alpha) \ell} \nonumber \\
&\le& \left(\frac{8e}{(c_0)^\alpha \cdot \alpha^\alpha (1 - \alpha)^{1 - \alpha}}\right)^\ell \nonumber \\
&\le& \left(\frac{8e}{(c_0)^{0.1} \cdot (1/e)^{2/e}}\right)^{\beta k / 2}  \label{eq:alpha-2}
\end{eqnarray}
In (\ref{eq:alpha}), the first term is the number of possible choices of the set $T\ (T = \ell)$ with $\alpha$ fraction of items in $[s+1, \infty]$, and the rest in $[1, s]$. And the second term upper bounds $\prod_{i \in T} q_i^\pi$ according to the discussion above. Here we have assumed $\alpha < 1$, otherwise if $\alpha = 1$, then $(\ref{eq:alpha}) \le {k \choose \ell} \cdot (\beta/c_0)^\ell \le (2e/c_0)^{\beta k/2}$, which is smaller than (\ref{eq:alpha-2}).
Now, (\ref{eq:pi-bad-1}) can be upper bounded by
\begin{eqnarray*}
&&2e^{-\beta k/8} +  \sum_{\pi} \Pr[\Pi = \pi\ |\ \Pi \text{ is normal}] \cdot 2\beta k \cdot \left(\frac{8e}{(c_0)^{0.1} \cdot (1/e)^{2/e}}\right)^{\beta k / 2}  \\
&=&2e^{-\beta k/8} +  2\beta k \cdot \left(\frac{8e}{(c_0)^{0.1} \cdot (1/e)^{2/e}}\right)^{\beta k / 2}  \\
&\le& 0.01. \quad \text{(for a sufficiently large constant $c_0$)}
\end{eqnarray*}
\end{proof}

Finally, combining Lemma~\ref{lem:pi-strong}, Lemma~\ref{lem:pi-bad-0} and Lemma~\ref{lem:pi-bad-1}, we get
\begin{eqnarray*}
&&\Pr[\Pi \text{ is good}] \ge \Pr[\Pi \text{ is good, normal and strong}] \\
&=& \Pr[\Pi \text{ is normal and strong}]  (1 - \Pr[\Pi \text{ is bad$_0$}\ |\ \Pi \text{ is normal and strong}] \\
&& \quad \quad \quad \quad - \Pr[\Pi \text{ is bad$_1$}\ |\ \Pi \text{ is normal and strong}]) \\
&\ge& \Pr[\Pi \text{ is normal and strong}]  (1 - \Pr[\Pi \text{ is bad$_0$}\ |\ \Pi \text{ is normal and strong}]) \\
&& \quad \quad \quad \quad - \Pr[\Pi \text{ is normal}]\ \Pr[\Pi \text{ is bad$_1$}\ |\ \Pi \text{ is normal}] \\
&\ge& 0.98 \cdot (1 - 0.01) - 0.01 \ge 0.96. 
\end{eqnarray*}

\iffull{\subsection{The \TWODISJ\ Problem}}
\ifconf{{\bf The \TWODISJ\ Problem:}}
In \TWODISJ\ problem, Alice has a set $x \subseteq [n]$ and Bob has a set $y \subseteq [n]$.
Their goal is to output $1$ if $x \cap y \neq \emptyset$, and $0$ otherwise.

We define the input distribution $\tau_\beta$ as follows. Let $\ell = (n+1)/4$. With probability $\beta$, $x$ and $y$ are random subsets of $[n]$ such that $\abs{x} = \abs{y} = \ell$ and $\abs{x \cap y} = 1$. And with probability $1 - \beta$, $x$ and $y$ are random subsets of $[n]$ such that $\abs{x} = \abs{y} = \ell$ and $x \cap y = \emptyset$.  Razborov~\cite{Raz90} proved that for $\beta = 1/4$, $D^{1/(400)}_{\tau_{1/4}}(\mbox{\TWODISJ}) = \Omega(n)$. It is easy to extend this result to general $\beta$ and the average-case complexity.
\begin{theorem}[\cite{PVZ12}, Lemma 2.2]
\label{thm:DISJ}
For any $\beta \le 1/4$,
it holds that $\ED^{\beta/100}_{\tau_\beta}(\mbox{\TWODISJ}) = \Omega(n)$, where the expectation is taken over the input distribution $\tau_\beta$.
\end{theorem}

In the rest of the section, we simply write $\tau_\beta$ as $\tau$.
 
\iffull{\subsection{The Complexity of $F_0$}}
\ifconf{\\\\{\bf The Complexity of $F_0$:}}
\subsubsection{Connecting $F_0$ and \bitD}

Set $\beta = 1/(k\eps^2)$, $B = 20000/\delta$, where $\delta$ is the small constant error parameter for \bit\ in Theorem~\ref{thm:bit}. 

We choose $f$ to be \TWODISJ\ with universe size $n = B/\eps^2$, set its input distribution to be $\tau$, and work on \bitD. Let $\mu$ be the input distribution of \bitD, which is a function of $\tau$ (see Section~\ref{sec:bit} for the detailed construction of $\mu$ from $\tau$). Let $\{X_1, \ldots, X_k, Y\} \sim \mu$. Let $Z_i = \text{\TWODISJ}(X_i, Y)$. Let $\zeta$ be the induced distribution of $\mu$ on $\{X_1, \ldots, X_k\}$ which we choose to be the input distribution for $F_0$. 
In the rest of this section, for convenience, we will omit the subscripts \TWODISJ\ and $\tau$ in \bitD\ when there is no confusion.

\iffull{
Let $N = \sum_{i \in [k]} Z_i = \sum_{i \in [k]} \text{\TWODISJ}(X_i, Y)$. Let $R = F_0(\cup_{i \in [k]}X_i \bigcap Y)$. The following lemma shows that $R$ will concentrate around its expectation $\E[R]$, which can be calculated exactly. 

\begin{lemma}
\label{lem:bin-ball}
With probability at least $(1 - 6500/B)$, we have $\abs{R - \E[R]} \le 1/(10\eps)$, where $\E[R] = (1 - \lambda) N$ for some fixed  constant $0 \le \lambda \le 4/B$.
\end{lemma}

\begin{proof}
We can think of our problem as a bin-ball game: Think each pair $(X_i, Y)$ such that \TWODISJ$(X_i, Y) = 1$ are balls (thus we have $N$ balls), and elements in the set $Y$ are bins. Let $\ell = \abs{Y}$. We throw each of the $N$ balls into one of the $\ell$ bins uniformly at random. Our goal is to estimate the number of non-empty bins at the end of the process.

By a Chernoff bound, with probability $\left(1 - e^{-\beta k/3}\right) \ge (1 - 100/B)$, $N \le 2\beta k = 2/\eps^2$. By Fact~$1$ and Lemma $1$ in \cite{Kane10}, we have $\E[R] = \ell \left(1 - (1 - 1/\ell)^N\right)$ and $\var[R] < 4N^2/\ell$. Thus by Chebyshev's inequality we have 
$$\Pr[\abs{R - \E[R]} > 1/(10\eps)] \le \frac{\var[R]}{1/(100\eps^2)} \le \frac{6400}{B}.$$
Let $\theta = N/\ell \le 8/B$. We can write
$$\E[R] = \ell\left(1 - e^{-\theta}\right) + O(1)  = \theta\ell \left(1 - \frac{\theta}{2!} + \frac{\theta^2}{3!} - \frac{\theta^3}{4!} + \right) + O(1).$$
This series converges and thus we can write $\E[R] = (1 - \lambda) \theta \ell = (1 - \lambda) N$ for some fixed  constant $0 \le \lambda \le \theta/2 \le 4/B$.
\end{proof}
}

The next lemma shows that we can use a protocol for $F_0$ to solve \bit\ with good properties.

\begin{lemma}
\label{lem:F0-reduction}
Any protocol $\cal P$ that computes a $(1 + \gamma \eps)$-approximation to $F_0$ (for a sufficiently small constant $\gamma$) on input distribution $\zeta$ with error probability $\delta/2$ can be used to compute  \bitD\  on input distribution $\mu$ with error probability $\delta$. 
\end{lemma}

\iffull{
\begin{proof}
Given an input $\{X_1, \ldots, X_k, Y\} \sim \mu$ for \bit. The $k$ sites and the coordinator use $\cal P$ to compute $\tilde{W}$ which is a $(1 + \gamma \eps)$-approximation to $F_0(X_1, \ldots, X_k)$, and then determine the answer to \bit\ to be
\begin{eqnarray*}
\label{eq:estimator}
 \frac{\tilde{W} - (n - \ell)}{1 - \lambda}.
\end{eqnarray*}
Recall that $0 \le \lambda \le 4/B$ is some fixed constant, $n = B/\eps^2$ and $\ell = (n+1)/4$.

\paragraph{Correctness.} Given a random input $(X_1, \ldots, X_k, Y) \sim \zeta$, the exact value of $W = F_0(X_1, \ldots, X_k)$ can be written as the sum of two components.
\begin{equation}
\label{eq:F0-1}
W = Q + R,
\end{equation}
where $Q$ counts $F_0(\cup_{i \in [k]} X_i \backslash Y)$, and $R$ counts $F_0(\cup_{i \in [k]} X_i \bigcap Y)$. First, from our construction it is easy to see by a Chernoff bound and the union bound that with probability $\left(1 - 1/\eps^2 \cdot e^{-\Omega(k)}\right) \ge 1 - 100/B$, we have $Q = \abs{\{[n] - Y\}} = n - \ell$, since each element in $\mathcal{S} \backslash Y$ will be chosen by every $X_i\ (i = 1, 2, \ldots, k)$ with a probability at least $1/4$. Second, by Lemma~\ref{lem:bin-ball} we know that with probability $(1 - 6500/B)$, $R$ is within $1/(10\eps)$ from its mean $(1 - \lambda)N$ for some fixed constant $0 \le \lambda \le 4/B$. Thus with probability $(1 - 6600/B)$, we can write Equation (\ref{eq:F0-1}) as 
\begin{equation}
\label{eq:F0-2}
W = (n - \ell) + (1 - \lambda) N+ \kappa_1, 
\end{equation}
for a value $\abs{\kappa_1} \le 1/(10\eps)$ and $N \le 2/\eps^2$.

Set $\gamma = 1/(20 B)$. Since $F_0(X_1, X_2, \ldots, X_k)$ computes a value $\tilde{W}$ which is a $(1 + \gamma \eps)$-approximation of $W$, we can substitute $W$ with $\tilde{W}$ in Equation~(\ref{eq:F0-2}), resulting in the following.
\begin{equation}
\label{eq:F0-3}
\tilde{W} = (n - \ell) + (1 - \lambda) N+ \kappa_1 + \kappa_2,
\end{equation}
where $\abs{\kappa_1} \le 1/(10\eps)$, $N \le 2/\eps^2$, and
\begin{eqnarray*}
\kappa_2 &\le& \gamma \eps \cdot W \\
&=& \gamma \eps \cdot ((n - \ell)  + (1 - \lambda) N+ \kappa_1) \\
&\le& \gamma \eps \cdot (B/\eps^2  + 2/\eps^2 + 1/(10\eps)) \\
&\le& 1/(10\eps). 
\end{eqnarray*}
Now we have
\begin{eqnarray*}
N & = & (\tilde{W} - (n - \ell) - \kappa_1 - \kappa_2)/(1 - \lambda) \\
& = & (\tilde{W} - (n - \ell))/(1 - \lambda) + \kappa_3,
\end{eqnarray*}
where $\abs{\kappa_3} \le (1/(10\eps) + 1/(10\eps))/(1 - 4/B) \le 1/(4\eps)$. Therefore $(\tilde{W} - (n - \ell))/(1 - \lambda)$ approximates $N = \sum_{i \in [k]}Z_i$ correctly up to an additive error $1/(4\eps) < \sqrt{\beta k} = 1/\eps$, thus computes \bit\ correctly. The total error probability of this simulation is at most $(\delta/2 + 6600/B)$, where the first term counts the error probability of $\cal P'$ and the second term counts the error probability introduced by the reduction. This is less than $\delta$ if we choose $B = 20000/\delta$.
\end{proof}
}

\subsubsection{An Embedding Argument}
\label{sec:embed}
\begin{lemma}
\label{lem:embed}
Suppose that there exists a deterministic protocol $\cal P'$ which computes $(1 + \gamma \eps)$-approximate $F_0$ (for a sufficiently small constant $\gamma$) on input distribution $\zeta$ with error probability $\delta/2$ (for a sufficiently small constant $\delta$) and communication $o(C)$, then there exists a deterministic protocol $\cal P$ that computes \TWODISJ\ on input distribution $\tau$ with error probability $\beta/100$ and expected communication complexity $o(\log(1/\beta) \cdot C/k)$, where the expectation is taken over the input distribution $\tau$.
\end{lemma}

\begin{proof}
In \TWODISJ, Alice holds $X$ and Bob holds $Y$ such that $(X, Y) \sim \tau$. We show that Alice and Bob can use the deterministic protocol $\cal P'$ to construct a deterministic protocol $\cal P$ for \TWODISJ$(X, Y)$ with desired error probability and communication complexity. 

Alice and Bob first use $\cal P'$ to construct a protocol 
$\cal P''$. During the construction they will use public and private randomness which will be fixed at the end. $\cal P''$ consists of two phases.

{\em Input reduction phase.} Alice and Bob construct an input for $F_0$ using $X$ and $Y$ as follows: They pick a random site $S_I\ (I \in [k])$ using public randomness. Alice assigns $S_I$ with input $X_I = X$, and Bob constructs inputs for the rest $(k - 1)$ sites using $Y$. For each $i \in [k] \backslash I$, Bob samples an $X_i$ according to $\tau\ |\ Y$ using independent private randomness and assigns it to $S_i$. Let $Z_i = \text{\TWODISJ}(X_i, Y)$. Note that $\{X_1, \ldots, X_k, Y\} \sim \mu$ and $\{X_1, \ldots, X_k\} \sim \zeta$.

{\em Simulation phase.} Alice simulates $S_I$ and Bob simulates the rest $(k - 1)$ sites, and they run protocol $\cal P'$ on $\{X_1, \ldots, X_k\} \sim \zeta$ to compute $F_0(X_1, \ldots, X_k)$ up to a $(1 + \gamma \eps)$-approximation for a sufficiently small constant $\gamma$ and error probability $\delta/2$. Let $\pi$ be the protocol transcript, and let $\tilde{W}$ be the output. By Lemma~\ref{lem:F0-reduction}, we can use $\tilde{W}$ to compute \bit\ with error probability $\delta$. And then by Theorem~\ref{thm:bit}, for $0.96$ fraction of $Z = z$  over the input distribution $\mu$ and $\pi = \Pi(z)$, it holds that for $0.9$ fraction of $\{i \in [k]\ |\ z_i = 0\}$, $q_i^\pi < \beta/c_0$, and $0.9$ fraction of $\{i \in [k]\ |\ z_i = 1\}$, $q_i^\pi > \beta/c_0$. Now $\cal P''$ outputs $1$ if $q_I^\pi > \beta/c_0$, and $0$ otherwise. Since $S_I$ is chosen randomly among the $k$ sites, and the inputs for the $k$ sites are identically distributed,  $\cal P''$ computes $Z_I = \text{\TWODISJ}(X, Y)$ on input distribution $\tau$ correctly with probability $0.96 \cdot 0.9 \ge 0.8$.

We now describe the final protocol $\cal P$: Alice and Bob  repeat $\cal P''$ independently for $c_R \log(1/\beta)$ times for a large enough constant $c_R$. At the $j$-th repetition, in the input reduction phase, they choose a random permutation $\sigma_j$ of $[n]$ using public randomness, and apply it to each element in $X_1, \ldots, X_k$ before assigning them to the $k$ sites. After running $\cal P''$ for $c_R \log(1/\beta)$ times, $\cal P$ outputs the majority of the outcomes. 

Since $Z_I = \text{\TWODISJ(X, Y)}$ is fixed at each repetition, the inputs $\{X_1, \ldots, X_k\}$ at each repetition have a small dependence, but conditioned on $Z_I$, they are all independent. Let $\mu'$ to be input distribution of $\{X_1, \ldots, X_k, Y\}$ conditioned on $Z_I = b$. Let $\zeta'$ be the induced distribution of $\mu'$ on $\{X_1, \ldots, X_k\}$. The successful probability of a run of $\cal P''$ on $\zeta'$ is at least $0.8 - \TV(\zeta, \zeta')$, where $\TV(\zeta, \zeta')$ is the total variation distance between distributions $\zeta, \zeta'$, which is at most 
$$\max\{\TV(\text{Binomial}(k, \beta), \text{Binomial}(k - 1, \beta)), \TV(\text{Binomial}(k, \beta), \text{Binomial}(k - 1, \beta) + 1)\},$$ and can be bounded by $O(1/\sqrt{\beta k}) = O(\eps)$ (see, e.g., Fact 2.4 of \cite{GMRZ11}). Since conditioned on $Z_I$, the inputs at each repetition are independent, and the success probability of each run of $\cal P''$ is at least $0.7$, by a Chernoff bound over the $c_R  \log(1/\beta)$ repetitions for a sufficiently large $c_R$, we conclude that $\cal P$ succeeds with error probability $\beta/1600$.

We next consider the communication complexity. At each run of $\cal P''$, let $\CC(S_I, S_{-I})$ be the expected communication cost between the site $S_I$ and the rest players (more precisely, between $S_I$ and the coordinator, since in the coordinator model all sites only talk to the coordinator, whose initial input is $\emptyset$), where the expectation is taken over the input distribution $\zeta$ and the choice of the random $I \in [k]$. Since 
conditioned on $Y$, all $X_i\ (i \in [k])$ are independent and identically distributed, if we take a random site $S_I$, the expected communication between $S_I$ and the coordinator should be equal to the total communication divided by a factor of $k$. Thus we have $\CC(S_I, S_{-I}) = o(C/k)$.
Finally, by the linearity of expectation, the expected total communication cost of the $O(\log(1/\beta))$  runs of $\cal P''$ is $o(\log(1/\beta) \cdot C/k)$.

At the end we fix all the randomness used in construction of protocol $\cal P$. We first use two Markov inequalities to fix all public randomness such that $\cal P$ succeeds with error probability $\beta/400$, and the expected total communication cost of the $o(\log(1/\beta) C / k)$, where  both the error probability and the cost expectation are taken over the input distribution $\mu$ and Bob's private randomness. We next use another two Markov inequalities to fix Bob's private randomness such that $\cal P$ succeeds with error probability $\beta/100$, and the expected total communication cost of the $o(\log(1/\beta) C / k)$, where both the error probability and the cost expectation are taken over the input distribution $\mu$.
\end{proof}

The following theorem is a direct consequence of Lemma~\ref{lem:embed}, Theorem~\ref{thm:DISJ} for \TWODISJ\ and Lemma~\ref{lem:yao} (Yao's Lemma). Recall that we set $n = O(1/\eps^2)$ and $1/\beta = \eps^2 k$. In the definition for \bit\ we need $c_\beta/k \le \beta \le 1/c_\beta$ for a sufficiently large constant $c_\beta$, thus we require $c_\beta \le 1/\eps^2 \le k/c_\beta$ for a sufficiently large constant $c_\beta$.
 
\begin{theorem}
\label{thm:F0}
Assuming that $c_\beta \le 1/\eps^2 \le k/c_\beta$ for a sufficiently large constant $c_\beta$. Any randomised protocol that computes a $(1 + \eps)$-approximation to $F_0$ with error probability $\delta$ (for a sufficiently small constant $\delta$) has communication complexity $\Omega\left(\frac{k}{\eps^2 \log (\eps^2 k)}\right)$.
\end{theorem}

\section{A Lower Bound for $F_p\ (p > 1)$}
\label{sec:F2-LB}
We first introduce a problem called \XOR\ which can be considered to some extent as a combination of two \DISJ\ (introduced in \cite{alon96:_space,BYJKS02}) instances, and then compose it with \gap\ (introduced in \cite{Sherstov11}) to create another problem that we call the $k$-BLOCK-THRESH-XOR (\thresh) problem. We prove that the communication complexity of \thresh\ is large. Finally, we prove a communication complexity lower bound for $F_p$ by performing a reduction from \thresh. In this section we work in the blackboard model.

\iffull{\subsection{The \gap\ Problem}}
\ifconf{{\bf The \gap\ Problem}:}
\label{sec:gap}
In the \gap\ problem we have two players Alice and Bob. Alice has a vector $x = \{x_1, \ldots, x_{1/\eps^2}\}\in \{0,1\}^{1/\eps^2}$ and Bob has a vector $y = \{y_1, \ldots, y_{1/\eps^2}\} \in \{0,1\}^{1/\eps^2}$. They want to compute  
\begin{eqnarray*}
\begin{array}{l}
2\textrm{-GAP-ORT}(x,y)  =  \left\{
  \begin{array}{rl}
   1, & \abs{\sum\limits_{i \in [1/\eps^2]} \textrm{XOR}(x_i, y_i) - \frac{1}{2\eps^2}} \ge \frac{2}{\eps},\\
   0, & \abs{\sum\limits_{i \in [1/\eps^2]} \textrm{XOR}(x_i, y_i) - \frac{1}{2\eps^2}} \le \frac{1}{\eps},\\
   *, & \text{otherwise.}
  \end{array}
  \right.
\end{array}
\end{eqnarray*}
Let $\phi$ be the uniform distribution on $\{0,1\}^{1/\eps^2} \times \{0,1\}^{1/\eps^2}$ and let $(X,Y)$ be a random input chosen from distribution $\phi$. The following theorem is recently obtained by Chakrabarti et al.~\cite{CKW12}.~\footnote{In our original conference paper, which was published before \cite{CKW12}, we proved the theorem for protocols with $\mathrm{poly}(N)$ communication, using a simple argument based on \cite{Sherstov11} and Theorem 1.3 of \cite{Barak10}. }

\begin{theorem}[\cite{CKW12}]
\label{thm:gap}
Let $\Pi$ be the transcript of any randomized protocol for \gap\ on input distribution $\phi$ with error probability $\iota$, for a sufficiently small constant $\iota > 0$. Then, $I(X, Y; \Pi) \ge {\Omega}(1/\eps^2)$.
\end{theorem}

\iffull{\subsection{The \XOR\ Problem}}
\ifconf{{\bf The \XOR\ Problem:}} 
In the \XOR\ problem we have $k$ sites $S_1, \ldots, S_k$. Each site $S_i\ (i = 1,2,\ldots,k)$ holds a block ${b}_i = \{b_{i,1}, \ldots, b_{i,n}\}$ of $n$ bits. Let $b = (b_1, \ldots, b_k)$ be the inputs of $k$ sites. Let $b_{[k], \ell}$ be the $k$ sites' inputs on the $\ell$-th coordinate. W.l.o.g., we assume $k$ is a power of $2$. 
The $k$ sites want to compute the following function.
\begin{eqnarray*}
\begin{array}{l}
k\textrm{-XOR}({b}_{1}, \ldots, {b}_{k})  =  \left\{
  \begin{array}{rl}
   1, & \text{if $\exists$ $j \in [n]$ such that $b_{i,j} = 1$} \\ 
   & \quad \text{for exactly $k/2$ $i$'s},\\
   0, & \text{otherwise.}
  \end{array}
  \right.
\end{array}
\end{eqnarray*}

We define the input distribution $\varphi_n$ for the \XOR\ problem as follows. For each coordinate $\ell\ (\ell \in [n])$ there is a variable $D_{\ell}$ chosen uniformly at random from $\{1, 2, \ldots, k\}$. Conditioned on $D_{\ell}$, all but the $D_{\ell}$-th sites set their inputs to $0$, whereas the $D_{\ell}$-th site sets its input to $0$ or $1$ with equal probability. We call the $D_{\ell}$-th site the special site in the $\ell$-th coordinate. Let $\varphi_1$ denote this input distribution on one coordinate.

Next, we choose a random special coordinate $M \in [n]$ and replace the $k$ sites' inputs on the $M$-th coordinate as follows: 
For the first $k/2$ sites, with probability $1/2$ we replace all $k/2$ sites' inputs with $0$, and with probability $1/2$ we replace all $k/2$ sites' inputs with $1$; and we independently perform the same operation to the second $k/2$ sites. Let $\psi_1$ denote the distribution on this special coordinate. And let $\psi_n$ denote the input distribution that on the special coordinate $M$ is distributed as $\psi_1$ and on each of the remaining $n-1$ coordinates is distributed as $\varphi_1$.

Let ${B}, B_i, B_{[k], \ell}, B_{i, \ell}$ be the corresponding random variables of ${b}, b_i, b_{[k], \ell}, b_{i, \ell}$ when the input of \XOR\ is chosen according to the distribution $\psi_n$. Let $D = \{D_1, \ldots, D_n\}$. 

\subsection{The \guess\ Problem}
The \guess\ problem can be seen as an augmentation of the \XOR\ problem. The $k$ sites are still given an input $B$, as that in the \XOR\ problem. In addition, we introduce another player called the {\em predictor}. The predictor will be given an input $Z$, but it cannot talk to any of the $k$ sites (that is, it cannot write anything to the blackboard). After the $k$ sites finish the whole communication, the predictor computes the final output $g(\Pi(B), Z)$, where $\Pi(B)$ is the transcript of the $k$ sites' communication on their input $B$, and $g$ is the (deterministic) maximum likelihood function (see Proposition \ref{prop:mut}). In this section when we talk about protocol transcripts, we always mean the concatenation of the messages exchanged by the $k$ sites, but excluding the output of the predictor. 

In the \guess\ problem, the goal is for the predictor to output $(X, Y)$, where $X = 1$ if the inputs of the first $k/2$ sites in the special coordinate $M$ are all $1$ and $X = 0$ otherwise, and $Y = 1$ if the inputs of the second $k/2$ sites in the special coordinate $M$ are all $1$ and $Y = 0$ otherwise.  We say the instance $B$ is a $00$-instance if $X = Y = 0$, a $10$-instance if $X = 1$ and $Y = 0$, a $01$-instance if $X = 0$ and $Y = 1$, and a $11$-instance if $X = Y = 1$. Let $S \in \{00, 01, 10, 11\}$ be the type of an instance.

We define the following input distribution for \guess: We assign an input $B \sim \psi_n$ to the $k$ sites, and $Z = \{D, M\}$ to the predictor, where $D, M$ are those used to construct the $B$ in the distribution $\psi_n$. Slightly abusing notation, we also use $\psi_n$ to denote the joint distribution of $B$ and $\{D, M\}$. We do the same for the one coordinate distributions $\psi_1$ and $\varphi_1$. That is, we also use $\psi_1$ (or $\varphi_1$) to denote the joint distribution of $B_{[k], \ell}$ and $D_\ell$ for a single coordinate $\ell$.

\begin{theorem}
\label{thm:XOR}
Let $\Pi$ be the transcript of any randomized protocol for \guess\ on input distribution $\psi_n$ with success probability $1/4 + \Omega(1)$. Then we have $I(B; \Pi\ |\ D, M) = {\Omega}(n/k)$, where information is measured~\footnote{When we say that the information is measured with respect to a distribution $\mu$ we mean that the inputs to the protocol are distributed according to $\mu$ when computing the mutual information (note that there is also randomness used by $\Pi$ when measuring the mutual information).} with respect to the input distribution $\psi_n$. 
\end{theorem}

\begin{proof}
By a Markov inequality, we know that for $\kappa_1 = \Omega(n)$ of $\ell \in [n]$, the protocol succeeds with probability $1/4 + \Omega(1)$ conditioned on $M = \ell$. Call an $\ell$ for which this holds {\em eligible}. Let $\kappa = n - \kappa_1/2$. We say an $\ell$ is {\em good} if $\ell$ is both eligible and $\ell \le \kappa$. Thus there are $\kappa_1 + \kappa - n = \Omega(n)$ good $\ell$.
Let $D_{-\ell}$ denote the random variable $D$ with $\ell$-th component missing. We say a $d$ is {\em nice} for a good $\ell$ if the protocol succeeds with probability $1/4 + \Omega(1)$ conditioned on $M = \ell$ and $D_{-\ell} = d$. By another Markov inequality, it holds that at least an $\Omega(1)$ fraction of $d$ is nice for a good $\ell$. 

Now we consider $I(B ; \Pi\ |\ D, M, S = 00, M > \kappa)$. Note that if we can show that $I(B ; \Pi\ |\ D, M, S = 00, M > \kappa) = \Omega(n/k)$, then it follows that $I(B ; \Pi\ |\ D, M) = \Omega(n/k)$, since $H(S) = 2$, and $\Pr[(S = 00) \wedge (M > \kappa)] = 1/4 \cdot {\Omega}(1) = {\Omega}(1)$. By the chain rule, expanding the conditioning, and letting $B_{[k], <\ell}$ be the inputs to the $k$ sites on the first $\ell-1$ coordinates, we have
\begin{eqnarray}
&& I(B ; \Pi\ |\ D, M, S = 00, M > \kappa) \nonumber \\
& = & \sum_{\ell \in [n]} I(B_{[k], \ell} ; \Pi\ |\ D, M, S = 00, M > \kappa, B_{[k], <\ell}) \label{eq:a-1}  \\
& = & \sum_{\ell \in [n]} I(B_{[k], \ell} ; \Pi\ |\ D, M, S = 00, M > \kappa) \label{eq:a-2} \\
& \geq & \sum_{\textrm{good } \ell} I(B_{[k], \ell} ; \Pi\ |\ D, M, S = 00, M > \kappa) \nonumber \\
& = &  \sum_{\textrm{good }\ell} \sum_{d} \Pr[D_{-\ell} = d] \cdot I(B_{[k], \ell} ; \Pi \ |\ D_{\ell}, M, S = 00, M > \kappa, D_{-\ell} = d) \nonumber \\
& \ge &  \sum_{\textrm{good }\ell} \sum_{\text{nice $d$ for $\ell$}} \Pr[D_{-\ell} = d] \cdot I(B_{[k], \ell} ; \Pi \ |\ D_{\ell}, M, S = 00, M > \kappa, D_{-\ell} = d),
\label{eq:a-3} 
\end{eqnarray}
where (\ref{eq:a-1}) to (\ref{eq:a-2}) is because $B_{[k], <\ell}$ is independent of $B_{[k], \ell}$ given other conditions, and we apply item $3$ in Proposition \ref{prop:mut}. 

Now let's focus on a good $\ell$ and a nice $d$ for $\ell$. We define a protocol $\Pi_{\ell, d}$ which on input $(A_1, \ldots, A_k, R) \sim \psi_1$, attempts to output $(U,V)$,
where $U = 1$ if $A_1 = \ldots = A_{k/2} = 1$ and $U = 0$ otherwise, and $V = 1$ if $A_{k/2+1} = \ldots = A_k = 1$ and $V = 0$ otherwise. Here $A_1, \ldots, A_k$ are inputs of the $k$ sites and $R$ is the input of the predictor.
The protocol $\Pi_{\ell, d}$ has $(\ell, d)$ hardwired into it, and works as follows. First, the $k$ sites construct an input $B$ for \guess\ distributed according to $\psi_n$, using $\{A_1, \ldots, A_k\}, d$ and their private randomness, without any communication: They set the input on the $\ell$-th coordinate to be $B_{[k], \ell} = \{A_1, \ldots, A_k\}$, and use their private randomness to sample the inputs for coordinates $\ell' \neq \ell$ using the value $d$ and the fact that the inputs to the $k$ sites are independent conditioned on $D_{-\ell} = d$. The predictor sets its input to be $\{D_\ell = R, D_{-\ell} = d, M = \ell\}$. Next, the $k$ sites run $\Pi$ on their input $B$. Finally, the predictor outputs $(U, V) = g(\Pi(B), R, d, \ell)$.

Let $\psi_{n, 00, d} = (\psi_n|S=00, D_{-\ell} = d)$. Let $\varphi_{n, d} =  (\varphi_n|D_{-\ell} = d)$. Let $\psi_{n, 00, d}^\ell$ and $\varphi_{n, d}^\ell$ be the distributions $\psi_{n, 00, d}$ and $\varphi_{n, d}$ of $B$ after embedding $(A_1, \ldots, A_k)$ to the $\ell$-th coordinate conditioned on $M > \kappa$ (recall that $\kappa \ge \ell$ for a good $\ell$, thus $M > \ell$), respectively. Since $\ell$ is good and $d$ is nice for $\ell$, 
it follows that 
\begin{eqnarray*}
\Pr[\Pi_{\ell, d}(A_1, \ldots, A_k, R) = (U,V)\ |\ S = 00, M > \kappa] & = & 1/4 + {\Omega}(1) - \TV(\psi_{n, 00, d}^\ell, \varphi_{n, d}^\ell) \\
& \ge & 1/4 + {\Omega}(1),
\end{eqnarray*} 
where the probability is taken over $(A_1, \ldots, A_k, R) \sim \psi_1$, and $\TV(\psi_{n, 00, d}^\ell, \varphi_{n, d}^\ell)$ is the total variation distance between distributions $\psi_{n, 00, d}^\ell$ and $\varphi_{n, d}^\ell$, which can be bounded by $O(1/\sqrt{n - \kappa}) = O(1/\sqrt{n})$. The proof will be given  shortly, and here
is where we use that $\kappa = n - \Omega(n)$. 

Hence, for a good $\ell$, and for a nice $d$ for $\ell$, we have
\begin{eqnarray}
 I(B_{[k],\ell} ; \Pi \ |\ D_{\ell}, M, S = 00, M > \kappa, D_{-\ell} = d) \ge I(A_1, \ldots, A_k ; \Pi'\ |\ R, M, S = 00, M > \kappa), \label{eq:b-1}
\end{eqnarray} 
where $\Pi'$ is the protocol that minimizes the information cost when the information (on the right side of (\ref{eq:b-1})) is measured with respect to the marginal distribution of $\psi_{n}$ on a good coordinate $\ell$, and succeeds in outputting $(U,V)$ with probability $1/4 + {\Omega}(1)$ when the $k$ sites and the predictor get input $\{A_1, \ldots, A_k, R\} \sim \psi_1$,  
The information on the left side of (\ref{eq:b-1}) is measured with respect to $\psi_{n}$. 

Combining (\ref{eq:a-3}) and (\ref{eq:b-1}), and given that we have $\Omega(n)$ good $\ell$, as well as at least an $\Omega(1)$ fraction of $d$ that are nice for any good $\ell$, we have
\begin{eqnarray}
\label{eq:b-2}
&& I(B ; \Pi\ |\ D, M, S = 00, M > \kappa) \ge {\Omega}(n) \cdot I(A_1, \ldots, A_k ; \Pi'\ |\ R, M, S = 00, M > \kappa).
\end{eqnarray} 

Now we analyze $\TV(\psi_{n, 00, d}^\ell, \varphi_{n, d}^\ell)$. First note that we can just focus on coordinates $\ell' > \ell$, since the distributions $\psi_{n, 00, d}^\ell$ and $\varphi_{n, d}^\ell$ are the same on coordinates $\ell' \le \ell$. Let $\varsigma_1$ and $\varsigma_2$ be the distribution of $\psi_{n, 00, d}^\ell$ and $\varphi_{n, d}^\ell$ on coordinates $\{\ell'\ |\ \ell' > \ell, \ell' \in [n]\}$, respectively. Observe that $\varsigma_2$ can be thought as a binomial distribution: for each coordinate $\ell' > \ell$, we set $B_{D_{\ell'}, \ell'}$ randomly to be $0$ or $1$ with equal probability. The remaining $B_{i, \ell'}\ (i \neq D_{\ell'})$ are all set to be $0$. Moreover, $\varsigma_1$ can be thought in the following way: we first sample according to $\varsigma_2$, and then randomly choose a coordinate $M > \kappa \ge \ell$ and reset $B_{D_M, M} = 0$. Since $M$ is random, the total variation distance between $\varsigma_1$ and $\varsigma_2$ is the total variation distance between Binomial$(n-\ell, 1/2)$ and Binomial$(n-\ell-1, 1/2)$ (that is, by symmetry, only the number of $1$'s in $\{B_{D_{\ell'}, \ell'}\ |\ \ell' > \ell\}$ matters), which is at most $O(1/\sqrt{n - \ell}) \le O(1/\sqrt{n - \kappa})$ (see, e.g., Fact 2.4 of \cite{GMRZ11}).

Let $\mathcal{E}$ be the event that all sites have the value $0$ in the $M$-th 
coordinate when the inputs are drawn from $\varphi_n$. 
Observe that $(\varphi_n | \mathcal{E}) = (\psi_{n}|S=00)$, thus
$$ I(B ; \Pi\ |\ D, M, S = 00, M > \kappa) \ge {\Omega}(n) \cdot I(A_1, \ldots, A_k ; \Pi'\ |\ R, M, M > \kappa, \mathcal{E}),$$
where the information on the left hand side is measured with respect to inputs $(B, \{D, M\})$ drawn from $\psi_{n}$, and the information on the right hand side is measured with respect to inputs $(A_1, \ldots, A_k, R)$ drawn from the marginal distribution of $\varphi_n$ on a good coordinate $\ell$, which is equivalent to $\varphi_1$ since $M > \kappa \ge \ell$. By the third item of Proposition~\ref{prop:mut}, and using that $A_1, \ldots, A_k$ are independent of $\mathcal{E}$ given $M > \kappa \ge \ell$ and $R$, we obtain
\begin{eqnarray*}
I(A_1, \ldots, A_k ; \Pi'\ |\ R, M, M > \kappa, \mathcal{E}) \ge I(A_1, \ldots, A_k\ ; \Pi'\ |\ R, M, M > \kappa).
\end{eqnarray*}
Finally, since $M, M > \kappa$ are independent of $A_1, \ldots, A_k$ and $R$, it holds that $I(A_1, \ldots, A_k \ ; \Pi'\ |\ R, M, M > \kappa) \ge I(A_1, \ldots, A_k ; \Pi'\ |\ R)$, where the information is measured with respect to the input distribution $\varphi_1$, and $\Pi'$ is a
protocol which succeeds with probability $1/4 +{\Omega}(1)$ on $\psi_1$. 

It remains to show that $I(A_1, \ldots, A_k ; \Pi'\ |\ R) = {\Omega}(1/k)$, where the information is measured with respect to $\varphi_1$, and the correctness is measured with respect to $\psi_1$. Let $\mathbf{0}$ be the all-$0$ vector, $\mathbf{1}$ be the all-$1$ vector and $\mathbf{e_i}$ be the standard basis vector with the $i$-th coordinate being $1$. By the relationship between mutual information
and Hellinger distance (see Proposition 2.51 and Proposition 2.53 of \cite{Bar-Yossef:02}), we have
\begin{eqnarray*}
I(A_1, \ldots, A_k ; \Pi'\ |\ R) & = & (1/k) \cdot \sum_{i \in [k]}  I(A_1, \ldots, A_k ; \Pi'\ |\ R = i) \\
& = & \textstyle \Omega(1/k) \cdot \sum_{i \in [k]} h^2(\Pi'(\mathbf{0}), \Pi'(\mathbf{e_i})),
\end{eqnarray*}
where $h(\cdot, \cdot)$ is the Hellinger distance (see Section \ref{sec:prelim} for a definition). 
Now we assume $k$ and $k/2$ are powers of $2$, and we use Theorem $7$ of \cite{Jayram09}, which says that the following three statements hold:
\begin{enumerate}
\item $\sum_{i \in [k]} h^2(\Pi'(\mathbf{0}), \Pi'(\mathbf{e_{i}})) = \Omega(1) \cdot h^2(\Pi'(\mathbf{0}), \Pi'(1^{k/2}0^{k/2}))$.

\item $\sum_{i \in [k]} h^2(\Pi'(\mathbf{0}), \Pi'(\mathbf{e_{i}})) = \Omega(1) \cdot h^2(\Pi'(\mathbf{0}), \Pi'(0^{k/2}1^{k/2}))$.

\item $\sum_{i \in [k]} h^2(\Pi'(\mathbf{0}), \Pi'(\mathbf{e_{i}})) = \Omega(1) \cdot h^2(\Pi'(\mathbf{0}), \Pi'(\mathbf{1}))$.
\end{enumerate}
It follows that
\begin{eqnarray*}
I(A_1, \ldots, A_k ; \Pi'\ |\ R) 
&=& \Omega(1/k) \cdot \left(h^2(\Pi'(\mathbf{0}), \Pi'(1^{k/2}0^{k/2})) \right. \\
&&\left. + h^2(\Pi'(\mathbf{0}), \Pi'(0^{k/2}1^{k/2}))+  h^2(\Pi'(\mathbf{0}), \Pi'(\mathbf{1}))\right).
\end{eqnarray*}
By the Cauchy-Schwartz inequality we have, 
\begin{eqnarray*}
I(A_1, \ldots, A_k ; \Pi'\ |\ R) &=& \Omega(1/k)  \cdot \left(h(\Pi'(\mathbf{0}), \Pi'(1^{k/2}0^{k/2}) ) \right. \\
&&\left.+ h(\Pi'(\mathbf{0}), \Pi'(0^{k/2}1^{k/2}) )+  h(\Pi'(\mathbf{0}), \Pi'(\mathbf{1}))\right)^2.
\end{eqnarray*}
We can rewrite this as (by changing the constant in the $\Omega(1/k)$):
\begin{eqnarray*}
I(A_1, \ldots, A_k ; \Pi'\ |\ R) 
&=& \Omega(1/k) \cdot \left (3h(\Pi'(\mathbf{0}), \Pi'(1^{k/2}0^{k/2}) ) \right. \\
&&\left.+ 3h(\Pi'(\mathbf{0}), \Pi'(0^{k/2}1^{k/2}) )+  3h(\Pi'(\mathbf{0}), \Pi'(\mathbf{1}))\right)^2.
\end{eqnarray*}
By the triangle inequality of the Hellinger distance, we get
\begin{enumerate}
\item $h(\Pi'(\mathbf{0}), \Pi'(\mathbf{1})) + h(\Pi'(\mathbf{0}), \Pi'(1^{k/2}0^{k/2}) ) \ge h(\Pi'(\mathbf{1}), \Pi'(1^{k/2}0^{k/2}))$,

\item $h(\Pi'(\mathbf{0}), \Pi'(\mathbf{1})) + h(\Pi'(\mathbf{0}), \Pi'(0^{k/2}1^{k/2}) ) \ge h(\Pi'(\mathbf{1}), \Pi'(0^{k/2}1^{k/2}))$,

\item $h(\Pi'(\mathbf{0}), \Pi'(0^{k/2}1^{k/2}) + h(\Pi'(\mathbf{0}), \Pi'(1^{k/2}0^{k/2}) ) \ge h(\Pi'(0^{k/2}1^{k/2}), \Pi'(1^{k/2}0^{k/2}))$.
\end{enumerate}
Thus we have
\begin{eqnarray*}
\label{eq:BTX-2}
\textstyle I(A_1, \ldots, A_k ; \Pi'\ |\ R)
&=& \textstyle \Omega(1/k) \cdot \left(\sum_{a,b \in \{\mathbf{0},\ \mathbf{1},\ 1^{k/2} 0^{k/2},\ 0^{k/2} 1^{k/2}\}} h(\Pi'(a), \Pi'(b))\right)^2.
\end{eqnarray*}
The claim is that at least one of $h(\Pi'(a), \Pi'(b))$ in the RHS in Equation~(\ref{eq:BTX-2}) is ${\Omega}(1)$, and this will complete the proof. By Proposition \ref{sec:varHell}, this is true if the total variation distance $\TV(\Pi'(a), \Pi'(b)) = {\Omega}(1)$ for a pair $(a,b) \in \{\mathbf{0}, \ \mathbf{1}, \ 1^{k/2}0^{k/2}, \ 0^{k/2}1^{k/2}\}^2$. We show that there must be such a pair $(a,b)$, for the following reasons. 

For a $z \in \{\mathbf{0}, \ \mathbf{1}, \ 1^{k/2}0^{k/2}, \ 0^{k/2}1^{k/2}\}$, let $\chi(z) = xy$ if $z = x^{k/2}y^{k/2}$. First, there must exist a pair $(a, b)$ such that
\begin{eqnarray} 
\label{eq:TV-1}
\TV(g(\Pi'(a), R), g(\Pi'(b), R)) = \Omega(1),
\end{eqnarray}
since otherwise, if $\TV(g(\Pi'(a), R), g(\Pi'(b), R)) = o(1)$ for all $(a,b) \in \{\mathbf{0}, \ \mathbf{1}, \ 1^{k/2}0^{k/2}, \ 0^{k/2}1^{k/2}\}^2$, then by Proposition~\ref{prop:diameter}, for any pair $(a, b)$ with $a \neq b$, and a $c$ chosen from $\{a, b\}$ uniformly at random, it holds that $\Pr[g(\Pi'(c), R) = \chi(c)] \le 1/2 + o(1)$, where the probability is taken over the distribution of $c$. Consequently, for a $c$ chosen from $\{\mathbf{0}, \ \mathbf{1}, \ 1^{k/2}0^{k/2}, \ 0^{k/2}1^{k/2}\}$ uniformly at random, it holds that $\Pr[g(\Pi'(c), R) = \chi(c)] \le 1/4 + o(1)$, violating the protocol's success probability guarantee. Second, since $g$ is a deterministic function, and $R$ is independent of $A_1, \ldots, A_k$ when $(A_1, \ldots, A_k, R) \sim \psi_1$, we have
\begin{eqnarray} 
\label{eq:TV-2}
\TV(\Pi'(a), \Pi'(b)) \ge \TV(g(\Pi'(a), R), g(\Pi'(b), R)).
\end{eqnarray}
The claim follows from (\ref{eq:TV-1}) and (\ref{eq:TV-2}).
\end{proof}

\subsection{The \thresh\ Problem}
The input of the \thresh\ problem is a concatenation of $1/\eps^2$ copies of inputs of the \XOR\ problem. That is, each site $S_i\ (i = 1, 2, \ldots, k)$ holds an input consisting of $1/\eps^2$ blocks each of which is an input for a site in the \XOR\ problem. More precisely, each $S_i\ (i \in [k])$ holds an input ${b_i} = ({b}_i^1,  \ldots, {b}_i^{1/\eps^2})$ where  ${b}_i^j = (b_{i,1}^j, \ldots, b_{i,n}^j)\ (j \in [1/\eps^2])$ is a vector of $n$ bits. Let $b^j = (b^j_1,  \ldots, b^j_k)$ be the list of inputs to the $k$ sites in the $j$-th block. Let ${b} = ({b_1},  \ldots, {b_k})$ be the list of inputs to the $k$ sites. In the \thresh\ problem the $k$ sites want to compute the following.
\begin{eqnarray*}
\begin{array}{l}
k\textrm{-BTX}({b}_1, \ldots, {b}_k)  =  \left\{
  \begin{array}{rl}
   1, & \text{if } \abs{\sum\limits_{j \in [1/\eps^2]} k\textrm{-XOR}({b}_1^j,  \ldots, {b}_k^j) - \frac{1}{2\eps^2}}  \ge 2/\eps,\\
   0, & \text{if }  \abs{\sum\limits_{j \in [1/\eps^2]} k\textrm{-XOR}({b}_1^j,  \ldots, {b}_k^j) - \frac{1}{2\eps^2}}  \le 1/\eps, \\
   *, & \text{otherwise}.
  \end{array}
  \right.
\end{array}
\end{eqnarray*}

We define the input distribution $\nu$ for the \thresh\ problem as follows: The input of the $k$ sites in each block is chosen independently according to the input distribution $\psi_n$, which is defined for the \XOR\ problem. Let ${B}, {B}_i, B^j, {B}_i^j, B_{i, \ell}^j$ be the corresponding random variables of ${b}, {b}_i, b^j, {b}_i^j, b_{i, \ell}^j$ when the input of \thresh\ is chosen according to the distribution $\nu$. Let ${D}^j = ({D}_1^j,  \ldots, {D}_{n}^j)$ where ${D}_\ell^j\ (\ell \in [n], j \in [1/\eps^2])$ is the special site in the $\ell$-th coordinate of block $j$, and let ${D} = ({D}^1, \ldots, {D}^{1/\eps^2})$. Let ${M} = (M^1, \ldots, M^{1/\eps^2})$ where $M^j$ is the special coordinate in block $j$. Let $S = (S^1, \ldots, S^{1/\eps^2})$ where $S^j \in \{00, 01, 10, 11\}$ is the type of the \XOR\ instance in block $j$. 

For each block $j\ (j \in [1/\eps^2])$, let $X^j = 1$ if the inputs of the first $k/2$ sites in the special coordinate $M^j$ are all $1$ and $X^j = 0$ otherwise; and similarly  let $Y^j = 1$ if the inputs of the second $k/2$ sites in the coordinate $M^j$ are all $1$ and $Y^j = 0$ otherwise. Let $X = (X^1, \ldots, X^{1/\eps^2})$ and $Y = (Y^1, \ldots, Y^{1/\eps^2})$. 


\paragraph{Linking \thresh\ to \gap.} 
We show that Alice and Bob, who are given $(X, Y) \sim \phi$, can construct a $2$-player protocol $\Pi'$ for \gap$(X, Y)$ using a protocol $\Pi$ for \thresh.

They first construct an input for \thresh\ using $(X, Y)$. Alice simulates the first $k/2$ players, and Bob simulates the second $k/2$ players. Alice and Bob use the public randomness to generate $M^j$ and $D^j$  for each $j \in [1/\eps^2]$. For each $j \in [1/\eps^2]$, Alice sets  the $M^j$-th coordinate of each of the first $k/2$ players to $X^j$. Similarly, Bob sets the $M^j$-th coordinate of each of the last $k/2$ players to $Y^j$. Alice and Bob then use private randomness and the $D^j$ vectors to fill in the remaining coordinates. Observe that the resulting inputs $B$ (for \thresh)  is distributed according to $\nu$.

Alice and Bob then run the protocol $\Pi$ on $B$. Every time a message is sent between any two of the $k$ players in $\Pi$, it is appended to the transcript. That is, if the two players are among the first $k/2$, Alice still forwards this message to Bob. If the two players are among the last $k/2$, Bob still forwards this message to Alice. If the message is between a player in the first group and the second group, Alice and Bob exchange a message. The output of $\Pi'$ is equal to that of $\Pi$. 

\begin{theorem} 
\label{thm:preBTX}
Let $\Pi$ be the transcript of any randomized protocol for \thresh\ on input distribution $\nu$ with error probability $\delta$ for a sufficiently small constant $\delta$.
Then $I(X, Y ; \Pi\ |\ M, D) = {\Omega}(1/\eps^2)$, where the information is measured with respect to the uniform distribution on $X,Y$. 
\end{theorem}
\begin{proof}
By a Markov inequality, we have that for at least $1/2$ fraction of choices of $(M, D) = (m, d)$, the $k$-party protocol $\Pi$ computes \thresh\ with error probability at most $2\delta$. Say such a pair $(m, d)$ {\em good}. According to our reduction, we have that the transcript of $\Pi$ is equal to the transcript of $\Pi'$ and the output of $\Pi'$ is the same as the output of that of $\Pi$. Hence, for a good pair $(m, d)$, the $2$-party protocol $\Pi'$ computes \gap\ with error probability at most $2\delta$ on distribution $\phi$. We have
\begin{eqnarray*}
I(X, Y ; \Pi\ |\ M, D) &= &\sum_{(m, d)} \Pr[(M, D) = (m, d)] \left[I(X, Y ; \Pi\ |\ (M, D) = (m, d))\right] \nonumber \\
& \ge &  \sum_{\text{good } (m, d)} \Pr[(M, D) = (m, d)]  \left[I(X, Y ; \Pi\ |\ (M, D) = (m, d))\right] \nonumber \\
& = &  \sum_{\text{good } (m, d)} \Pr[(M, D) = (m, d)] \left[I(X, Y ; \Pi'\ |\ (M, D) = (m, d))\right] \label{eq:z-2} \\
& \ge & 1/2 \cdot {\Omega}(1/\eps^2) \ge {\Omega}(1/\eps^2). \quad (\text{By Theorem~\ref{thm:gap}}) \nonumber
\end{eqnarray*}
\end{proof}

Now we are ready to prove our main theorem for \thresh.
\begin{theorem}
\label{thm:BTX}
Let $\Pi$ be the transcript of any randomized protocol for \thresh\ on input distribution $\nu$ with error probability $\delta$ for a sufficiently small constant $\delta$. We have $I(B; \Pi\ |\ M, D) \ge {\Omega}(n/(k\eps^2))$, where the information is measured with respect to the input distribution $\nu$. 
\end{theorem}

\iffull{
\begin{proof}
By Theorem~\ref{thm:preBTX} we have $I(X, Y; \Pi\ |\ M, D) = {\Omega}(1/\eps^2)$. Using the chain rule and a Markov inequality, it holds that 
$$I(X^j, Y^j ; \Pi\ |\ M, D, X^{<j}, Y^{<j}) = {\Omega}(1)$$ for at least ${\Omega}(1/\eps^2)$ of $j \in [1/\eps^2]$, where $X^{<j} = \{X^1, \ldots, X^{j-1}\}$ and similarly for $Y^{<j}$. We say such a $j$ for which this holds is {\em good}. 

Now we consider a good $j \in [1/\eps^2]$, and show that 
$$I(B^j; \Pi\ |\ M, D, B^{<j}) = {\Omega}(n/k).$$ Since $B^{<j}$ determines $(X^{<j}, Y^{<j})$ given $M$, and $B^{<j}$ is independent of $B^j$ given $M$ and $D$, by item $3$ of Proposition \ref{prop:mut}, it suffices to prove that $I(B^j; \Pi\ |\ M, D, X^{<j}, Y^{<j}) = {\Omega}(n/k)$. By expanding the conditioning, we can write $I(B^j; \Pi\ |\ M, D, X^{<j}, Y^{<j})$ as
\begin{eqnarray*}
\label{eq:F2-DS-1}
&&\sum_{(x, y, m, d)}\Pr\left[(M^{-j}, D^{-j}, X^{<j}, Y^{<j}) = (m, d, x, y)\right] \\
&&\quad \quad \quad \quad \times I\left(B^j ; \Pi\ |\ M^j, D^j, (M^{-j}, D^{-j}, X^{<j}, Y^{<j}) = (m, d, x, y)\right).
\end{eqnarray*}
By the definition of a good $j \in [1/\eps^2]$, we know by a Markov bound that with probability ${\Omega}(1)$ over the choice of $(x, y, m, d)$, we have 
\begin{equation*}
I(X^j, Y^j ; \Pi \ |\ M^j, D^j, (M^{-j}, D^{-j}, X^{<j}, Y^{<j}) = (m, d, x, y)) = {\Omega}(1).
\end{equation*}
Call these $(x, y, m, d)$ for which this holds {\em good} for $j$. 

Note that 
$H(X^j, Y^j \ |\  M^j, D^j, (M^{-j}, D^{-j}, X^{<j}, Y^{<j}) = (m, d, x, y)) = 2$, since $M, D, X^{<j}, Y^{<j}$ are independent of $X, Y$, Therefore, for a good $j$ and a tuple $(x, y, m, d)$ that is good for $j$, we have
\begin{eqnarray*}
&&H(X^j, Y^j \ |\  \Pi, M^j, D^j, (M^{-j}, D^{-j}, X^{<j}, Y^{<j}) = (m, d, x, y)) \\
&=& H(X^j, Y^j \ |\  M^j, D^j, (M^{-j}, D^{-j}, X^{<j}, Y^{<j}) = (m, d, x, y)) \\
&&- I(X^j, Y^j ; \Pi \ |\ M^j, D^j, (M^{-j}, D^{-j}, X^{<j}, Y^{<j}) = (m, d, x, y)) \\
& = & 2 - \Omega(1).
\end{eqnarray*}
By the Maximum Likelihood Principle in Proposition \ref{prop:mut}, the maximum likelihood
function $g$ computes $(X^j, Y^j)$ from the transcript of $\Pi$ and $M^j, D^j, (j, m, d, x, y)$,
with error probability $\delta_g$, over $X^j, Y^j, M^j, D^j$ and the randomness of $\Pi$, satisfying
\begin{eqnarray}
\label{eq:error}
1 - \delta_g \ge \frac{1}{2^{H(X^j, Y^j \ |\ \Pi, M^j, D^j, (M^{-j}, D^{-j}, X^{<j}, Y^{<j}) = (m, d, x, y))}} \ge \frac{1}{2^{2-{\Omega}(1)}} 
= \frac{1}{4} + {\Omega}(1),
\end{eqnarray}

Now for a good $j$, and a tuple $(x, y, m, d)$ that is good for $j$, we define a protocol $\Pi_{j, x, y, m, d}$ which computes the \guess\ problem on input $(A_1, \ldots, A_k, \{Q, R\}) \sim \psi_n$ correctly with probability $1/4 + \Omega(1)$. Here $A_1, \ldots, A_k$ are inputs of the $k$ sites and $\{Q, R\}$ is the input of the predictor.
The protocol $\Pi_{j, x, y, m, d}$ has $(j, x, y, m, d)$ hardwired into it, and works as follows. First, the $k$ sites construct an input $B$ for the \thresh\ problem distributed according to $\nu$, using $\{A_1, \ldots, A_k\}$, $(j, x, y, m, d)$ and their private randomness, without any communication: They set $B^j = \{A_1, \ldots, A_k\}$, and use their private randomness to sample inputs for blocks $j' \neq j$ using the values $(x, y, m, d)$ and the fact that the inputs to the $k$ sites are independent conditioned on $(x, y, m, d)$. The predictor sets its input to be $\{M^j = Q, D^j = R, (M^{-j}, D^{-j}, X^{<j}, Y^{<j}) = (m, d, x, y) \}$. Next, the $k$ sites run $\Pi$ on their input $B$. Finally, the predictor outputs $(X^j, Y^j) = g(\Pi(B), Q, R, j, m, d, x, y)$.

Combining these with Theorem~\ref{thm:XOR}, we obtain
\begin{eqnarray*}
\label{eq:BTX-1}
I(B; \Pi\ |\ M, D) & \ge & \sum_{\text{good } j} I(B^j; \Pi\ |\ M, D, B^{<j}) \\
& \ge & \sum_{\text{good } j} \sum_{\text{good } (x, y, m, d) \text{ for } j}\Pr\left[(M^{-j}, D^{-j}, X^{<j}, Y^{<j}) = (m, d, x, y)\right]\\
&& \quad \quad \quad \quad \times I\left(B^j ; \Pi\ |\ M^j, D^j, (M^{-j}, D^{-j}, X^{<j}, Y^{<j}) = (m, d, x, y)\right) \\
& = & \Omega(1/\eps^2) \cdot {\Omega}(n/k) \quad (\mbox{By (\ref{eq:error}) and Theorem~\ref{thm:XOR}}) \\
&\ge& {\Omega}(n/(k\eps^2)).
\end{eqnarray*}
This completes the proof. 
\end{proof}
}

\iffull{By Proposition~\ref{prop:icost}}\ifconf{By the fact} that says that the randomized communication complexity is always at least the conditional information cost\ifconf{ (see the preliminaries of the full version of this paper)}, we have the following immediate corollary. 
\begin{corollary}
\label{cor:BTX}
Any randomized protocol  that computes \thresh\ on input distribution $\nu$ with error probability $\delta$ for some sufficient small constant $\delta$ has communication complexity ${\Omega}(n/(k\eps^2))$.
\end{corollary}
\iffull{\subsection{The Complexity of $F_p\ (p > 1)$}}
\ifconf{{\bf The Complexity of $F_p\ (p > 1)$:}}
\label{sec:F2-reduction}
\iffull{The input of $\eps$-approximate $F_p\ (p > 1)$ is chosen to be the same as \thresh\ by setting $n = k^p$. That is, we choose $\{B_1, \ldots, B_k\}$ randomly according to distribution $\nu$. $B_i$ is the input vector for site $S_i$ consisting of $1/\eps^2$ blocks each having $n =k^p$ coordinates.} We prove the lower bound for $F_p$ by performing a reduction from \thresh.

\begin{lemma}
\label{lem:F2-reduction}
If there exists a protocol $\cal P'$ that computes a $(1 + \alpha\eps)$-approximate $F_p\ (p > 1)$ for a sufficiently small constant $\alpha$ on input distribution $\nu$ with communication complexity $C$ and error probability at most $\delta$, then there exists a protocol $\cal P$ for \thresh\ on input distribution $\nu$ with communication complexity $C$ and error probability at most $3\delta + \sigma$, where $\sigma$ is an arbitrarily small constant.
\end{lemma}

\iffull{
\begin{proof}
We pick a random input $B = \{B_1, \ldots, B_k\}$ from distribution $\nu$. Each coordinate (column) of $B$ represents an item. Thus we have a total of $1/\eps^2 \cdot k^p = k^p/\eps^2$ possible items, which we identify with the set $[k^p/\eps^2]$. If we view each input vector $B_i\ (i \in [k])$ as a set, then each site has a subset of $[k^p/\eps^2]$ corresponding to these $1$ bits. Let $W_0$ be the exact value of $F_p(B)$. $W_0$ can be written as the sum of three components:
\begin{eqnarray} 
W_0  & = & \left(\frac{k^p - 1}{2\eps^2} + Q\right) \cdot 1^p   +  \left(\frac{1}{2\eps^2} + U\right) \cdot (k/2)^p + \left(\frac{1}{4\eps^2} + V\right) \cdot k^p, \label{eq:F2-1}
\end{eqnarray} 
where $Q, U, V$ are random variables (it will be clear why we write it this way in what follows). The first term of the RHS of Equation~(\ref{eq:F2-1}) is the contribution of non-special coordinates across all blocks in each of which one site has $1$. The second term is the contribution of the special coordinates across all blocks in each of which $k/2$ sites have $1$. The third term is the contribution of the special coordinates across all blocks in each of which all $k$ sites have $1$. 

Note that \thresh$(B_1, \ldots, B_k)$ is $1$ if $\abs{U} \ge 2/\eps$ and $0$ if $\abs{U} \le 1/\eps$. Our goal is to use a protocol $\cal P'$ for $F_p$ to construct a protocol $\cal P$ for \thresh\ such that we can differentiate the two cases (i.e., $\abs{U} \ge 2/\eps$ or $\abs{U} \le 1/\eps$) with a very good probability. 

Given a random input $B$, let $W_1$ be the exact $F_p$-value on the first $k/2$ sites, and $W_2$ be the exact $F_p$-value on the second $k/2$ sites. That is, $W_1 = F_p(B_1,  \ldots, B_{k/2})$ and 
$W_2 = F_p(B_{k/2+1},  \ldots, B_{k})$. We have
\begin{eqnarray} 
\label{eq:F2-2}
W_1 + W_2 & = & \left(\frac{k^p - 1}{2\eps^2} + Q\right) \cdot 1^p 
 +  \left(\frac{1}{2\eps^2} + U\right) \cdot (k/2)^p + \left(\frac{1}{4\eps^2} + V\right) \cdot 2 \cdot (k/2)^p.
\end{eqnarray} 
By Equation~(\ref{eq:F2-1}) and (\ref{eq:F2-2}) we can cancel out $V$:
\begin{eqnarray} 
2^{p-1}(W_1 + W_2) - W_0 & = &  (2^{p-1} - 1) \left(\left(\frac{k^p - 1}{2\eps^2} + Q\right) + \left(\frac{1}{2\eps^2} + U\right) \cdot (k/2)^p\right). \label{eq:F2-3}
\end{eqnarray}
Let $\tilde{W}_0$, $\tilde{W}_1$ and $\tilde{W}_2$ be the estimated $W_0$, $W_1$ and $W_2$ obtained by running $\cal P'$ on the $k$ sites' inputs, the first $k/2$ sites' inputs and the second $k/2$ sites' inputs, respectively. Observe that $W_0 \le (2^p+1)k^p/\eps^2$ and $W_1, W_2 \le 2k^p/\eps^2$. By the randomized approximation guarantee of $\cal P'$  and the discussion above we have that with probability at least $1 - 3\delta$,
{
\begin{equation} 
\label{eq:F2-4}
2^{p-1}(W_1 + W_2) - W_0  =  2^{p-1}(\tilde{W}_1 + \tilde{W}_2) -  \tilde{W}_0 \pm \beta' k^p/\eps,
\end{equation}
}
where $\abs{\beta'} \le 3(2^p+1)\alpha$. 

By a Chernoff bound we have that $\abs{Q} \le c_1 k^{p/2}/\eps$ with probability at least $1 - \sigma$, where $\sigma$ is an arbitrarily small constant and $c_1 \le \kappa \log^{1/2}(1/\sigma)$ for some universal constant $\kappa$.  Combining this fact with Equation (\ref{eq:F2-3}) and (\ref{eq:F2-4}) and letting $\tilde{W} = (2^{p-1}(\tilde{W}_1 + \tilde{W}_2) - \tilde{W}_0)/(2^{p-1} - 1)$, we have that with probability at least $1 - 3\delta - \sigma$,
\begin{eqnarray} 
\label{eq:F2-6}
U & = &  \frac{2^p \tilde{W}}{k^p} - \frac{2^p+1}{2\eps^2} - \frac{2^p \beta}{(2^{p-1}-1) \eps},
\end{eqnarray}
where $\abs{\beta} \le 3(2^p+1)\alpha + o(1)$.

\paragraph{Protocol $\cal P$.}  Given an input $B$ for \thresh, protocol $\cal P$ first uses $\cal P'$ to obtain the value $\tilde{W}$ described above, and then determines the answer to \thresh\ as follows:
\begin{eqnarray*}
\begin{array}{l}
k\textrm{-BTX}(B)  =  \left\{
  \begin{array}{rl}
   1, & \text{if} \quad \abs{2^p \tilde{W}/k^p - (2^p+1)/(2\eps^2)} \ge 1.5/\eps,\\
   0, & \text{otherwise}.
  \end{array}
  \right.
\end{array}
\end{eqnarray*}

\paragraph{Correctness.} Note that with probability at least $1 - 3\delta - \sigma$, we have $\abs{\beta} \le 3(2^p+1)\alpha + o(1)$, where $\alpha> 0$ is a sufficiently small constant, and thus $\abs{\frac{2^p\beta}{(2^{p-1}-1)\eps}} < 0.5/\eps$. Therefore, in this case protocol $\cal P$ will always succeed. 
\end{proof}
}

Theorem~\ref{thm:BTX} (set $n = k^p$) and Lemma~\ref{lem:F2-reduction} directly imply the following main theorem for $F_p$.
\begin{theorem}
\label{thm:F2}
Any protocol that computes a $(1 + \eps)$-approximate $F_p\ (p > 1)$ on input distribution $\nu$ with error probability $\delta$ for some sufficiently small constant $\delta$ has communication complexity ${\Omega}(k^{p-1}/\eps^2)$.
\end{theorem}



\section{An Upper Bound for $F_p\ (p > 1)$}
\label{sec:F2-UB}

We describe the following protocol to give a factor $(1+\Theta(\eps))$-approximation to $F_p$ 
at all points in time in the union of $k$ streams each held by a different site. 
Each site has a non-negative vector $v^i \in \mathbb{R}^m$,~\footnote{We use $m$ instead of $N$ for universe size only in this section.} which evolves with time, 
and at all times the coordinator holds a $(1+\Theta(\eps))$-approximation to $\|\sum_{i=1}^k v^i\|_p^p$. 
Let $n$ be the length of the union of the $k$ streams. 
We assume $n = \poly(m)$, and that $k$ is a power of $2$.  

As observed in \cite{CMY11}, up to a factor of $O(\eps^{-1} \log n \log (\eps^{-1} \log n))$ in 
communication, the problem is equivalent 
to the threshold problem: given a threshold $\tau$, with probability  $2/3$:
when $\|\sum_{i=1}^k v^i\|_p^p > \tau$,
the coordinator outputs $1$, when $\|\sum_{i=1}^k v^i\|_p^p < \tau/(1+\eps)$, the coordinator
outputs $0$, and for $\tau/(1+\eps) \leq \|\sum_{i=1}^k v^i\|_p^p \leq \tau$, the coordinator 
can output either
$0$ or $1$\footnote{To see the equivalence, by independent repetition, we can assume the success
probability of the protocol for the threshold problem is $1-\Theta(\eps/\log n)$. Then
we can run a protocol for each
$\tau = 1, (1+\eps), (1+\eps)^2, (1+\eps)^3, \ldots, \Theta(n^2)$, 
and we are correct on all instantiations with 
probability at least $2/3$.}.

We can thus assume we are given a threshold $\tau$ in the following
algorithm description. For notational convenience, define $\tau_{\ell} = \tau/2^{\ell}$ 
for an integer $\ell$. A nice property of the algorithm is that it is one-way, 
namely, all communication is from the sites to the coordinator. 
\iffull{
We leave optimization of the $\poly(\eps^{-1} \log n)$ factors
in the communication complexity to future work. 
}
\iffull{\subsection{Our Protocol}}
\ifconf{\\\\{\bf Our Protocol:}}
The protocol consists of four algorithms illustrated in Algorithm~\ref{algo:interp} to Algorithm~\ref{algo:coord}.
%
Let $v = \sum_{i=1}^k v^i$ at any point in time during the 
union of the $k$ streams. At times we will make the following assumptions on the algorithm
parameters $\gamma, B,$ and $r$: we assume $\gamma = \Theta(\eps)$ 
is sufficiently small, and
 $B = \poly(\eps^{-1} \log n)$ and $r = \Theta(\log n)$ are sufficiently large. 
\begin{algorithm}[t!]
\caption{Intepretation of the random public coin by sites and the coordinator}
\label{algo:interp}
$r = \Theta(\log n)$ \tcc{A parameter used by the sites and coordinator}
\For{$z = 1, 2, \ldots, r$}
{
\For{$\ell = 0, 1, 2, \ldots, \log m$}
{Create a set $S^{z}_{\ell}$ by including each coordinate in $[m]$ 
independently with probability $2^{-\ell}$.}
}
\end{algorithm}
\begin{algorithm}[t!]
\caption{Initialization at Coordinator}
\label{algo:coord-init}
$\gamma = \Theta(\eps), B = \poly(\eps^{-1} \log n)$. Choose $\eta \in [0, 1]$
uniformly at random
\tcc{Parameters}
\For{$z = 1, 2, \ldots, r$}
{
\For{$\ell = 0, 1, 2, \ldots, \log m$}
{
\For{$j = 1, 2, \ldots, m$}
{$f_{z, \ell, j} \leftarrow 0$ \tcc{Initialize all frequencies seen to $0$}
}
}
}
$out \leftarrow 0$ \tcc{The coordinator's current output}
\end{algorithm}
\begin{algorithm}[t!]
\caption{When Site $i$ receives an update $v^i \leftarrow v^i + e_j$ for 
standard unit vector $e_j$}
\label{algo:site-process}
\For{$z = 1, 2, \ldots, r$}
{\For{$\ell = 0, 1, 2, \ldots, \log m$}
{\If{$j \in S_{\ell}^z$ and $v_j^i > \tau_{\ell}^{1/p} / (kB)$}
  {With probability $\min (B/\tau_{\ell}^{1/p}, 1)$, 
send $(j, z, \ell)$
to the coordinator}
}
}
\end{algorithm}
\begin{algorithm}[t!]
\caption{Algorithm at Coordinator if a tuple $(j, z, \ell)$ arrives}
\label{algo:coord}
$f_{z, \ell, j} \leftarrow f_{z, \ell, j} + \tau_{\ell}^{1/p}/B$\\
\For{$h = 0, 1, 2, \ldots, O(\gamma^{-1} \log(n/\eta^p))$}
{\For{$z = 1, 2, \ldots, r$}
{
Choose $\ell$ for which 
$2^{\ell} \leq \frac{\tau}{\eta^p (1+\gamma)^{ph} B} < 2^{\ell+1}$, or $\ell = 0$ if no such
$\ell$ exists\\
Let $F_{z, h} = \{j \in [m] \mid f_{z, \ell, j} \in [\eta (1+\gamma)^h, \eta (1+\gamma)^{h+1})\}$\\ 
}
$\tilde{c}_{h} = \textrm{median}_z \ 2^{\ell} \cdot |F_{z, h}|$ 
}
\If{$\sum_{h \geq 0} \tilde{c}_{h} \cdot \eta^p \cdot (1+\gamma)^{ph} > (1-\eps)\tau$}
{$out \leftarrow 1$\\
Terminate the protocol
}
\end{algorithm}

\iffull{\subsection{Communication Cost}}
\begin{lemma}\label{lem:comm}
Consider any setting of $v^1, \ldots, v^k$ for which we have
$\|\sum_{i=1}^k v^i\|_p^p \leq 2^p \cdot \tau.$ Then the expected total communication
is $k^{p-1} \cdot \poly(\eps^{-1} \log n)$ bits.
\end{lemma}

\begin{proof}
Fix any particular $z \in [r]$ and $\ell \in [0, 1, \ldots, \log m]$. Let $v_{j}^{i, \ell}$ equal $v_j^i$ if $j \in S_{\ell}$
and equal $0$ otherwise. Let $v^{i, \ell}$ be the vector with coordinates $v_{j}^{i, \ell}$
for $j \in [m]$. Also let $v^{\ell} = \sum_{i=1}^k v^{i, \ell}$. 
Observe that
${\bf E}[\|v^{\ell}\|_p^p] \leq 2^p \cdot \tau/2^{\ell} = 2^p \cdot \tau_{\ell}$. 

Because of non-negativity of the $v^i$, 
$$\sum_{i=1}^k \sum_{j \in S_{\ell}} (v_{j}^{i, \ell})^p \leq \sum_{i=1}^k \|v^{i, \ell}\|_p^p 
\leq \|v^{\ell}\|_p^p.$$
Notice that a $j \in S_{\ell}$ is sent by a site 
with probability at most $B/\tau_{\ell}^{1/p}$ and only if 
$(v_j^i)^p \geq \frac{\tau_{\ell}}{k^pB^p}$. Hence the expected number of messages sent
for this $z$ and $\ell$, over all randomness, is
\begin{equation}
 \frac{B}{\tau_{\ell}^{1/p}} 
{\bf E}\left[\sum_{i,j \ \mid \ (v_j^i)^p \geq \frac{\tau_{\ell}}{k^p B^p}} v_j^i \right]
\leq \frac{B}{\tau_{\ell}^{1/p}} \cdot 
\frac{{\bf E}[\|v^{\ell}\|_p^p]}{\tau_{\ell}/(k^p B^p)} \cdot \frac{\tau_{\ell}^{1/p}}{kB} 
\leq \frac{2^p \cdot \tau_{\ell} \cdot k^{p-1} \cdot B^p}{\tau_{\ell}}
= 2^p \cdot k^{p-1} \cdot B^p,
\end{equation}
where we used that $\sum v_j^i$ is maximized subject to $(v_j^i)^p \geq \frac{\tau_{\ell}}{k^p B^p}$
and $\sum (v_j^i)^p \leq \|v^{\ell}\|_p^p$ when all the $v_j^i$ are equal to $\tau_{\ell}^{1/p}/(kB)$. 
Summing over all $z$ and $\ell$, it follows that the expected number of messages sent
in total is $O(k^{p-1} B^p \log^2 n)$. Since each message is $O(\log n)$ bits,
the expected number of bits is $k^{p-1} \cdot \poly(\eps^{-1} \log n)$. 
%
\end{proof}
\iffull{\subsection{Correctness}}
\ifconf{\noindent {\bf Correctness:}}
\ifconf{Due to space constraints, we defer the whole proof of correctness to the full version.}

\iffull{
We let $C > 0$ be a sufficiently large constant. 

\subsubsection{Concentration of Individual Frequencies}
We shall make use of the following standard multiplicative Chernoff bound.
\begin{fact}\label{fact:chernoff}
Let $X_1, \ldots X_s$ be i.i.d. Bernoulli$(q)$ random variables. Then for all $0 < \beta < 1$,
$$\Pr \left [|\sum_{i=1}^s X_i - qs| \geq \beta q s \right ] \leq 2 \cdot e^{-\frac{\beta^2qs}{3}}.$$
\end{fact}

\begin{lemma}\label{lem:concentration}
For a sufficiently large constant $C > 0$, with probability $1-n^{-\Omega(C)}$,
for all $z$, $\ell$, $j \in S_{\ell}$, and all times in the union of
the $k$ streams,
\begin{enumerate}
\item $f_{z,\ell, j} \leq 2e \cdot v_j + \frac{C \tau_{\ell}^{1/p}\log n}{B}$, and
\item if $v_j \geq \frac{C (\log^5 n) \tau_{\ell}^{1/p}}{B\gamma^{10}}$, then 
$|f_{z, \ell, j} - v_j| \leq \frac{\gamma^5}{\log^2 n} \cdot v_j$. 
\end{enumerate}
\end{lemma}
\begin{proof}
Fix a particular time snapshot in the stream. 
Let $g_{z, \ell, j} = f_{z, \ell, j} \cdot B/\tau_{\ell}^{1/p}$. Then $g_{z, \ell, j}$
is a sum of indicator variables, where the number of indicator variables depends on the
values of the $v_j^i$. The indicator variables are independent, each with expectation 
$\min(B/\tau_{\ell}^{1/p}, 1)$. 

{\bf First part of lemma.}
The number $s$ of indicator variables is at most $v_j$,
and the expectation of each is at most $B/\tau_{\ell}^{1/p}$. Hence, the 
probability that $w = 2e \cdot v_j \cdot B/\tau_{\ell}^{1/p} + C \log n$ or more of them 
equal $1$ is at most
$${v_j \choose w} \cdot \left (\frac{B}{\tau_{\ell}^{1/p}} \right )^{w}
\leq \left (\frac{ev_j B}{w \tau_{\ell}^{1/p}} \right )^{w} \leq \left (\frac{1}{2} \right )^{C \log n} = n^{-C}.$$
This part of the lemma now follows by scaling the $g_{z,\ell,j}$ by $\tau_{\ell}^{1/p}/B$ to obtain a 
bound on the $f_{z, \ell, j}$. 

{\bf Second part of lemma.}
Suppose at this time $v_j \geq \frac{C(\log^5 n) \tau_{\ell}^{1/p}}{B\gamma^{10}}$.
 The number $s$ of indicator variables is minimized when there 
are $k-1$ distinct $i$ for which $v_j^i = \frac{\tau_{\ell}^{1/p}}{kB}$, and one value of $i$
for which 
$$v_j^i = v_j - (k-1) \cdot \frac{\tau_{\ell}^{1/p}}{kB}.$$
Hence, 
$$s \geq v_j - (k-1) \cdot \frac{\tau_{\ell}^{1/p}}{kB} - \frac{\tau_{\ell}^{1/p}}{kB}
= v_j - \frac{\tau_{\ell}^{1/p}}{B}.$$
 If the expectation is $1$, then
$f_{z, \ell, j} = v_j - \frac{\tau_{\ell}^{1/p}}{B}$, and using that 
$v_j \geq \frac{C (\log^5 n) \tau_{\ell}^{1/p}}{B\gamma^{10}}$ establishes this part of the lemma. 
Otherwise, 
applying Fact \ref{fact:chernoff} with 
$s \geq v_j - \frac{\tau_{\ell}^{1/p}}{B} \geq \frac{C (\log^5 n) \tau_{\ell}^{1/p}}{2B\gamma^{10}}$
and $q = \frac{B}{\tau_{\ell}^{1/p}}$, and using that 
$qs \geq \frac{C \log^5 n}{2\gamma^{10}}$, we have
$$\Pr \left [|g_{z, \ell, j} - qs| > \frac{\gamma^5 qs}{2 \log^2 n} \right ] = n^{-\Omega(C)}.$$
Scaling by $\frac{\tau_{\ell}^{1/p}}{B} = \frac{1}{q}$, we have
$$\Pr \left [|f_{s, \ell, j}-s| > \frac{\gamma^5 s}{2 \log^2 n} \right ] = n^{-\Omega(C)},$$
and since $v_j - \frac{\tau_{\ell}^{1/p}}{B} \leq s \leq v_j$,
$$\Pr \left [|f_{s, \ell, j}-v_j| \geq \frac{\gamma^5 v_j}{2\log^2 n} + \frac{\tau_{\ell}^{1/p}}{B} \right ] 
= n^{-\Omega(C)},$$
and finally using that $\frac{\tau_{\ell}^{1/p}}{B} < \frac{\gamma^5 v_j}{2\log^2 n}$, and union-bounding
over a stream of length $n$ as well as all choices of $z, \ell,$ and $j$, the lemma follows. 
\end{proof}

\subsubsection{Estimating Class Sizes}
Define the classes $C_h$ as follows:
$$C_h = \{j \in [m] \mid \eta (1+\gamma)^h \leq v_j < \eta(1+\gamma)^{h+1}\}.$$
Say that $C_h$ contributes at a point in time in the union of the $k$ streams if 
$$|C_h|\cdot  \eta^p (1+\gamma)^{ph} \geq \frac{\gamma \|v\|_p^p}{B^{1/2} \log (n/\eta^p)}.$$
Since the number of non-zero $|C_h|$ is $O(\gamma^{-1} \log (n/\eta^p))$, we have
\begin{eqnarray}\label{eqn:contribute}
\sum_{\textrm{ non-contributing }h} |C_h| \cdot \eta^p (1+\gamma)^{ph+p} 
= O \left ( \frac{\|v\|_p^p}{B^{1/2}} \right ).
\end{eqnarray}

\begin{lemma}\label{lem:easy}
With probability $1-n^{-\Omega(C)}$, at all points in time in the union of the $k$ streams
and for all $h$ and $\ell$, for at least a $3/5$ fraction of the $z \in [r]$,
$$|C_h \cap S_{\ell}^z| 
\leq 3 \cdot 2^{-\ell} \cdot |C_h|$$
\end{lemma}
\begin{proof}
The random variable $|C_h \cap S_{\ell}^z|$ is a sum of $|C_h|$ independent
Bernoulli$(2^{-\ell})$ random variables. By a Markov bound, 
$\Pr[|C_h \cap S_{\ell}^z| \leq 3 \cdot 2^{-\ell} |C_h|] \geq 2/3$. 
Letting $X_z$ be an indicator variable which is $1$ iff $|C_h \cap S_{\ell}^z| \leq 3 \cdot 2^{-\ell}|C_h|$,
the lemma follows by applying Fact \ref{fact:chernoff} to the $X_z$, using that $r$ is large enough, 
and union-bounding over a stream of length $n$ and all $h$ and $\ell$. 
\end{proof}
%
For a given $C_h$, let $\ell(h)$ be the value of $\ell$ for which we have
$2^{\ell} \leq \frac{\tau}{\eta^p (1+\gamma)^{ph}B} < 2^{\ell+1}$, or $\ell = 0$
if no such $\ell$ exists.

\begin{lemma}\label{lem:set2}
With probability $1-n^{-\Omega(C)}$, at all points in time in the union of the $k$ streams
and for all $h$, for at least a $3/5$ fraction of the $z \in [r]$, 
\begin{enumerate}
\item $2^{\ell(h)} \cdot |C_h \cap S_{\ell(h)}^z| \leq 3|C_h|,$ and
\item if at this time $C_h$ contributes and $\|v\|_p^p \geq \frac{\tau}{5}$, then 
$2^{\ell(h)} \cdot |C_h \cap S_{\ell(h)}^z| 
= \left (1 \pm \gamma \right) |C_h|.$
\end{enumerate}
\end{lemma}
\begin{proof}
We show this statement for a fixed $h$ and at a particular point in time in the union of the $k$ 
streams. The lemma will follow by a union bound. 

The first part of the lemma follows from Lemma \ref{lem:easy}.

We now prove the second part. In this case $\|v\|_p^p \geq \frac{\tau}{5}$. 
We can assume that there exists an $\ell$ for which 
$2^{\ell} \leq \frac{\tau}{\eta^p (1+\gamma)^{ph}B} < 2^{\ell+1}$. Indeed, otherwise $\ell(h) = 0$ and
$|C_h \cap S_{\ell(h)}^z| = |C_h|$ and the second part of the lemma follows. 

Let $q(z) = |C_h \cap S_{\ell(h)}^z|$,
which is a sum of independent 
indicator random variables and so ${\bf Var}[q(z)] \leq {\bf E}[q(z)]$. Also,
\begin{eqnarray}\label{eqn:firstw}
{\bf E}[q(z)] & = & 2^{-\ell} |C_h| \geq \frac{\eta^p (1+\gamma)^{ph}B}{\tau} \cdot |C_h|.
\end{eqnarray}
Since $C_h$ contributes, 
$|C_h| \cdot \eta^p \cdot (1+\gamma)^{ph} \geq \frac{\gamma \|v\|_p^p}{B^{1/2} \log (n/\eta^p)}$,
and combining this with (\ref{eqn:firstw}), 
$${\bf E}[q(z)] \geq \frac{B\gamma \|v\|_p^p}{B^{1/2}\tau \log (n/\eta^p)} 
\geq \frac{B^{1/2} \gamma}{5\log (n/\eta^p)}.$$
It follows that for $B$ sufficiently large, and assuming $\eta \geq 1/n^C$ which happens with probability
$1-1/n^C$, we have ${\bf E}[q(z)] \geq \frac{3}{\gamma^2}$, and so by Chebyshev's
inequality, 
$$\Pr \left [|q(z)-{\bf E}[q(z)]| \geq \gamma {\bf E}[q(z)] \right ] \leq 
\frac{{\bf Var}[q(z)]}{\gamma^2 \cdot {\bf E}^2[q(z)]} \leq \frac{1}{3}.$$ 
Since ${\bf E}[q(z)] = 2^{-\ell}|C_h|$, and $r = \Theta(\log n)$ is large enough, the lemma follows
by a Chernoff bound. 
%
\end{proof}

\subsubsection{Combining Individual Frequency Estimation and Class Size Estimation}
We define the set $T$ to be the set of times in the input stream for which the $F_p$-value
of the union of the $k$ streams first exceeds $(1+\gamma)^i$ for an $i$ satisfying
$$0 \leq i \leq \log_{(1+\gamma)} 2^p \cdot \tau.$$
\begin{lemma}\label{lem:set}
With probability $1-O(\gamma)$, for all times in $T$ and all $h$, 
\begin{enumerate}
\item $\tilde{c}_h \leq 3|C_h| + 3\gamma(2+\gamma)(|C_{h-1}| + |C_{h+1}|)$, and
\item if at this time $C_h$ contributes and $\|v\|_p^p \geq \frac{\tau}{5}$, then
$$(1-4\gamma)|C_h| \leq \tilde{c}_h \leq (1 + \gamma) |C_h| + 3\gamma(2+\gamma)(|C_{h-1}| + |C_{h+1}|).$$
\end{enumerate}
\end{lemma}
\begin{proof}
We assume the events of Lemma \ref{lem:concentration} and Lemma \ref{lem:set2} occur,
and we add $n^{-\Omega(C)}$ to the error probability. Let us fix a class $C_h$, 
a point in time in $T$, and a $z \in [r]$ which is among the at least 
$3r/5$ different $z$ that satisfy Lemma 
\ref{lem:set2} at this point in time.


By Lemma \ref{lem:concentration}, for any $j \in C_h \cap S_{\ell(h)}^z$ for which
$v_j \geq \frac{C (\log^5 n) \tau_{\ell(h)}^{1/p}}{B\gamma^{10}}$, if
\begin{eqnarray}\label{eqn:contain}
|\min(v_j - \eta (1+\gamma)^h, \eta (1+\gamma)^{h+1}-v_j)| \geq \frac{\gamma^5}{\log^2 n} \cdot v_j,
\end{eqnarray}
then $j \in F_{z,h}$. Let us first verify that for $j \in C_h$, we have 
$v_j  \geq \frac{C (\log^5 n) \tau_{\ell(h)}^{1/p}}{B\gamma^{10}}$. We have
\begin{eqnarray}\label{eqn:initialBound}
v_j^p \geq \eta^p (1+\gamma)^{ph} \geq \frac{\tau}{2^{\ell(h)+1}B}
\geq \frac{\tau_{\ell(h)}}{2B},
\end{eqnarray}
and so 
$$v_j \geq \left (\frac{\tau_{\ell(h)}}{2B} \right )^{1/p} \geq \frac{C (\log^5 n) \tau_{\ell(h)}^{1/p}}{B\gamma^{10}},$$
where the final inequality follows for  large enough $B = \poly(\eps^{-1}\log n)$ and $p > 1$. 

It remains to consider the case when (\ref{eqn:contain}) does not hold. 

Conditioned
on all other randomness,  $\eta \in [0, 1]$ is uniformly random subject
to $v_j \in C_h$, or equivalently,
$$\frac{v_j}{(1+\gamma)^{h+1}} < \eta \leq \frac{v_j}{(1+\gamma)^h}.$$ 
If (\ref{eqn:contain}) does not hold, then either
$$\frac{(1-\gamma^5/\log^2 n)v_j}{(1+\gamma)^h} \leq \eta, \textrm{   or   }
\eta \leq \frac{(1+\gamma^5/\log^2 n)v_j}{(1+\gamma)^{h+1}}.$$
Hence, the probability over $\eta$ that inequality (\ref{eqn:contain}) holds is at least
$$1 - \frac{\frac{\gamma^5 v_j}{(1+\gamma)^h \log^2 n} + \frac{\gamma^5 v_j}{(1+\gamma)^{h+1} \log^2 n}}
{\frac{v_j}{(1+\gamma)^h} - \frac{v_j}{(1+\gamma)^{h+1}}}
= 1 - \frac{\gamma^4 (2+\gamma)}{\log^2 n}.$$
It follows by a Markov bound that
\begin{eqnarray}\label{eqn:trueSet}
\Pr \left [|C_h \cap S_{\ell(h)}^z| \geq |C_h| \cdot \left (1-\gamma (2+\gamma) \right ) \right ]
\leq \frac{\gamma^3}{\log^2 n}.
\end{eqnarray}
Now we must consider the case that 
there is a $j' \in C_{h'} \cap S^z_{\ell(h)}$ for which $j' \in F_{z,h}$ for
an $h' \neq h$. There are two cases, namely, if 
$v_{j'} < \frac{C (\log^5 n) \tau_{\ell(h)}^{1/p}}{B\gamma^{10}}$ or if 
$v_{j'} \geq \frac{C (\log^5 n) \tau_{\ell(h)}^{1/p}}{B\gamma^{10}}$. 
We handle each case in turn.
\\\\
{\bf Case: $v_{j'} < \frac{C (\log^5 n) \tau_{\ell(h)}^{1/p}}{B\gamma^{10}}$.} Then
by Lemma \ref{lem:concentration}, 
$$f_{z, \ell(h), j'} \leq 2e \cdot v_{j'} + \frac{C \tau_{\ell(h)}^{1/p} \log n}{B}.$$
Therefore, it suffices to show that
$$ 2e \cdot \frac{C (\log^5 n) \tau_{\ell(h)}^{1/p}}{B\gamma^{10}} 
+ \frac{C \tau_{\ell(h)}^{1/p} \log n}{B} < \eta (1+\gamma)^h,$$
from which we can conclude that $j' \notin F_{z,h}$. But by (\ref{eqn:initialBound}),
$$\eta (1+\gamma)^h \geq \left (\frac{\tau_{\ell(h)}}{2B} \right )^{1/p}
>  2e \cdot \frac{C (\log^5 n) \tau_{\ell(h)}^{1/p}}{B\gamma^{10}} 
+ \frac{C \tau_{\ell(h)}^{1/p} \log n}{B},
$$
where the last inequality follows for large enough $B = \poly(\eps^{-1} \log n)$. 
Hence, $j' \notin F_{z,h}$. 
\\\\
{\bf Case: $v_{j'} \geq \frac{C (\log^5 n) \tau_{\ell(h)}^{1/p}}{B\gamma^{10}}$.}
We claim
that $h' \in \{h-1, h+1\}$. Indeed, by Lemma \ref{lem:concentration} we must have
$$\eta (1+\gamma)^h - \frac{\gamma^5}{\log^2 n} \cdot v_{j'} 
\leq v_{j'} \leq \eta (1+\gamma)^{h+1} + \frac{\gamma^5}{\log^2 n} \cdot v_{j'}.$$ 
This is equivalent to 
$$\frac{\eta(1+\gamma)^h}{1+\gamma^5/\log^2 n} \leq
v_{j'} \leq \frac{\eta (1+\gamma)^{h+1}}{1-\gamma^5/\log^2 n},$$
If $j' \in C_{h'}$ for $h' < h-1$, then
$$v_{j'} \leq \eta (1+\gamma)^{h-1} = \frac{\eta (1+\gamma)^h}{1+\gamma}
< \frac{\eta (1+\gamma)^h}{1+\gamma^5/\log^2 n},$$
which is impossible. Also, if $j' \in C_{h'}$ for $h' > h+1$, then
$$v_{j'} \geq \eta (1+\gamma)^{h+2} = \eta (1+\gamma)^{h+1} \cdot (1+\gamma)
> \frac{\eta (1+\gamma)^{h+1}}{1-\gamma^5/\log^2 n},$$
which is impossible. Hence, $h' \in \{h-1, h+1\}$. 

Let $N_{z,h} = F_{z,h} \setminus C_h$. Then 
\begin{eqnarray}\label{eqn:others}
{\bf E}[|N_{z,h}| \leq \frac{\gamma^4 (2+\gamma)}{\log^2 n} \cdot (|C_{h-1} \cap S_{\ell(h)}^z| 
+ |C_{h+1} \cap S_{\ell(h)}^z|).
\end{eqnarray}
By (\ref{eqn:trueSet}) and applying a Markov bound to (\ref{eqn:others}), together with 
a union bound, with probability $\geq 1-\frac{2\gamma^3}{\log^2 n}$,
\begin{eqnarray}\label{eqn:fzh}
(1-\gamma(2+\gamma)) \cdot |C_h \cap S_{\ell(h)}^z| \leq |F_{z,h}|
\end{eqnarray}
\begin{equation}
|F_{z,h}| \leq |C_h \cap S_{\ell(h)}^z|
+ \gamma (2+\gamma) \cdot (|C_{h-1} \cap S_{\ell(h)}^z| 
 + |C_{h+1} \cap S_{\ell(h)}^z|). \label{eqn:fzh2}
\end{equation}
By Lemma \ref{lem:easy},
\begin{equation}
2^{\ell(h)} |C_{h-1} \cap S_{\ell(h)}^z| \leq 3|C_{h-1}|  \textrm{    and    } \ 
2^{\ell(h)} |C_{h+1} \cap S_{\ell(h)}^z| \leq 3|C_{h+1}|. \label{eqn:abound}
\end{equation}
\\
{\bf First part of lemma.} 
At this point we can prove the first part of this lemma. 
By the first part of Lemma \ref{lem:set2}, 
\begin{eqnarray} \label{eqn:thish}
2^{\ell(h)} \cdot |C_h \cap S_{\ell(h)}^z| \leq 3|C_h|.
\end{eqnarray}
Combining (\ref{eqn:fzh2}), (\ref{eqn:abound}), and (\ref{eqn:thish}), 
we have with probability at least $1-\frac{2\gamma^{3}}{\log^2 n}-n^{-\Omega(C)}$,
$$2^{\ell(h)}|F_{z,h}| \leq 3|C_h| + 3\gamma(2+\gamma)(|C_{h-1}| + |C_{h+1}|).$$
Since this holds for at least $3r/5$ different $z$, it follows that
$$\tilde{c}_h \leq 3|C_h| + 3\gamma(2+\gamma)(|C_{h-1}| + |C_{h+1}|),$$
and the first part of the lemma follows by a union bound. Indeed, the number of 
$h$ is $O(\gamma^{-1}\log(n/\eta^p))$, which with probability $1-1/n$, say, is 
$O(\gamma^{-1} \log n)$ since with this probability $\eta^p \geq 1/n^p$.
Also, $|T| = O(\gamma^{-1} \log n)$. Hence, the probability this holds for all $h$ and all times
in $T$ is $1-O(\gamma)$. 
\\\\
{\bf Second part of the lemma.} Now we can prove the second part of the lemma. By the second
part of Lemma \ref{lem:set2}, if at this time $C_h$ contributes and $\|v\|_p^p \geq \frac{\tau}{5}$,
then 
\begin{eqnarray}\label{eqn:thish2}
2^{\ell(h)} \cdot |C_h \cap S_{\ell(h)}^z| = (1 \pm \gamma)|C_h|. 
\end{eqnarray}
Combining (\ref{eqn:fzh}), (\ref{eqn:fzh2}), (\ref{eqn:abound}), and (\ref{eqn:thish2}), 
we have with probability at least $1-\frac{2\gamma^3}{\log^2 n}-n^{-\Omega(C)}$,
\begin{equation*}
(1-\gamma(2+\gamma))(1-\gamma)|C_h| 
\leq 2^{\ell(h)}|F_{z,h}| \leq (1+\gamma)|C_h| + 3\gamma(2+\gamma)(|C_{h-1}| + |C_{h+1}|).
\end{equation*}
Since this holds for at least $3r/5$ different $z$, it follows that
\begin{equation*}
(1-\gamma(2+\gamma))(1-\gamma)|C_h| \leq \tilde{c}_h \leq 
(1+\gamma)|C_h| + 3\gamma(2+\gamma)(|C_{h-1}| + |C_{h+1}|).
\end{equation*}
and the second part of the lemma now follows by a union bound over all $h$ and all times
in $T$, exactly in the same way as the first part of the lemma. Note that
$1-4\gamma \leq (1-\gamma(2+\gamma))(1-\gamma)$ for small enough $\gamma = \Theta(\eps)$. 
\end{proof}
%

\subsubsection{Putting It All Together}

\begin{lemma}\label{lem:correct}
With probability at least $5/6$, at all times the coordinator's output is correct.
\end{lemma}
\begin{proof}
The coordinator outputs $0$ up until
the first point in time in the union of the $k$ streams for which 
$\sum_{h \geq 0} \tilde{c}_h \cdot \eta^p \cdot (1+\gamma)^{ph} > (1-\eps/2)\tau$.
It suffices to show that 
\begin{eqnarray}\label{eqn:mainThing}
\sum_{h \geq 0} \tilde{c}_h \eta^p (1+\gamma)^{ph} = (1 \pm \eps/2) \|v\|_p^p
\end{eqnarray} at
all times in the stream. We first show that with probability at least $5/6$, for all times in $T$, 
\begin{eqnarray}\label{eqn:intermediateThing}
\sum_{h \geq 0} \tilde{c}_h \eta^p (1+\gamma)^{ph} = (1 \pm \eps/4) \|v\|_p^p,
\end{eqnarray} 
and then use the structure
of $T$ and the protocol to argue that (\ref{eqn:mainThing}) holds at all times in the stream.

Fix a particular time in $T$. 
We condition on the event of Lemma \ref{lem:set}, which by setting $\gamma = \Theta(\eps)$ small enough,
can assume occurs with probability at least $5/6$. 

First, suppose at this point in time we have $\|v\|_p^p < \frac{\tau}{5}$. Then by Lemma \ref{lem:set},
for sufficiently small $\gamma = \Theta(\eps)$, we have
\begin{eqnarray*}
\sum_{h \geq 0} \tilde{c}_h \cdot \eta^p (1+\gamma)^{ph} 
& \leq & 
\sum_{h \geq 0} (3|C_h|+3\gamma(2+\gamma)(|C_{h-1}| + |C_{h+1}|)) \cdot \eta^p(1+\gamma)^{ph}\\
& \leq & \sum_{h \geq 0} \left (3 \sum_{j \in C_h} v_j^p 
+ 3 \gamma(2+\gamma) (1+\gamma)^2 \sum_{j \in C_{h-1} \cup C_{h+1}} v_j^p \right ) \\
& \leq & 4 \|v\|_p^p\\
& \leq & \frac{4\tau}{5}, 
\end{eqnarray*}
and so the coordinator will correctly output $0$, provided $\eps < \frac{1}{5}$. 

We now handle the case $\|v\|_p^p \geq \frac{\tau}{5}$. 
Then for all contributing $C_h$, we have
$$(1-4\gamma)|C_h| \leq \tilde{c}_h \leq (1 + \gamma) |C_h| + 3\gamma(2+\gamma)(|C_{h-1}| + |C_{h+1}|),$$
while for all $C_h$, we have
$$\tilde{c}_h \leq 3|C_h| + 3\gamma(2+\gamma)(|C_{h-1}| + |C_{h+1}|).$$
Hence, using (\ref{eqn:contribute}), 
\begin{eqnarray*}
\sum_{h \geq 0} \tilde{c}_h \cdot \eta^p (1+\gamma)^{ph}
& \geq & \sum_{\textrm{contributing } C_h} (1 - 4\gamma)|C_h| \eta^p (1+\gamma)^{ph}\\
& \geq & \frac{(1-4\gamma)}{(1+\gamma)^2} \sum_{\textrm{contributing } C_h} \sum_{j \in C_h} v_j^p\\
& \geq & (1-6\gamma) \cdot (1 - O(1/B^{1/2})) \cdot \|v\|_p^p.
\end{eqnarray*}
For the other direction,
\begin{eqnarray*}
&&\sum_{h \geq 0} \tilde{c}_h \cdot \eta^p(1+\gamma)^{ph} \\
& \leq & \sum_{\textrm{contributing } C_h} (1+\gamma)|C_h| \eta^p (1+\gamma)^{ph} + 
\sum_{\textrm{non-contributing } C_h} 3|C_h| \eta^p (1+\gamma)^{ph}\\
&& +  \sum_{h \geq 0} 3\gamma(2+\gamma)(|C_{h-1}| + |C_{h+1}|)\eta^p(1+\gamma)^{ph}\\
& \leq & (1+\gamma) \sum_{\textrm{contributing }C_h} \sum_{j \in C_h} v_j^p + O (1/B^{1/2}) \cdot \|v\|_p^p + O(\gamma) \cdot \|v\|_p^p\\
& \leq & (1+O(\gamma) + O(1/B^{1/2})) \|v\|_p^p.
\end{eqnarray*}
Hence, (\ref{eqn:intermediateThing}) follows for all times in $T$ provided that $\gamma = \Theta(\eps)$ is small enough and 
$B = \poly(\eps^{-1} \log n)$ is large enough. 

It remains to argue that (\ref{eqn:mainThing}) holds for all points in time in the union of the $k$ streams. 
Recall that each time in the union of the $k$ streams for which 
$\|v\|_p^p \geq (1+\gamma)^i$ for an integer $i$ is included in $T$, provided $\|v\|_p^p \leq 2^p\tau$.  

The key observation is that the quantity $\sum_{h \geq 0} \tilde{c}_h \eta^p (1+\gamma)^{ph}$ is non-decreasing, since
the values $|F_{z,h}|$ are non-decreasing. Now, the value of $\|v\|_p^p$ at a time $t$ not in $T$ is, by
definition of $T$, within a factor of $(1 \pm \gamma)$ of the value of $\|v\|_p^p$ for some time in $T$. Since
(\ref{eqn:intermediateThing}) holds for all times in $T$, it follows that the value of
 $\sum_{h \geq 0} \tilde{c}_h \eta^p (1+\gamma)^{ph}$ at time $t$ satisfies
$$(1-\gamma)(1-\eps/4)\|v\|_p^p \leq \sum_{h \geq 0} \tilde{c}_h \eta^p (1+\gamma)^{ph} \leq (1+\gamma)(1+\eps/4)\|v\|_p^p,$$
which implies for $\gamma = \Theta(\eps)$ small enough 
that (\ref{eqn:mainThing}) holds for all points in time in the union of the $k$ streams.
This completes the proof. 
\end{proof}
}

\begin{theorem}(MAIN)
With probability at least $2/3$, at all times the coordinator's output is correct
and the total communication is $k^{p-1} \cdot \poly(\eps^{-1} \log n)$ bits.
\end{theorem}
\iffull{
\begin{proof}
Consider the setting of $v^1, \ldots, v^k$ at the first time in the stream for which
$\|\sum_{i=1}^k v^i\|_p^p > \tau$. For any non-negative integer vector $w$ and any update $e_j$, we have
$\|w + e_j\|_p^p \leq (\|w\|_p+1)^p \leq 2^p \|w\|_p^p$. Since 
$\|\sum_{i=1}^k v^i\|_p^p$ is an integer and $\tau \geq 1$, we therefore
have $\|\sum_{i=1}^k v^i\|_p^p \leq 2^p \cdot \tau$. 
By Lemma \ref{lem:comm}, the expected communication for these $v^1, \ldots, v^k$
is $k^{p-1} \cdot \poly(\eps^{-1} \log n)$ bits, so with probability at least $5/6$ the communication
is $k^{p-1} \cdot \poly(\eps^{-1} \log n)$ bits. By Lemma \ref{lem:correct}, 
with probability at least $5/6$, the protocol terminates at or before the time
for which the inputs held by the players equal $v^1, \ldots, v^k$. The theorem 
follows by a union bound.
\end{proof}
}

%



\section{Related Problems}
\label{sec:app}
In this section we show that the techniques we have developed for distributed $F_0$ and $F_p\ (p > 1)$ can also be used to solve other fundamental problems. In particular, we consider the problems: all-quantile, heavy hitters, empirical entropy and $\ell_p$ for any $p >0$. For the first three problems, we are able to show that our lower bounds holds even if we allow some additive error $\eps$. From definitions below one can observe that lower bounds for additive $\eps$-approximations also hold for their multiplicative $(1+\eps)$-approximation counterparts. 

\subsection{The All-Quantile and Heavy Hitters}
\label{sec:quantile}
We first give the definitions of the problems. Given a multiset $A = \{a_1, a_2, \ldots, a_m\}$ where each $a_i$ is drawn from the universe $[N]$, let $f_i$ be the frequency of item $i$ in the set $A$. Thus $\sum_{i \in [N]}f_i = m$.
\begin{definition}{($\phi$-heavy hitters)}
For any $0 \le \phi \le 1$, the set
of $\phi$-heavy hitters of $A$ is $H_{\phi}(A) = \{x\ |\ f_x \ge \phi m \}$. If an $\eps$-approximation is allowed, then the returned set of heavy hitters must contain $H_{\phi}(A)$ and cannot include any $x$ such that $f_x < (\phi - \eps) m$. If $(\phi - \eps)m \le f_x < \phi m$, then $x$ may or may not be included in $H_{\phi}(A)$.
\end{definition}

\begin{definition}{($\phi$-quantile)}
For any $0 \le \phi \le 1$, the $\phi$-quantile of $A$ is some $x$ such that there are at most $\phi m$ items of $A$ that are smaller than $x$ and at most $(1 - \phi) m$ items of $A$ that are greater than $x$. If an $\eps$-approximation is allowed, then when asking for the $\phi$-quantile of $A$ we are allowed to return any $\phi'$-quantile of $A$ such that $\phi - \eps \le \phi' \le \phi + \eps$.
\end{definition}

\begin{definition}{(All-quantile)}
The $\eps$-approximate all-quantile (\quan) problem is defined in the coordinator model, where we have $k$ sites and a coordinator. Site $S_i\ (i \in [k])$ has a set $A_i$ of items. The $k$ sites want to communicate with the coordinator so that at the end of the process the coordinator can construct a data structure from which all $\eps$-approximate $\phi$-quantiles for any $0 \le \phi \le 1$ can be extracted. The cost is defined as the total number of bits exchanged between the coordinator and the $k$ sites.
\end{definition}

\begin{theorem}
\label{thm:hhquan}
Any randomized protocol that computes $\eps$-approximate \quan\ or  $\eps$-approximate $\min\{\frac{1}{2}, \frac{\eps\sqrt{k}}{2}\}$-heavy hitters with error probability $\delta$ for some sufficiently small constant $\delta$ has communication complexity $\Omega(\min\{\sqrt{k}/\eps, 1/\eps^2\})$ bits. 
\end{theorem}

\subsubsection{The \gm\ Problem}
Before proving Theorem~\ref{thm:hhquan}, we introduce a problem we call \gm. 

In this section we fix $\beta = 1/2$. In the \gm\ problem we have $k$ sites $S_1, S_2, \ldots, S_k$, and each site has a bit $Z_i\ (1 \le i \le k)$ such that $\Pr[Z_i = 0] = \Pr[Z_i = 1] = 1/2$. Let $\eth$ be the distribution of $\{Z_1, \ldots, Z_k\}$. The sites 
want to compute the following function.
\begin{eqnarray*}
\begin{array}{l}
k\mbox{-GAP-MAJ}(Z_1, \ldots, Z_k)  =  \left\{
  \begin{array}{rl}
   0, & \text{if } \sum_{i \in [k]} Z_i \le \beta k -\sqrt{\beta k},\\
   1, & \text{if } \sum_{i \in [k]} Z_i \ge \beta k + \sqrt{\beta k},\\
   *, & \text{otherwise,}
  \end{array}
  \right.
\end{array}
\end{eqnarray*}
where $*$ means that the answer can be either $0$ or $1$.

Notice that \gm\ is very similar to \bit: We set $\beta = 1/2$ and directly assign $Z_i$'s to the $k$ sites. Also, instead of approximating the sum, we just want to decide whether the sum is large or small, up to a gap which is roughly equal to the standard deviation.

We will prove the following theorem for \gm.
\begin{theorem}
\label{thm:gm}
Let $\Pi$ be the transcript of any private randomness protocol for \gm\ on input distribution $\eth$ with error probability $\delta$ for some sufficiently small constant $\delta$, then $I(Z; \Pi) = \Omega(k)$.
\end{theorem}
\begin{remark}
The theorem holds for private randomness protocols, though for our applications, we only need it to hold for deterministic protocols. Allowing private randomness could be useful when the theorem is used for direct-sum types of arguments in other settings.
\end{remark}

The following definition is essentially the same as Definition~\ref{def:normal}, but with a different setting  of parameters. For convenience, we will still use the terms {\em rare} and {\em normal}. Let $\kappa_1$ be a constant chosen later. For a transcript $\pi$, we define $q_i^\pi = \Pr_{\eth'}[Z_i = 1\ |\ \Pi = \pi]$. Thus $\sum_{i \in [k]} q_i^\pi = \E_{\eth'}[\sum_{i \in [k]} Z_i\ |\ \Pi = \pi]$.
\begin{definition}
\label{def:normal}
We say a transcript $\pi$ is {\em rare$^+$} if $\sum_{i \in [k]} q_i^\pi \ge \beta k + \kappa_1 \sqrt{\beta k}$ and {\em rare$^-$} if $\sum_{i \in [k]} q_i^\pi  \le \beta k - \kappa_1 \sqrt{\beta k}$. In both cases we say $\pi$ is {\em rare}. Otherwise we say it is {\em normal}.
\end{definition}

Let $\eth'$ be the joint distribution of $\eth$ and the distribution of $\Pi$'s private randomness. The following lemma is essentially the same as Lemma~\ref{lem:pi-normal}. For completeness we still include a proof.

\begin{lemma}
\label{lem:pi-normal-2}
Under the assumption of Theorem~\ref{thm:gm},
$\Pr_{\eth'}[\Pi \textrm{ is normal}] \ge 1 - 8e^{-\kappa_1^2/4}$.
\end{lemma}

\begin{proof}
Set $\kappa_2 = \kappa_1 - 1$. We will redefine the term {\em joker} which was defined in Definition~\ref{def:joker}, with a different setting of parameters. We say $Z = \{Z_1, Z_2, \ldots, Z_k\}$ is a {\em joker$^+$} if $\sum_{i \in [k]}Z_i\ge \beta k + \kappa_2 \sqrt{\beta k}$, and a {\em joker$^-$} if $\sum_{i \in [k]}Z_i\le \beta k - \kappa_2 \sqrt{\beta k}$. In both cases we say $Z$ is a {\em joker}. 

First, we can apply a Chernoff bound on random variables $Z_1, \ldots, Z_k$, and obtain 
$$\textstyle \Pr_\eth[Z \textrm{ is a joker}^+] = \Pr_\eth\left[\sum_{i \in [k]}Z_i\ge \beta k + \kappa_2\sqrt{\beta k} \right] \le e^{-\kappa_2^2/3}.$$ 

Second, by Observation \ref{obs:first}, we can 
apply a Chernoff bound on random variables $Z_1, \ldots, Z_k$ conditioned on $\Pi$ being bad$^+$,
\begin{eqnarray*} 
&& \Pr_{\eth'}[Z \textrm{ is a joker}^+\ |\ \Pi \textrm{ is bad}^+]  \\
&\ge& \sum_{\pi}\Pr_{\eth'}\left[\Pi = \pi\ |\ \Pi \textrm{ is bad}^+\right] \Pr_{\eth'} \left[Z \textrm{ is a joker}^+\ |\ \Pi = \pi, \Pi \textrm{ is bad}^+\right]\\
&=&   \sum_{\pi}\Pr_{\eth'}\left[\Pi = \pi\ |\ \Pi \textrm{ is bad}^+\right] \textstyle \Pr_{\eth'}\left[ \left. \sum_{i \in [k]}Z_i\ge \beta k + \kappa_2\sqrt{\beta k}\ \right|  \sum_{i \in [k]} q_i^\pi \ge \beta k + \kappa_1\sqrt{\beta k}, \Pi = \pi \right] \\
&\ge&  \sum_{\pi}\Pr_{\eth'}\left[\Pi = \pi\ |\ \Pi \textrm{ is bad}^+\right] \left(1 - e^{-(\kappa_1 - \kappa_2)^2/3}\right) \\
&=&  \left(1 - e^{-(\kappa_1 - \kappa_2)^2/3}\right).
\end{eqnarray*}

Finally by Bayes' theorem, we have that 
\begin{eqnarray*} 
 \Pr_{\eth'}[\Pi \textrm{ is bad}^+]
& = & \frac{\Pr_{\eth}[Z \textrm{ is a joker}^+] \cdot \Pr_{\eth'}[\Pi \textrm{ is bad}^+\ |\ Z \textrm{ is a joker}^+]}{\Pr_{\eth'}[Z \textrm{ is a joker}^+\ |\ \Pi \textrm{ is bad}^+]} \\
&\le& \frac{e^{-\kappa_2^2/3}}{1 - e^{-(\kappa_1 - \kappa_2)^2/3}}.
\end{eqnarray*} 
By symmetry (since we have set $\beta = 1/2$), we can also show that $$\Pr_{\eth'}[\Pi \textrm{ is bad}^-] \le {e^{-\kappa_2^2/3}}/{(1 - e^{-(\kappa_1 - \kappa_2)^2/3})}.$$ Therefore $\Pr_{\eth'}[\Pi \textrm{ is bad}] \le 2{e^{-(\kappa_1-1)^2/3}}/{(1 - e^{-1/3})} \le 8e^{-\kappa_1^2/4}$ (recall that we have set $\kappa_2 = \kappa_1 - 1$).
\end{proof}

The following definition is essentially the same as Definition~\ref{def:pi-strong}, but is for private randomness protocols. 
\begin{definition}
\label{def:pi-strong-2}
We say a transcript $\pi$ is {\em weak} if $\sum_{i \in [k]}q_i^\pi(1 - q_i^\pi) \ge  \beta k / (40 c_0)$ (for a sufficiently large constant $c_0$), and {\em strong} otherwise.
\end{definition}

The following lemma is similar to Lemma~\ref{lem:pi-strong}, but with the new definition of a {\em normal} $\pi$.

\begin{lemma}
\label{lem:pi-strong-2}
Under the assumption of Theorem~\ref{thm:gm}, $\Pr_{\eth'}[\Pi \textrm{ is normal and strong}] \ge 0.98$.
\end{lemma}

\begin{proof}
We first show that for a normal and weak transcript $\pi$, there exists a universal constant $c$ such that
\begin{eqnarray*}
\textstyle \Pr_{\eth'}\left[ \left. \sum_{i \in [k]}Z_i\le \beta k - \kappa_1 \sqrt{\beta k} \ \right|\ \Pi = \pi \right]  
&\ge&  c \cdot e^{- 60c_0 \kappa_1^2},\\ 
  \text{and \quad}  \textstyle \Pr_{\eth'}\left[\left. \sum_{i \in [k]}Z_i\ge \beta k + \kappa_1 \sqrt{\beta k} \ \right|\ \Pi = \pi \right] 
&\ge& c \cdot e^{- 60c_0 \kappa_1^2}.
\end{eqnarray*}

We only need to prove the first inequality. The second will follow by symmetry, since we set $\beta = 1/2$.
For a good and weak $\Pi = \pi$, we have 
\begin{eqnarray*}
\var_{\eth'}\left(\sum_{i \in [k]} Z_i\ |\ \Pi = \pi\right) = \sum_{i \in [k]} \var_{\eth'}(Z_i\ |\ \Pi = \pi) 
\ge \beta k / (40c_0).
\end{eqnarray*}
Set $\kappa_3 = 2\kappa_1$. By Fact~\ref{lem:feller} we have for a universal constant $c$,
\begin{eqnarray*}
\label{eq:bit-2}
&& \textstyle \Pr_{\eth'}\left[\left. \sum_{i \in [k]} Z_i \ge \sum_{i \in [k]} q_i^\pi + \kappa_3 \sqrt{\beta k}\ \right|\ \Pi = \pi \right] \\
&\ge& c \cdot e^{-\frac{(\kappa_3  \sqrt{\beta k})^2}{3 \cdot \beta k / (40c_0)}} 
\ge c \cdot e^{- 60c_0 \kappa_1^2}.
\end{eqnarray*}
Together with the fact that $\pi$ is good, we obtain
\begin{eqnarray*}
&&\textstyle \Pr_{\eth'}\left[\left. \sum_{i \in [k]}Z_i\ge \beta k + \kappa_1 \sqrt{\beta k} \ \right|\ \Pi = \pi \right] \\
& = & \textstyle \Pr_{\eth'}\left[\left. \sum_{i \in [k]}Z_i\ge \beta k + (\kappa_3-\kappa_1) \sqrt{\beta k} \ \right|\ \Pi = \pi \right] \\
& \ge & \textstyle \Pr_{\eth'}\left[\left. \sum_{i \in [k]} Z_i - \sum_{i \in [k]} q_i^\pi
\ge \kappa_3 \sqrt{\beta k}\ \right|\ \Pi = \pi \right] \\
&\ge& c \cdot e^{- 60c_0 \kappa_1^2}.
\end{eqnarray*}

Now set $\kappa_1 = \sqrt{4\ln 800}$. Suppose conditioned on $\Pi$ being normal, it is weak with probability more than $0.01$. Then the error probability of the protocol (taken over the distribution $\eth'$) is at least 
$$(1 - 8e^{-\kappa_1^2/4}) \cdot 0.01 \cdot c \cdot e^{- 60c_0 \kappa_1^2} \ge \delta,$$
for a sufficiently small constant $\delta$, violating the success guarantee of Theorem~\ref{thm:gm}. Therefore with probability at least 
$$(1 - 8e^{-\kappa_1^2/4}) \cdot (1 - 0.99) \ge 0.98,$$
$\Pi$ is both normal and strong.
\end{proof}


\begin{proof} (for Theorem~\ref{thm:gm})
Recall that for a transcript $\pi$, we have defined $q_i^\pi = \Pr_{\eth'}[Z_i = 1\ |\ \Pi = \pi]$. Let $p_i^\pi = \min\{q_i^\pi, 1 - q_i^\pi\}$, thus $p_i^\pi \in [0, 1/2]$. We will omit the superscript $\pi$ in $p_i^\pi$ when it is clear from the context. 

For a strong $\pi$, we have 
$\sum_{i \in [k]} p_i \cdot 1/2 \le \sum_{i \in [k]} p_i (1 - p_i) < \frac{\beta k}{40c_0}$. Thus $\sum_{i \in [k]} p_i < \frac{\beta k}{20c_0}$.
For each $p_i$, if $p_i < \frac{\beta}{16c_0} < 1/2$ (for a sufficiently large constant $c_0$), then $H_b(p_i)  < H_b\left(\frac{\beta}{16c_0}\right)$ (recall that $H_b(\cdot)$ is the binary entropy function). Otherwise, it holds that $H_b(p_i) = p_i \log\frac{1}{p_i} + (1 - p_i) \log\frac{1}{1 - p_i} \le p_i \log\frac{16c_0}{\beta} + 2(1 - p_i) p_i$, since $\log\frac{1}{1 - p_i} = p_i + p_i^2/2 + p_i^3/3 + \ldots \le 2p_i$ if $p_i \le 1/2$. Thus we have
\begin{eqnarray*}
&&\sum_{i \in [k]} H(Z_i\ |\ \Pi = \pi) 
 =  \sum_{i \in [k]} H_b(p_i) \label{eq:z-1}\\
& \le & \sum_{i \in [k]}\max \left\{H_b\left(\frac{\beta}{16c_0}\right),\ p_i \log\frac{16c_0}{\beta} + 2(1 - p_i) p_i  \right\}\\
& \le & \max \left\{k \cdot H_b\left(\frac{\beta}{16c_0}\right),\ \sum_{i \in [k]} p_i \cdot \log\frac{16c_0}{\beta} + 2\sum_{i \in [k]} p_i (1 - p_i) \right\} \\
& < & \max \left\{k \cdot H_b\left(\frac{\beta}{16c_0}\right),\ \frac{\beta k}{20c_0} \cdot \log\frac{16c_0}{\beta} + \frac{\beta k}{20c_0} \right\}  \\
&=& k \cdot H_b\left(\frac{\beta}{16c_0}\right).
\end{eqnarray*}
Therefore, if $\Pr_{\eth'}[\Pi \textrm{ is strong}] \ge 0.98$, then
\begin{eqnarray*}
&& \sum_{i \in [k]} H(Z_i\ |\ \Pi) \\
&=& \sum_{\pi : \pi \textrm{ is strong}} \left( \Pr_{\eth'}[\Pi = \pi] \sum_{i \in [k]} H(Z_i\ |\ \Pi = \pi) \right) + \sum_{\pi : \pi \textrm{ is weak}}  \left( \Pr_{\eth'}[\Pi = \pi] \sum_{i \in [k]} H(Z_i\ |\ \Pi = \pi) \right) \\
&\le& 0.98 \cdot k  \cdot H_b\left(\frac{\beta}{16c_0}\right) + 0.02 \cdot k \cdot 1\\
&\le& k \cdot H_b\left(\frac{\beta}{16c_0}\right)  + 0.02k.
\end{eqnarray*}
Now, for a large enough constant $c_0$ and $\beta = 1/2$,
$$I(Z; \Pi) \ge \sum_{i \in [k]}I(Z_i; \Pi) \ge \sum_{i \in [k]}(H(Z_i) - H(Z_i\ |\ \Pi)) \ge k \cdot H_b(\beta) - k \cdot H_b\left(\frac{\beta}{16c_0}\right) - 0.02k \ge \Omega(k). $$
\end{proof} 

\subsubsection{Proof of Theorem~\ref{thm:hhquan}}
\begin{proof}
We first prove the theorem for \quan.
In the case that $k \ge 1/\eps^2$, we prove an $\Omega(1/\eps^2)$ information complexity lower bound. We prove this by a simple reduction from \gm. We can assume $k = 1/\eps^2$ since if $k > 1/\eps^2$ then we can just give inputs to the first $1/\eps^2$ sites. Set $\beta = 1/2$. Given a random input $Z_1, Z_2, \ldots, Z_k$ of \gm\ chosen from distribution $\eth$, we simply give $Z_i$ to site $S_i$. It is easy to observe that a protocol that computes $\eps/2$-approximate \quan\ on $A = \{Z_1, Z_2, \ldots, Z_k\}$ with error probability $\delta$ also computes \gm\ on input distribution $\eth$ with error probability $\delta$, since the answer to \gm\ is simply the answer to $\frac{1}{2}$-quantile. The $\Omega(1/\eps^2)$ lower bound follows from Theorem~\ref{thm:gm}.

In the case that $k < 1/\eps^2$, we prove an $\Omega(\sqrt{k}/\eps)$ information complexity lower bound. We again perform a reduction from \gm. Set $\beta = 1/2$. The reduction works as follows. We are given $\ell = 1/(\eps\sqrt{k})$ independent copies of \gm\ with $Z^1, Z^2, \ldots, Z^\ell$ being the inputs, where $Z^i = \{Z^i_1, Z^i_2, \ldots, Z^i_k\} \in \{0,1\}^k$ is chosen from distribution $\eth$. We construct an input for \quan\ by giving the $j$-th site the item set $A_j = \{Z^1_j, 2 + Z^2_j, 4 + Z^3_j, \ldots, 2(\ell - 1) + Z^\ell_j\}$. It is not difficult to observe that a protocol that computes $\eps/2$-approximate \quan\ on the set $A = \{A_1, A_2, \ldots, A_j\}$ with error probability $\delta$ also computes the answer to each copy of \gm\ on distribution $\eth$  with error probability $\delta$, simply by returning $(X_i - 2(i - 1))$ for the $i$-th copy of \gm, where $X_i$ is the $\eps/2$-approximate $\frac{i - 1/2}{\ell}$-quantile.

On the other hand, any protocol that computes each of the $\ell$ independent copies of \gm\ correctly with error probability $\delta$ for a sufficiently small constant $\delta$ has information complexity $\Omega(\sqrt{k}/\eps)$. This is simply because for any transcript $\Pi$, by Theorem~\ref{thm:gm}, independence and the chain rule we have that
\begin{equation}
I(Z^1, Z^2, \ldots, Z^\ell; \Pi) \ge \sum_{i \in [\ell]} I(Z^i; \Pi) \ge \Omega(\ell k) \ge \Omega(\sqrt{k}/\eps).
\end{equation}

The proof for heavy hitters is done by essentially the same reduction as that for \quan. In the case that $k = 1/\eps^2$ (or $k \ge 1/\eps^2$ in general), a protocol that computes $\eps/2$-approximate $\frac{1}{2}$-heavy hitters on $A = \{Z_1, Z_2, \ldots, Z_k\}$ with error probability $\delta$ also computes \gm\ on input distribution $\eth$ with error probability $\delta$. In the case that $k < 1/\eps^2$, it also holds that a protocol that computes $\eps/2$-approximate $\frac{\eps\sqrt{k}}{2}$-heavy hitters on the set $A = \{A_1, A_2, \ldots, A_j\}$ where $A_j = \{Z^1_j, 2 + Z^2_j, 4 + Z^3_j, \ldots, 2(\ell - 1) + Z^\ell_j\}$ with error probability $\delta$ also computes the answer to each copy of \gm\ on distribution $\eth$  with error probability $\delta$.
\end{proof}

\subsection{Entropy Estimation}
We are given a set $A = \{(e_1, a_1), (e_2, a_2), \ldots, (e_m, a_m)\}$ where each $e_k\ (k \in [m])$ is drawn from the universe $[N]$, and $a_k \in \{+1, -1\}$ denotes an insertion or a deletion of item $e_k$. The entropy estimation problem (\entropy) asks for the value
$H(A) = \sum_{j \in [N]}(\abs{f_j}/L) \log(L/\abs{f_j})$ where $f_j = \sum_{k : e_k = j} a_k$ and $L = \sum_{j \in [N]} \abs{f_j}$. In the $\eps$-approximate \entropy\ problem, the items in the set $A$ are distributed among $k$ sites who want to compute a value $\tilde{H}(A)$ for which $\abs{\tilde{H}(A) - H(A)} \le \eps$. In this section we prove the following theorem.

\begin{theorem}
\label{thm:entropy}
Any randomized protocol that computes $\eps$-approximate \entropy\  with error probability at most $\delta$ for some sufficiently small constant $\delta$ has communication complexity ${\Omega}(k/\eps^2)$.
\end{theorem}

\begin{proof}
As with $F_2$, we prove the lower bound for the \entropy\ problem by a reduction from \thresh. Given a random input $B$ for \thresh\ according to distribution $\nu$ with $n = \gamma^2 k^2$ for some parameter $\gamma = \log^{-d} (k/\eps)$ for large enough constant $d$, we construct an input for \entropy\ as follows. Each block $j \in [1/\eps^2]$ in \thresh\ corresponds to one coordinate item $e_j$ in the vector for \entropy; so we have in total $1/\eps^2$ items in the entropy vector. The $k$ sites first use shared randomness to sample $\gamma^2 k^2/\eps^2$ random $\pm 1$ values for each coordinate across all blocks in $B$~\footnote{By Newman's theorem (cf. ~\cite{eyal97:_commun_}, Chapter 3) we can get rid of the public randomness by increasing the total communication complexity by no more than an additive $O(\log(\gamma k/\eps))$ factor which is negligible in our proof.}. Let $\{R_{1}^1, R_{1}^2, \ldots, R_{\gamma^2 k^2}^ {1/\eps^2}\}$ be these random $\pm 1$ values. Each site looks at each of its bits $B_{i,\ell}^j\ (i \in [k], \ell \in \gamma^2 k^2, j \in [1/\eps^2])$, and generates an item $(e_j, R_{\ell}^j)$ (recall that $R_{\ell}^j$ denotes insertion or deletion of the item $e_j$) if $B_{i,\ell}^j = 1$. Call the resulting input distribution $\nu'$.

We call an item in group $G_P$ if the \XOR\ instance in the corresponding block is a $00$-instance; and  in group $G_Q$ if it is a $11$-instance; in group $G_U$ if it is a $01$-instance or a $10$-instance. Group $G_U$ is further divided to two subgroups $G_{U_1}$ and $G_{U_2}$, containing all $10$-instance and all $01$-instance, respectively. Let $P, Q, U, U_1, U_2)$ be the cardinalities of these groups. Now we consider the frequency of each item type.
\begin{enumerate}
\item For an item $e_j \in G_P$, its frequency $f_j$ is distributed as follows: we choose a value $i$ from the binomial distribution on $n$ values each with probability $1/2$,
then we take the sum $\kappa_j$ of $i$ i.i.d. $\pm 1$ random variables. We can thus write $|f_j| = |\kappa_j|$. 
%
%

\item For an item $e_j \in G_Q$, its frequency $f_j$ is distributed as follows: we choose a value $i$ from the binomial distribution on $n$ values each with probability $1/2$, then we take the sum $\kappa_j$ of $i$ i.i.d. $\pm 1$ random variables. Then we add the value $R_{\ell^*}^j \cdot k$, where $\ell^*$ is the index of the special column in block $j$. We can thus
write $|f_j|$ as $|k + R_{\ell^{*}}^j \cdot \kappa_j|$. 
%
By a Chernoff-Hoeffding bound, with probability $1 - 2e^{-\lambda^2/2}$, we have  $\abs{\kappa_j} \le \lambda \gamma k$. We choose $\lambda = \log (k/\eps)$, and thus $\lambda \gamma = o(1)$. Therefore $\kappa_j$ will not affect the sign of $f_j$ for any $j$ (by a union bound) and we can write $\abs{f_j} = k + R_{\ell^*}^j \cdot \kappa_j$. Since $\kappa_j$ is symmetric about $0$ and $R_{\ell^*}^j$ is a random $\pm 1$ variable, we can simply drop $R_{\ell^*}^j$ and write $\abs{f_j} = k + \kappa_j$.

\item For an item $e_j \in G_U$, its frequency $f_j$ is distributed as follows: we choose a value $i$ from the binomial distribution on $n$ values each with probability $1/2$,
then we take the sum $\kappa_j$ of $i$ i.i.d. $\pm 1$ random variables. Then we add the value $R_{\ell^*}^j \cdot k/2$, where $\ell^*$ is the index of the special column in block $j$. We can thus 
write $|f_j|$ as $|k/2 + R_{\ell^*}^j \cdot \kappa_j$. 
As in the previous case, with probability $1 - 2e^{-\lambda^2/2}$, $\kappa_j$ will not affect the sign of $f_j$ and we can write $\abs{f_j} = k/2 + \kappa_j$.
\end{enumerate}
By a union bound, with error probability at most $\delta_1 = 1/\eps^2 \cdot 2e^{-\lambda^2/2}  = 1 - o(1)$, each $\kappa_j\ (e_j \in G_Q \cup G_U)$ will not affect the sign  of the corresponding $f_j$. Moreover, by another Chernoff bound we have that with error probability $\delta_2 = 10e^{-c_0^2/3}$, $P, Q, U_1, U_2$ are equal to $1/4\eps^2 \pm c_0/\eps$, and $U = 1/(2\eps^2) \pm c_0/\eps$. Here $\delta_2$ can be sufficiently small if we set constant $c_0$ sufficiently large. Thus we have that with arbitrary small constant error $\delta_0 = \delta_1 + \delta_2$, all the concentration results claimed above hold. For simplicity we neglect this part of error since it can be arbitrarily small and will not affect any of the analysis. In the rest of this section we will ignore arbitrarily small errors and drop some lower order terms as long as such operations will not affect any the analysis.

The analysis of the next part is similar to that for our $F_2$ lower bound, where we end up computing $F_2$ on three different vectors. Let us calculate $H_0, H_1$ and $H_2$, which stand for the entropies of all $k$-sites, the first $k/2$ sites and the second $k/2$ sites, respectively. Then we show that using $H_0, H_1$ and $H_2$ we can estimate $U$ well, and thus compute \thresh\ correctly with an arbitrarily small constant error. Thus if there is a protocol for \entropy\ on distribution $\nu'$ then we obtain a protocol for \thresh\ on distribution $\nu$ with the same communication complexity, completing the reduction and consequently proving Theorem~\ref{thm:entropy}.

Before computing $H_0, H_1$ and $H_2$, we first compute the total number $L$ of items. We can write 
\begin{eqnarray}
L  & = &  \sum_{e_j \in G_P} \abs{f_j} + \sum_{e_j \in G_Q} \abs{f_j} + \sum_{e_j \in G_U} \abs{f_j} \nonumber \\
\label{eq:entropy-L-1} & = & Q \cdot k + U \cdot k/2 + \sum_{e_j \in G_p} \abs{\kappa_j} +   \sum_{e_j \in G_Q \cup G_U} \kappa_j.
\end{eqnarray}
The absolute value of the fourth term in (\ref{eq:entropy-L-1}) can be bounded by $O(\gamma k / \eps)$ with arbitrarily large constant probability, using a Chernoff-Hoeffding bound, which will be $o(\eps L)$ and thus can be dropped. For the third term, by Chebyshev's inequality we can assume
(by increasing the constant in the big-Oh) that with arbitrarily large constant probability, 
$\sum_{e_j \in G_p} |\kappa_j| = (1 \pm \eps) \cdot 1/(4\eps^2) \cdot {\bf E}[|\kappa_j|]$, where ${\bf E}[|\kappa_j|] = \Theta(\gamma k)$ follows by approximating the binomial
distribution by a normal distribution (or, e.g., Khintchine's inequality). Let $z_1 = {\bf E}[|\kappa_j|]$ be a value which can be computed exactly. Then, 
$\sum_{e_j \in G_p} |\kappa_j| = 1/(4\eps^2) \cdot z_1 \pm O(\gamma k/\eps) = z_1/(4\eps^2) \pm o(\eps L)$, and so we can drop the additive $o(\eps L)$ term. 

Finally, we get,
\begin{eqnarray}
L  & = &  Q \cdot k + U \cdot k/2 + r_1
\label{eq:entropy-L}
\end{eqnarray}
where $r_1 = z_1/(4\eps^2) = \Theta(\gamma k /\eps^2)$ is a value that can be computed by any site without any comunication. 

Let $p_j = \abs{f_j} / L\ (j \in [1/\eps^2])$. We can write $H$ as follows.
\begin{eqnarray}
H_0 & = & \sum_{e_j \in P} p_j \log (1/p_j) + \sum_{e_j \in Q} p_j \log (1/p_j) + \sum_{e_j \in U} p_j \log (1/p_j).  \nonumber \\
& = & \log L - S / L \label{eq:entropy-H-0}
\end{eqnarray}
where 
\begin{eqnarray}
S & = & \sum_{e_j \in P} \abs{f_j} \log \abs{f_j} + \sum_{e_j \in Q} \abs{f_j} \log \abs{f_j} + \sum_{e_j \in U} \abs{f_j} \log \abs{f_j} \label{eq:entropy-H-1}
\end{eqnarray}
We consider the three summands in (\ref{eq:entropy-H-1}) one by one. For the second term in (\ref{eq:entropy-H-1}), we have
\begin{eqnarray}
\sum_{e_j \in Q} \abs{f_j} \log \abs{f_j}  & = & \sum_{e_j \in Q} (k + \kappa_j) \log  (k + \kappa_j) \nonumber \\
& = & Q \cdot k \log k + k \sum_{e_j \in Q} \log(1 + \kappa_j / k) \nonumber \\
& = & Q \cdot k \log k \pm O(\sum_{e_j \in Q} \kappa_j). \label{eq:entropy-Q-1}
\end{eqnarray}
The second term in (\ref{eq:entropy-Q-1}) is at most $o(\eps Q \cdot \gamma k) = o(k/eps)$, and can be dropped. By a similar analysis we can obtain that the third term in (\ref{eq:entropy-H-1}) is (up to an $o(k/\eps)$ term)
\begin{eqnarray}
\sum_{e_j \in U} \abs{f_j} \log \abs{f_j} = U \cdot (k/2) \log(k/2). \label{eq:entropy-U-1}
\end{eqnarray}
Now consider the first term. We have 
\begin{eqnarray}
\sum_{e_j \in P} \abs{f_j} \log \abs{f_j}  & = &  \sum_{e_j \in P} \abs{\kappa_j} \log \abs{\kappa_j} \nonumber \\
& = & (1/4\eps^2 \pm O(1/\eps)) \cdot \E[\abs{\kappa_j} \log \abs{\kappa_j}] \nonumber \\
& = & (1/4\eps^2) \cdot \E[\abs{\kappa_j} \log \abs{\kappa_j}]   \pm O(1/\eps) \cdot \E[\abs{\kappa_j} \log \abs{\kappa_j}] \label{eq:entropy-P-1}
\end{eqnarray}
where $z_2 =  \E[\abs{\kappa_j} \log \abs{\kappa_j}] = O(\gamma k \log k)$ can be computed exactly. Then the second term in (\ref{eq:entropy-P-1}) is at most $O(\gamma k \log k/\eps) = o(k/\eps)$, and thus can be dropped. Let $r_2 = (1/4\eps^2) \cdot z_2 = O(\gamma k \log k / \eps^2)$. By Equations (\ref{eq:entropy-H-0}), (\ref{eq:entropy-L}), (\ref{eq:entropy-H-1}), (\ref{eq:entropy-Q-1}), (\ref{eq:entropy-U-1}), (\ref{eq:entropy-P-1}) we can write
\begin{eqnarray}
H_0 & = & \log(Q \cdot k + U \cdot k/2 + r_1) + \frac{Q \cdot k\log k + U \cdot (k/2) \log (k/2) + r_2}{Q \cdot k + U \cdot k/2 + r_1}.  \label{eq:entropy-H-2}
\end{eqnarray}
Let $U = 1/(2\eps^2) + U'$ and $Q = 1/4\eps^2 + Q'$, and thus $U' = O(1/\eps)$ and $Q' = O(1/\eps)$. Next we convert the RHS of (\ref{eq:entropy-H-2}) to a linear function of $U'$ and $Q'$.
\begin{eqnarray}
H_0 & = & \log(k/(2\eps^2) + Q' k + U' k/2 + r_1) + \frac{\left(\frac{1}{4\eps^2} + Q'\right) \cdot k\log k + \left(\frac{1}{(2\eps^2)} + U'\right) \cdot \frac{k}{2}\log\frac{k}{2} + r_2}{k/(2\eps^2) + Q' k + U' k/2 + r_1}  \label{eq:entropy-H-4} \\
& = & \log(k/(2\eps^2) + r_1) + \log\left(1 + (U'  + 2 Q')\frac{\eps^2}{1 + 2\eps^2 r_1/k}\right) \nonumber \\
&& +\ \left(\left(\frac{1}{4\eps^2} + Q'\right) \cdot k\log k + \left(\frac{1}{(2\eps^2)} + U'\right) \cdot \frac{k}{2}\log\frac{k}{2} + r_2 \right) \cdot  \frac{2\eps^2}{k + 2 \eps^2 r_1} \cdot  \nonumber \\ 
&&\left(1 -  (U'  + 2 Q')\frac{\eps^2}{1 + 2\eps^2 r_1/k}\right)  \nonumber \\ 
&& \pm\ O(\eps^2)  \label{eq:entropy-H-5} \\
& = & \log(k/(2\eps^2) + r_1) + (U'  + 2 Q')\frac{\eps^2}{1 + 2\eps^2 r_1/k} \nonumber \\
&& +\ \left(\left(\frac{k (2\log{k} - 1)}{4\eps^2} + r_2 \right) + Q' \cdot k\log k + U' \cdot \frac{k}{2}\log\frac{k}{2} \right) \cdot  \frac{2\eps^2}{k + 2 \eps^2 r_1} \cdot  \nonumber \\ 
&&\quad \left(1 -  U' \cdot \frac{\eps^2}{1 + 2\eps^2 r_1/k} -  Q' \cdot \frac{2\eps^2}{1 + 2\eps^2 r_1/k}\right)  \nonumber \\ 
&& \pm\ O(\eps^2)  \label{eq:entropy-H-51} \\
& = & \alpha_1 + \alpha_2 U' + \alpha_3 Q', \label{eq:entropy-H-6}
\end{eqnarray} 
up to $o(\eps)$ factors (see below for discussion), 
where 
\begin{eqnarray} 
\alpha_1 & = & \log(k/(2\eps^2) + r_1) + \left(\frac{k (2\log{k} - 1)}{4\eps^2} + r_2 \right) \cdot \frac{2\eps^2}{k + 2 \eps^2 r_1},  \nonumber \\ 
\alpha_2 & = & \frac{\eps^2}{1 + 2\eps^2 r_1/k} + \left(\frac{k}{2}\log\frac{k}{2} -  \left(\frac{k (2\log{k} - 1)}{4\eps^2} + r_2 \right) \cdot \frac{\eps^2}{1 + 2\eps^2 r_1/k} \right) \cdot \frac{2\eps^2}{k + 2 \eps^2 r_1}, \nonumber \\ 
\alpha_3 & = & \frac{2\eps^2}{1 + 2\eps^2 r_1/k} + \left(k \log k -  \left(\frac{k (2\log{k} - 1)}{4\eps^2} + r_2 \right) \cdot \frac{2\eps^2}{1 + 2\eps^2 r_1/k} \right) \cdot \frac{2\eps^2}{k + 2 \eps^2 r_1}, \label{eq:entropy-H-a}
\end{eqnarray}



From (\ref{eq:entropy-H-4}) to (\ref{eq:entropy-H-5}) we use the fact that $1/(1+\eps) = 1 - \eps + O(\eps^2)$.  From (\ref{eq:entropy-H-5}) to (\ref{eq:entropy-H-51}) we use the fact that $\log(1 + \eps) = \eps + O(\eps^2)$. From (\ref{eq:entropy-H-51}) to (\ref{eq:entropy-H-6}) we use the fact that all terms in the form of $U' Q', {U'}^2, {Q'}^2$ are at most $\pm o(\eps)$ (we are assuming $O(\eps^2 \log k) = o(\eps)$ which is fine since we are neglecting polylog$(N)$ factors), therefore we can drop all of them together with the other $\pm O(\eps^2)$ terms, and consequently obtain a linear function on $U'$ and $Q'$.

Next we calculate $H_1$, and the calculation of $H_2$ will be exactly the same. The  values $t_1, t_2$ used in the following expressions are essentially the same as $r_1, r_2$ used for calculating $H_0$, with $t_1 = \Theta(\gamma k)$ and $t_2 = O(\gamma k \log k / \eps^2)$. Set $U'_1 = U_1 - 1/4\eps^2$ and $U'_2 = U_2 - 1/4\eps^2$.
\begin{eqnarray}
H_1 & = & \log((Q + U_1) k/2 + t_1) + \frac{Q \cdot \frac{k}{2} \log \frac{k}{2} + U_1 \cdot \frac{k}{2} \log \frac{k}{2} + t_2}{(Q + U_1) k/2 + t_1} \nonumber \\
& = & \log(k/(4\eps^2) + (Q' + U'_1) k/2 + t_1) + \frac{Q' \cdot \frac{k}{2} \log \frac{k}{2} + U'_1 \cdot \frac{k}{2} \log \frac{k}{2} + \frac{k\log(k/2)}{4\eps^2} + t_2}{k/(4\eps^2) + (Q' + U'_1) k/2 + t_1} \nonumber \\
& = & \log(k/(4\eps^2) + t_1) + (Q' + U'_1) \frac{2\eps^2}{1 + 4\eps^2 t_1/k} \nonumber \\
&& + \left( Q' \cdot \frac{k}{2} \log \frac{k}{2} + U'_1 \cdot \frac{k}{2} \log \frac{k}{2} + \frac{k\log(k/2)}{4\eps^2} + t_2 \right) \cdot \frac{4\eps^2}{k + 4\eps^2 t_1}  \cdot \left( 1 -   \frac{Q' \cdot 2\eps^2+ U'_1 \cdot 2\eps^2}{1 + 4\eps^2 t_1/k}\right)  \nonumber \\
& = & \beta_1 + \beta_2 U'_1 + \beta_3 Q'.   \label{eq:entropy-H-65}
\end{eqnarray}
where
\begin{eqnarray} 
\beta_1 & = &  \log(k/(4\eps^2) + t_1) + \left(\frac{k\log(k/2)}{4\eps^2} + t_2 \right) \cdot \frac{4\eps^2}{k + 4\eps^2 t_1}, \nonumber \\ 
\beta_2 & = &  \frac{2\eps^2}{1 + 4\eps^2 t_1/k} + \left( \frac{k}{2} \log \frac{k}{2} - \left(\frac{k\log(k/2)}{4\eps^2} + t_2\right) \frac{2\eps^2}{1 + 4\eps^2 t_1/k} \right) \frac{4\eps^2}{k + 4\eps^2 t_1}, \nonumber \\ 
\beta_3 & = & \frac{2\eps^2}{1 + 4\eps^2 t_1/k} + \left( \frac{k}{2} \log \frac{k}{2} - \left(\frac{k\log(k/2)}{4\eps^2} + t_2\right) \frac{2\eps^2}{1 + 4\eps^2 t_1/k} \right) \frac{4\eps^2}{k + 4\eps^2 t_1}.  \label{eq:entropy-H-b}
\end{eqnarray}
By the same calculation we can obtain the following equation for $H_2$.
\begin{eqnarray}
H_2 & = & \beta_1 + \beta_2 U'_2 + \beta_3 Q'.  \label{eq:entropy-H-7}
\end{eqnarray}
Note that $U' = U'_1 + U'_2$. Combining (\ref{eq:entropy-H-65}) and (\ref{eq:entropy-H-7}) we have 
\begin{eqnarray}
H_1 + H_2 & = & 2\beta_1 + \beta_2 U' + 2\beta_3 Q'.  \label{eq:entropy-H-9}
\end{eqnarray}
It is easy to verify that Equations (\ref{eq:entropy-H-6}) and (\ref{eq:entropy-H-9}) are linearly independent: by direct calculation (notice that $r_1, r_2$ are lower order terms) we obtain $\alpha_2 = \frac{\eps^2}{2} (1 \pm o(1))$ and $\alpha_3 = 3\eps^2 (1 \pm o(1)$. Therefore $\alpha_2 / \alpha_3 = (1 \pm o(1))/6$. Similarly we can obtain $\beta_2 / 2\beta_3 = (1 \pm o(1))/2$. Therefore the two equations are linearly independent. Furthermore, we can compute all the coefficients $\alpha_1, \alpha_2, \alpha_3, \beta_1, \beta_2, \beta_3$ up to a $(1 \pm o(\eps))$ factor. Thus if we have $\sigma \eps$ additive approximations of $H_0, H_1, H_2$ for a sufficient small constant $\sigma$, then we can estimate $U'$ (and thus $U$) up to an additive error of $\sigma' /\eps$ for a sufficiently small constant $\sigma'$ by Equation (\ref{eq:entropy-H-6}) and (\ref{eq:entropy-H-9}), and therefore \thresh. This completes the proof. 
\end{proof}

\subsection{$\ell_p$ for any constant $p \geq 1$}
Consider an $n$-dimensional vector $x$ with integer entries. 
It is well-known that for a vector $v$ of $n$ i.i.d. $N(0,1)$ random variables 
that $\langle v, x \rangle \sim N(0, \|x\|_2^2)$. Hence, for any real $p > 0$, 
${\bf E}[|\langle v, x \rangle|^p] = \|x\|_2^p G_p$, where $G_p > 0$ is the $p$-th moment of the standard half-normal 
distribution (see \cite{gaussian} 
for a formula for these moments in terms of confluent hypergeometric functions). 
Let $r = O(\eps^{-2})$, and $v^1, \ldots, v^r$ be independent $n$-dimensional 
vectors of i.i.d. $N(0,1)$ random variables. Let $y_j = \langle v^j, x \rangle / G_p^{1/p}$, so that 
$y = (y_1, \ldots, y_r)$. By Chebyshev's inequality for $r = O(\eps^{-2})$ sufficiently large, 
$\|y\|_p^p = (1 \pm \eps/3)\|x\|_2^p$ with probability at least $1-c$ for an arbitrarily small constant $c > 0$. 

We thus have the following reduction which shows that estimating $\ell_p$ up to a $(1+\eps)$-factor 
requires communication complexity ${\Omega}(k/\eps^2)$ for any $p > 0$. Let the $k$ 
parties have respective inputs $x^1, \ldots, x^k$, and let $x = \sum_{i=1}^k x^i$. The parties 
use the shared randomness to choose shared vectors $v^1, \ldots, v^r$ as described above. 
For $i = 1, \ldots, k$ and $j = 1, \ldots, r$, let $y^i_j = \langle v^j, x^i \rangle/G_p^{1/p}$, so that 
$y^i = (y^i_1, \ldots, y^i_r)$. Let $y = \sum_{i=1}^k y^i$. By the above, $\|y\|_p^p = (1 \pm \eps/3)\|x\|_2^p$ 
with probability at least $1-c$ for an arbitrarily small constant $c > 0$. We note that the entries
of the $v^i$ can be discretized to $O(\log n)$ bits, changing the $p$-norm of $y$ by only a $(1 \pm O(1/n))$ factor,
which we ignore. 

Hence, given a randomized protocol for estimating $\|y\|_p^p$ up to a $(1+\eps/3)$ factor with probability 
$1-\delta$, and given that the 
parties have respective inputs $y^1, \ldots, y^k$, this implies a randomized protocol for estimating $\|x\|_2^p$ 
up to a $(1 \pm \eps/3) \cdot (1 \pm \eps/3) = (1 \pm \eps)$ factor with probability at least $1-\delta-c$, 
and hence a protocol for estimating $\ell_2$ up to a $(1 \pm \eps)$ factor with this probability. The communication 
complexity of the protocol for $\ell_2$ is the same as that for $\ell_p$. By our communication 
lower bound for estimating $\ell_2$ (in fact, for estimating $F_2$ in which all coordinates of $x$ 
are non-negative), 
this implies the following theorem. 

\begin{theorem} 
The randomized communication complexity of approximating the $\ell_p$-norm, $p \geq 1$, up to a factor of 
$1+\eps$ with constant probability, is ${\Omega}(k/\eps^2)$. 
\end{theorem}





\subsection*{Acknowledgements}
We would like to thank Elad Verbin for many helpful discussions, in particular, for helping us with the $F_0$ lower bound, which was discovered in joint conversations with him. We also thank Amit Chakrabarti and Oded Regev for helpful discussions, as well as the anonymous referees for useful comments. Finally, we thank the organizers of the Synergies in Lower Bounds workshop that took place in Aarhus for bringing the authors together. 

\bibliographystyle{abbrv}
\bibliography{paper}

\iffull{
\appendix
\section{Proof for Observation~\ref{obs:first}}
\label{sec:independent}
We first show the rectangle property of private randomness protocols in the message passing model.
The proof is just a syntactic change to that in \cite{Bar-Yossef:02}, Section 6.4.1, which was designed for the blackboard model. The only difference between the blackboard model and the message-passing model is that in the blackboard model, if one speaks, everyone else can hear.

\begin{property}
\label{prop:rectangle}
Given a $k$-party private randomness protocol $\Pi$ on inputs in $\mathcal{Z}_1 \times \cdots \times \mathcal{Z}_k$ in the message passing model, for all $\{z_1, \ldots, z_k\} \in \mathcal{Z}_1 \times \cdots \times \mathcal{Z}_k$, and for all possible transcripts $\pi \in \{0,1\}^*$, we have
\begin{equation}
\label{eq:x-1}
\Pr_R [\Pi = \pi\ |\ Z_1 = z_1, \ldots, Z_k = z_k] = \prod_{i=1}^k \Pr_{R_i}[\Pi_i = \pi_i\ |\ Z_i = z_i],
\end{equation}
where $\pi_i$ is the part of transcript $\pi$ that player $i$ sees (that is, the concatenation of all messages sent from or received by $S_i$), and $R = \{R_1, \ldots, R_k\} \in \{\mathcal{R}_1 \times \ldots \times \mathcal{R}_k\}$ is the players' private randomness. Furthermore, we have
\begin{equation}
\label{eq:x-2}
\Pr_{\mu, R}[\Pi = \pi] = \prod_{i \in [k]} \Pr_{\mu, R_i}[\Pi_i = \pi_i].
\end{equation}
\end{property}

\begin{proof}
We can view the input of $P_i\ (i \in [k])$ as a pair $(z_i, r_i)$, where $z_i \in \mathcal{Z}_i$ and $r_i \in \mathcal{R}_i$. Let $\C_1(\pi_1) \times \cdots \times \C_k(\pi_k)$ be the combinatorial rectangle containing all tuples $\{(z_1, r_1), \ldots, (z_k, r_k)\}$ such that $\Pi(z_1, \ldots, z_k, r_1, \ldots, r_k) = \pi$. 

For each $i \in [k]$, for each $z_i \in \mathcal{Z}_i$, let $\C_i(z_i, \pi_i)$ be the projection of $\C_i(\pi_i)$ on pairs of the form $(z_i, *)$, and $\C_i(z_i)$ be the collection of all pairs of the form $(z_i, *)$. Note that $\abs{\C_i(z_i, \pi_i)}/\abs{\C_i(z_i)} = \Pr_{R_i}[\Pi_i = \pi_i\ |\ Z_i = z_i]$. If each player $P_i$ chooses $r_i$ uniformly at random from $\mathcal{R}_i$, then the transcript will be $\pi$ if and only if $(z_i, r_i) \in \C_i(z_i, \pi_i)$. Since the choices of $R_1, \ldots, R_k$ are all independent, it follows that $\Pr_R[\Pi = \pi\ |\ Z_1 = z_1, \ldots, Z_k = z_k] = \prod_{i \in [k]} \Pr_{R_i}[\Pi_i = \pi_i\ |\ Z_i = z_i]$.

To show (\ref{eq:x-2}), we sum over all possible values of $Z = \{Z_1, \ldots, Z_k\}$. Let $Z_{-i} = \{Z_1, \ldots, Z_{i-1}, Z_{i+1}, Z_k\}$.
\begin{eqnarray*}
&&\Pr_{\mu, R}[\Pi = \pi] \\
& = & \sum_{z} (\Pr_R[\Pi = \pi\ |\ Z = z] \Pr_\mu [Z = z]) \\
& = & \sum_{z} \left( \prod_{i \in [k]} \Pr_{R_i}[\Pi_i = \pi_i \ |\ Z_i = z_i] \prod_{i \in [k]}\Pr_\mu[Z_i = z_i] \right) \quad (\text{by (\ref{eq:x-1}) and $Z_i$ are independent})\\
& = & \sum_{z} \prod_{i \in [k]}  \Pr_{\mu, R_i}[\Pi_i = \pi_i \wedge Z_i = z_i]\\
& = & \sum_{z_{-i}} \left(\left(\sum_{z_i} \Pr_{\mu, R_i}[\Pi_i = \pi_i \wedge Z_i = z_i] \right) \prod_{j \neq i} \Pr_{\mu, R_j}[\Pi_j = \pi_j \wedge Z_j = z_j] \right) \\
& = &\Pr_{\mu, R_i}[\Pi_i = \pi_i] \cdot \sum_{z_{-i}} \prod_{j \neq i} \Pr_{\mu, R_j}[\Pi_j = \pi_j \wedge Z_j = z_j] \\
& = & \prod_{i \in [k]} \Pr_{\mu, R_i} [\Pi_i = \pi_i].
\end{eqnarray*}
\end{proof}

Now we prove Observation~\ref{obs:first}. 
That is, to show
\begin{equation*}
\Pr_{\mu, R} [Z_1 = z_1, \ldots, Z_k = z_k\ |\ \Pi = \pi] = \prod_{i \in [k]} \Pr_{\mu, R_i}[Z_i = z_i\ |\ \Pi_i = \pi_i],
\end{equation*}

We show this using the rectangle property.  
\begin{eqnarray*}
&&\Pr_{\mu, R} [Z_1 = z_1, \ldots, Z_k = z_k\ |\ \Pi = \pi] \\
&=& \frac{\Pr_R[\Pi = \pi\ |\ Z_1 = z_1, \ldots, Z_k = z_k] \cdot \Pr_\mu [Z_1 = z_1, \ldots, Z_k = z_k]}{\Pr_{\mu, R}[\Pi = \pi]} \quad (\text{Bayes' theorem}) \\
&=& \frac{\prod_{i \in [k]} \Pr_{R_i}[\Pi_i = \pi_i\ |\ Z_i = z_i] \cdot \prod_{i \in [k]} \Pr_\mu[Z_i = z_i]}{\prod_{i \in [k]} \Pr_{\mu, R_i}[\Pi_i = \pi_i]} \quad (\text{by } (\ref{eq:x-1}), (\ref{eq:x-2}) \text{ and } Z_i \text{ are independent}) \\
&=& \prod_{i \in [k]} \frac{\Pr_{R_i}[\Pi_i = \pi_i\ |\ Z_i = z_i] \cdot \Pr_{\mu}[Z_i = z_i]}{\Pr_{\mu, R_i}[\Pi_i = \pi_i]} \\
&=& \prod_{i \in [k]} \Pr_{\mu, R_i}[Z_i = z_i\ |\ \Pi_i = \pi_i]  \quad (\text{Bayes' theorem})
\end{eqnarray*}

}

\end{document}